\documentclass{article}
\usepackage{amsfonts}
\usepackage{amsmath}
\usepackage{amsthm}
\usepackage[margin=\parindent]{caption}
\usepackage{dsfont}
\usepackage{enumitem}
\usepackage[T1]{fontenc}
\usepackage[osf]{sourcesanspro}
\usepackage[bitstream-charter,greekuppercase=upright,cal=cmcal]{mathdesign}
\usepackage{mathtools}
\usepackage{svg}
\usepackage{subcaption}
\usepackage{xcolor}
\usepackage[most]{tcolorbox}
\usepackage{placeins}
\usepackage{stmaryrd}
\usepackage{pgfplots}
\usepackage{float}
\usepackage{makecell}
\pgfplotsset{compat=1.18}
\usepackage[normalem]{ulem}
\usepackage{tikz}
\usetikzlibrary{arrows.meta}

\usepackage{algorithm}
\usepackage{algpseudocode}
\algrenewcommand\algorithmicrequire{\textbf{Input:}}
\algnewcommand{\algorithmicassumption}{\textbf{Assumption:}}
\algnewcommand{\Assumption}{\item[\algorithmicassumption]}
\algrenewcommand\algorithmicensure{\textbf{Output:}}

\usepackage{tikz}
\usetikzlibrary{backgrounds,fit,positioning,calc,decorations.markings,decorations.pathreplacing,calligraphy}
\usepackage[paper=a4paper,margin=3cm]{geometry}

\usepackage{sectsty}
\allsectionsfont{\fontseries{sb}\sffamily}
\definecolor{commentgreen}{RGB}{0,110,65}

\usepackage[pdfusetitle]{hyperref}
\usepackage[affil-it]{authblk}
\definecolor{dullmagenta}{rgb}{0.5,0,0.4}
\definecolor{dullblue}{rgb}{0.0,0,0.86}
\hypersetup{colorlinks=true,citecolor=dullblue,linkcolor=dullblue,filecolor=dullblue,urlcolor=dullblue,breaklinks=true,bookmarksnumbered=true, bookmarksopen=true,bookmarksopenlevel=1}
\usepackage{bookmark}
\usepackage{zref-clever}
\zcRefTypeSetup{problem}{
    name-sg = problem,
    Name-sg = Problem,
    name-pl = problems,
    Name-pl = Problems,
}

\newcommand{\cref}[1]{\zcref[S]{#1}}
\newcommand{\Cref}[1]{\zcref[S]{#1}}

\usepackage{amsthm} 
\usepackage{hyperref} 
\newtheorem{theorem}{Theorem}[section]
\newtheorem{lemma}{Lemma}[section]
\newtheorem{proposition}{Proposition}[section]
\newtheorem{corollary}{Corollary}[section]

\newtheorem{definition}{Definition}[section]
\newtheorem{problem}{Problem}[section]

\newlist{problemdescr}{description}{1} 
\setlist[problemdescr,1]{font=\textendash\enskip\scshape\bfseries,noitemsep,topsep=0pt}

\tikzset{
  fgfac/.style = {
    shape=rectangle, 
    semithick, 
    fill=white,
    text=black, 
    rounded corners=0.5pt, 
    draw=black, 
    inner sep=0.5mm, 
    outer sep=0.0mm,
    minimum height=4mm, 
    minimum width=4mm,
    text height=1.5ex,
    text depth=.25ex
  },
  fgvar/.style = {
    shape=circle, 
    semithick, 
    fill=white,
    text=black, 
    draw=black,
    inner sep=0.5mm, 
    outer sep=0.0mm,
    minimum height=4mm, 
    minimum width=4mm,
    text height=1.5ex,
    text depth=.25ex
  },
  box/.style = {
    shape=rectangle, 
    semithick, 
    text=black, 
    rounded corners=0.5pt, 
    draw=black, 
    fill=black!5, 
    inner sep=1.5mm, 
    outer sep=0.0mm,
    minimum height=8mm, 
    minimum width=8mm,
    text height=1.5ex,
    text depth=.25ex
  },
  edge/.style = {
    semithick,
    -stealth
  }
 }

\DeclareFontFamily{U}{bbold}{}
\DeclareFontShape{U}{bbold}{m}{n}
 {
  <-5.5> s*[1.069] bbold5
  <5.5-6.5> s*[1.069] bbold6
  <6.5-7.5> s*[1.069] bbold7
  <7.5-8.5> s*[1.069] bbold8
  <8.5-9.5> s*[1.069] bbold9
  <9.5-11> s*[1.069] bbold12 
  <11-15> s*[1.069] bbold12
  <15-> s*[1.069] bbold17
 }{}

\DeclareRobustCommand{\id}{%
  \text{\usefont{U}{bbold}{m}{n}1}%
}

\DeclareMathOperator{\Tr}{Tr}
\DeclareMathOperator{\sech}{sech}
\DeclareMathOperator{\asech}{asech}
\DeclareMathOperator{\Ima}{Im}

\newcommand{\abs}[1]{\left\lvert#1\right\rvert}
\newcommand{\norm}[1]{\left\lVert #1 \right\rVert}
\newcommand{\symp}{\mathrm{Sympl}}
\providecommand{\ket}[1]{\ensuremath{\left\lvert #1 \right\rangle}}

\providecommand{\braket}[2]{\ensuremath{\left\langle\smash{#1}\middle\vert\smash{#2}\right\rangle}}
\providecommand{\braopket}[3]{\ensuremath{\left\langle\smash{#1}\middle\vert\smash{#2}\middle\vert\smash{#3}\right\rangle}}
\providecommand{\ketbra}[2]{\ensuremath{\left\lvert \smash{#1}\middle\rangle\!\middle\langle\smash{#2}\right\rvert}}
\providecommand{\proj}[1]{\ketbra{#1}{#1}}
\DeclareMathOperator{\tr}{Tr}
\DeclareMathOperator\atanh{atanh}

\def\ff{\mathbb F}

\def\cC{\mathcal C}
\def\cD{\mathcal D}
\def\cH{\mathcal H}

\def\cP{\mathcal P}
\def\calC{\mathcal C}
\def\calH{\mathcal H}

\def\psc{\ensuremath{\textsc{PSC}}}
\def\bec{\ensuremath{\textsc{BEC}}}
\def\bsc{\ensuremath{\textsc{BSC}}}

\def\pfail{\ensuremath{p_{\mathrm{fail}}}}
\def\supp{\ensuremath{\mathrm{supp}}}

\newcommand{\vnop}{\varoast}
\newcommand{\cnop}{\boxast}

\title{\fontseries{m} \LARGE\sffamily Gibbs state preparation through quantum decoding}
\author[1]{\sffamily Christophe Piveteau}
\author[2]{\sffamily Yihui Quek}
\author[1,3]{\sffamily M. Isabel Franco Garrido}
\affil[1]{\small IBM Research Europe - Zurich}
\affil[2]{\small École Polytechnique Fédérale de Lausanne}
\affil[3]{\small Institute for Quantum Information and Matter, California Institute of Technology}
\date{}

\begin{document}
\maketitle

\begin{abstract}
Recent approaches to quantum Gibbs-state preparation have largely sought to mimic thermalization. Here, we instead exploit algebraic structure in the Hamiltonian for Gibbs sampling. Our starting point is Hamiltonian Decoded Quantum Interferometry (HDQI), a recent application of Regev's reduction that reduces Gibbs-state preparation for Pauli Hamiltonians to
decoding error-correcting codes. We introduce quantum decoding HDQI (QD-HDQI),
replacing HDQI's coherent classical decoder by a
\emph{quantum decoder}: a quantum algorithm that acts directly on
superpositions of noisy codewords and can exploit interference between different errors. This reveals a channel-coding picture of Gibbs state preparation that immediately improves the temperatures accessible to HDQI.

For commuting Pauli Hamiltonians, our reduction maps temperature directly onto decoding noise: preparing colder states requires a decoder that tolerates more noise. For uniformly random Hamiltonian signs, the vanishing-error threshold of an efficient decoder directly determines the lowest temperature that our algorithm can reach. Even using only classical decoders, this improves upon HDQI's achievable inverse temperature from linear to square-root in the correctable error fraction; quantum decoders can improve further by tolerating still stronger noise. A similar picture extends to the noncommuting regime previously handled by HDQI, where now the noise is correlated within small, independent blocks. 


We study two such quantum decoders in our QD-HDQI Gibbs sampling framework. The first, based on unambiguous state discrimination, works efficiently for arbitrary binary linear codes. However, we ``dequantize" this strategy for commuting Hamiltonians, contributing a classical rejection sampler that matches its temperature threshold. The second is belief propagation with quantum messages (BPQM), a quantum generalization of classical belief propagation. Using a novel analysis of density evolution for correlated noise, we show that for certain dilute spin glasses, BPQM's decoding threshold exceeds the Shannon limit for classical decoding, reaching temperatures inaccessible to any classical decoder within our framework.
This advantage persists after adding a sufficiently small constant density of noncommuting defects to the Hamiltonian. 

Our results thus identify quantum decoding as a new route to Gibbs-state preparation at lower temperatures, with future improvements tied directly to the development of better quantum decoders.
\end{abstract}

\newpage
\tableofcontents

\newpage
\section{Introduction}
Preparing a Gibbs state is usually approached by engineering quantum dynamics that relax to the desired thermal state. Here we explore a different perspective, approaching Gibbs sampling as a problem in error-correcting codes. This draws from a new quantum algorithmic paradigm in which computational problems are reduced to decoding tasks over associated error-correcting codes \cite{thomas_reduction_2024,chen_quantumalgo_2022,yamakawa_quantumadvantage_2024,jordan_dqi_2025,chailloux_quantumadvanteage_2025,garrido2026regevsreductioncandidatequantum}. 
This line of work can be traced back to Regev's foundational result establishing a reduction from certain worst-case lattice problems to learning with errors~\cite{regev_reduction_2009}. Decoded Quantum Interferometry (DQI)\cite{jordan_dqi_2025} applies this idea to optimization, while Hamiltonian Decoded Quantum Interferometry (HDQI)\cite{schmidhuber2025hamiltoniandecodedquantuminterferometry} extends it to Gibbs state and ground state preparation for Pauli Hamiltonians. HDQI, however, reduces these state preparation problems to the problem of decoding a classical error-correcting code related to the Hamiltonian, and addresses this by coherently applying classical decoders. It also first approximates the Gibbs function by a low-degree polynomial, whose degree cannot exceed the maximum error weight the decoder can correct.

For other algorithms based on Regev’s reduction, it has been observed that the decoding step can sometimes be performed more effectively using a \emph{quantum decoder} rather than a classical one~\cite{chen_quantumalgo_2022,chailloux_quantumdecoding_2023, garrido2026regevsreductioncandidatequantum}. A quantum decoder is a quantum algorithm that acts directly on a superposition of noisy
codewords and can exploit interference between different errors.
This raises a natural question for Gibbs-state preparation:
\begin{center}
\textit{Can this quantum decoding advantage be converted into Gibbs samplers at lower temperatures?}
\end{center}

In this work, we extend the HDQI framework beyond classical decoding and show how to incorporate quantum decoders into the state-preparation procedure.
We consider general Hamiltonians of the form
\begin{equation}\label{eq:hamiltonian}
    \cH = \sum_{i=1}^m (-1)^{y_i}v_i P_i
    \, .
\end{equation}
where $v_i\in [0,1],y_i \in \mathbb{F}_2$ denote the magnitude and sign of the $i-$th coupling.
For some target real-valued function $F$, we propose an algorithm called \emph{quantum decoding HDQI} (QD-HDQI) which reduces the problem of preparing filters of the Hamiltonian
\begin{equation}
    \rho_{\mathrm{QD-HDQI}} = \frac{F^2(\cH)}
    {\Tr[F^2(\cH)]}
\end{equation}
to a certain \emph{quantum} decoding problem over a code associated with the Hamiltonian (its {\em symplectic code}, introduced later) with a noise function $f$ that is determined by $F$. By contrast, $F$ in HDQI is restricted to be a polynomial.

More broadly, our reformulation suggests a different way to think about HDQI. Rather than going through the lossy step of an intermediate polynomial approximation, the target function $F$ now directly defines the noise model that will be applied to the Hamiltonian's symplectic code, which must be decoded. For Gibbs-state preparation, this viewpoint has a particularly appealing interpretation: lowering the temperature corresponds directly to increasing the noise faced by the decoder.

This perspective raises two questions that organize the rest of our work:
\begin{enumerate}
    \item When does the noise model induced by the Gibbs function retain enough structure for QD-HDQI to yield an efficient Gibbs sampler?
    \item For what Hamiltonians do quantum decoders beat classical decoders within our algorithm for Gibbs sampling?
\end{enumerate}

We find unexpectedly clean answers to both questions. For commuting Hamiltonians, we prove that the reduction implement by QD-HDQI is efficient, and the resulting decoding problem is that of transmitting the Hamiltonian’s symplectic code through independent binary pure-state channels, with the strength of the channel noise explicitly determined by the temperature. 
This setting naturally allows us to leverage techniques from classical-quantum channel coding to solve the quantum decoding problem efficiently.
Even if we limit our algorithm to use classical decoders, the achievable temperature improves upon HDQI's achievable inverse temperature from linear to square-root in the correctable error fraction; quantum decoders can improve further by tolerating still stronger noise. 
The construction and its efficiency extend to non-commuting Hamiltonians whose anti-commutation graph has logarithmic-size connected components. In this setting, the induced channel decomposes into small blocks, within which the noise may exhibit correlations\footnote{Here it is less obvious how to directly compare our achievable temperatures to those of HDQI, as those temperatures depend on the exact topology of the clusters in the Hamiltonian's anticommutation graph.}

The second question has a more subtle answer: pairing our algorithm with an unambiguous state discrimination (USD) measurement followed by Gaussian elimination decoding of the resulting erasure channel allows us to efficiently decode any code family, even completely structureless random linear codes. Yet we construct a classical Gibbs-sampling algorithm whose performance, for {\em any} commuting Hamiltonian family, exactly matches that of QD-HDQI with the USD-based quantum decoder. On the positive side, a different quantum decoder, belief propagation with quantum messages, yields a genuine advantage over classical decoders for a noncommuting Hamiltonian family, sparse spin-glass Hamiltonians with defects. 

\subsection{State preparation using classical syndrome decoding}\label{sec:state_prep_with_syndrome_decoding}
Before presenting our contributions, and specifically the general QD-HDQI algorithm, it is useful to describe a prototypical HDQI-like algorithm to reduce state preparation to \emph{classical syndrome decoding} (as opposed to quantum decoding).
This algorithm can essentially be viewed as a reformulation of the HDQI algorithm, but in somewhat more general language.
It will serve to motivate and explain the key technical improvements that lead to our QD-HDQI algorithm.

Consider a fixed Hamiltonian $\cH$ of the form \cref{eq:hamiltonian}, and let's define the product Pauli operator
\begin{equation}
    Q_e\coloneqq P_1^{e_1}P_2^{e_2}\cdots P_m^{e_m}
\end{equation}
Notice in particular that the symplectic representation of $Q_e$ is $\symp(Q_e)=He$, where $H$ is the $2n\times m$ parity check matrix whose columns are the symplectic representations of the Paulis in $\calH$:
\begin{equation}\label{eq:symplectic_code_matrix}
    \calC_{\calH}\coloneqq\ker(H)
    \, , \qquad
    H = \begin{pmatrix}
        \vert & \vert & & \vert \\
        \symp(P_1) & \symp(P_2) & \cdots & \symp(P_m) \\
        \vert & \vert & & \vert
    \end{pmatrix}
    \, .
\end{equation}
The code $\calC_{\calH}$ with parity check matrix $H$ is called the \emph{symplectic code} $\calC_{\calH}$ associated to $\cH$.

Consider a normalized function $f:\ff_2^m\rightarrow\mathbb{C}$ where $\norm{f}_2=1$.
The task is to prepare the state $\rho_f$ which is defined as follows.
\begin{equation}\label{eq:def_rhof}
    \rho_f \coloneqq \frac{1}{2^n} A_fA_f^\dagger
    \qquad\text{where}\qquad
    A_f \coloneqq \sum_ef(e)Q_e
    \, .
\end{equation}
For a specific choice of $f$, the following three-step procedure reduces the problem of preparing $\rho_f$ to a syndrome decoding problem over the symplectic code.
\begin{enumerate}
    \item \textbf{(Pilot state preparation)} Prepare the $m$-qubit state
    \begin{equation}
        \sum_{e\in\ff_2^m} f(e) \ket{e}
    \end{equation}
    as well as a maximally-entangled $2n$-qubit state $\ket{\Psi^n} = 2^{-n/2}\sum_{x\in\ff_2^n}\ket{x}\ket{x}$.

    \item \textbf{(Syndrome computation)} Controlled on the pilot state, apply the Pauli operator $Q_e$ on the first half of $\ket{\Psi^n}$, obtaining
    \begin{equation}
        \sum_{e\in\ff_2^m} f(e)\ket{e} \otimes (Q_e\otimes\mathds{1})\ket{\Psi^n}
        \, .
    \end{equation}
    Recall that for the Pauli strings $Q$, the states $(Q\otimes\mathds{1})\ket{\Psi^n}$ form an orthonormal basis.
    Therefore, there exists a unitary map $U_{\mathrm{sympl}}$ acting as $(Q_e\otimes\mathds{1})\ket{\Psi^n}\mapsto \eta(Q_e)\ket{He}$ for all $e\in\ff_2^m$, where $\eta(Q_e)$ is some phase term.
    Applying $U_{\mathrm{sympl}}$, we obtain the state
    \begin{equation}
        \sum_{e\in\ff_2^m} f(e)\eta(Q_e) \ket{e} \otimes \ket{He}
        \, .
    \end{equation}

    \item \textbf{(Uncomputation through syndrome decoding)}
    In order to uncompute the first register, coherently run a classical syndrome decoder that can recover $e$ from $He$ for any $e\in\supp(f)$.
    \begin{equation}
        \sum_{e\in\ff_2^m} f(e)\eta(Q_e) \ket{He}
        \, .
    \end{equation}
    After applying the unitary $U_{\mathrm{sympl}}^\dagger$, we obtain
    \begin{equation}\label{eq:hdqi_final_state}
        \left( \left( \sum_{e\in\ff_2^m} f(e)Q_e \right) \otimes \mathds{1} \right) \ket{\Psi^n}
        = (A_f\otimes\mathds{1})\ket{\Psi^n}
    \end{equation}
\end{enumerate}
The final state in \cref{eq:hdqi_final_state} is precisely the purification of the desired state, i.e., tracing out the second register yields $\rho_f$.

There are two important observations we can make about the above algorithm.
First, it is efficient only if the preparation of the pilot state $\sum_ef(e)\ket{e}$ and the decoding step $He\mapsto e$ are efficiently implementable.
Second, the uncomputation step can only succeed if the map $e\mapsto He$ is injective on $\supp(f)$.

The HDQI algorithm in \cite{schmidhuber2025hamiltoniandecodedquantuminterferometry} can be understood as a specialization of this procedure, with a specific choice of $f$, that fulfills the above conditions.
Concretely, for a given polynomial $\cP$, the authors construct an explicit noise function $f$ such that $\rho_f=\rho_{\mathrm{HDQI}}$ and the associated pilot state can be prepared in polynomial time.
Moreover, the authors' construction ensures that the maximum weight of errors $e\in \supp(f)$ is at most $\deg(\cP)$.
That means that the maximum error weight tolerated by the decoder directly determines the maximum achievable degree of the polynomial $\cP$.

\subsection{Overview of contributions}\label{sec:contributions}
Our first contribution is to reformulate HDQI so that its uncomputation
step can use a quantum codeword decoder rather than a classical syndrome
decoder. We first explain this change of viewpoint, before examining its
consequences for Gibbs-state preparation.

\subsubsection{From syndrome decoding to codeword decoding} \label{subsec:qdecoder}

HDQI is formulated in what we call the ``syndrome picture".
As seen in \cref{sec:state_prep_with_syndrome_decoding}, the algorithm prepares a state of a form that is a superposition over syndromes $\sum_{e\in\ff_2^m}g(e)\ket{He}$, where $g(e)\coloneqq f(e)\eta(Q_e)$, which is unitarily equivalent to the purification of the desired target state $\rho_f$.
The corresponding uncomputation step (step 3 in the previous section) runs a classical syndrome decoder on each branch of the superposition to achieve
\begin{equation}
    \sum_e g(e) \ket{e}\ket{He} \rightarrow \sum_e g(e) \ket{\mathbf 0}\ket{He}.
    \qquad
    \left(\parbox{3cm}{\centering\footnotesize uncomputation through syndrome decoding}\right)
\end{equation}
On the other hand, quantum decoders naturally work in the ``codeword picture", taking as input a noisy codeword
\begin{equation}\label{eq:input_qdecoder}
    \ket{\Psi_c} :=\sum_{e\in\ff_2^m}g(e)\ket{c+e}
\end{equation}
and inferring the underlying codeword $c\in \cC$.
\begin{definition}[The Quantum Decoding Problem (informal) \cite{chailloux_quantumdecoding_2023}]
Consider a known classical code $\cC\subset\ff_2^m$ and function $g:\ff_2^m\rightarrow\mathbb{C}$, $\norm{g}_2=1$.
The task is to recover $c\in\cC$ given access to a noisy codeword state $\ket{\Psi_c}$.
\end{definition}
Indeed, a quantum decoder is nothing but a measurement that performs state discrimination on the possible noisy-codeword
states $\{\ket{\Psi_c}\}_{c\in C}$.
One possible strategy is to measure in the computational basis,
obtaining $c+e$ with probability $|g(e)|^2$, and then apply a classical
decoder. A quantum decoder need not make this intermediate measurement:
it can exploit interference between different error components to
distinguish the codewords more reliably.

In general, quantum decoders can be used to perform a different uncomputation step of the form
\begin{equation}
\begin{array}{ccc} 
    \frac{1}{\sqrt{\abs{\cC}}}\sum_{c\in\cC}\ket{c}\ket{\Psi_c} & \rightarrow & \frac{1}{\sqrt{\abs{\cC}}}\sum_{c\in\cC}\ket{\mathbf 0}\ket{\Psi_c} \\ \rotatebox{90}{$=$} & & \rotatebox{90}{$=$} \\
   \frac{1}{\sqrt{|\cC|}} \sum\limits_{\substack{c\in\cC\\ e\in\ff_2^m}} g(e) \ket{c}\ket{c+e} & & \frac{1}{\sqrt{|\cC|}}\sum\limits_{\substack{c\in\cC\\ e\in\ff_2^m}} g(e) \ket{\mathbf 0}\ket{c+e}
    \end{array}
    \, .
    \qquad
    \left(\parbox{3cm}{\centering\footnotesize uncomputation through codeword decoding}\right)
\end{equation}

To use quantum decoders for Gibbs state preparation, we must first build the bridge between these two pictures by formulating a transformation between the two involved states.
In \cref{lem:codeword_to_syndrome}, we construct an efficiently implementable unitary $U_{\text{transform}}$ such that
\begin{equation}\label{eq:Utransform}
\frac{1}{\sqrt{\abs{\cC}}}
\sum_{\substack{c\in\cC\\ e\in\ff_2^m}}
g(e)\ket{c+e}
\mathrel{
\begin{gathered}
\xrightarrow{\quad U_{\text{transform}}\quad}\\[-0.7ex]
\xleftarrow{\quad U_{\text{transform}}^\dagger\quad}
\end{gathered}
}
\sum_{e\in\ff_2^m} g(e)\ket{He}.
\end{equation}
This opens up the following strategy: first, use quantum decoders to prepare the ``codeword" state on the left-hand-side of \Cref{eq:Utransform} and then transform back to the ``syndrome" state which is the form required for Gibbs state preparation.
That is:

\begin{tcolorbox}[
    colback=blue!5,
    colframe=blue!70!black,
    boxrule=1.5pt,
    arc=4pt,
    left=3pt,
    right=3pt,
    top=2pt,
    bottom=2pt
]
\begin{equation}\label{eq:codeword_skeleton_intro}
\begin{aligned}
&
\frac{1}{\sqrt{\abs{\cC}}}\sum_{c\in\cC}
\ket{c}\left(\sum_{e} g(e)\ket{c+e}\right)
\xrightarrow{\mathrm{Decode}}
\frac{1}{\sqrt{\abs{\cC}}}
\sum_{c \in \cC,e}
g(e)\ket{c+e}
\xrightarrow{U_\text{transform}}
\sum_{e} g(e)\ket{He}.
\end{aligned}
\end{equation}
\end{tcolorbox}
Notably, the initial state on the left now precisely matches the intermediate state used in Regev's reduction (formulated with quantum decoders by \cite{chailloux_quantumdecoding_2023}), before the quantum Fourier transform.
It can be prepared by first separately preparing  a superposition over errors and a superposition over codewords and then applying transversal CNOT gates:
\begin{equation}
    \left(\frac{1}{\sqrt{\abs{\cC}}} \sum_{c\in \calC} \ket{c} \right)\otimes \left(\sum_{e} g(e) \ket{e}\right) \rightarrow \frac{1}{\sqrt{\abs{\cC}}}\sum_{c\in\cC}
    \sum_{e} g(e)\ket{c}\ket{c+e}.
\end{equation}
Preparing the superposition over codewords is always efficient because the map from messages to codewords is linear and injective, and thus can be made reversible. 
In contrast, the noise state $\sum_e g(e)\ket{e}$ might not be efficiently preparable.
The efficiency must be considered on a case-by-case basis.

To summarize the above considerations, we have shown that the preparation of the state $\rho_f$ can be reduced to the quantum decoding problem w.r.t. the noise function $g$. 
Our algorithm, which will realize this reduction, is called \emph{quantum decoding HDQI} (QD-HDQI).
QD-HDQI is efficient only if the noise state $\sum_eg(e)\ket{e}$ can itself be prepared efficiently.
In the next sections, we will see that for certain classes of Hamiltonians this can be guaranteed, for example if the terms in the Hamiltonians all commute, or if they anti-commute only in a restricted fashion.
To obtain a practical state preparation procedure for $\rho_f$, one also separately requires a quantum decoding algorithm that is efficient.

\subsubsection{Gibbs state preparation for commuting Hamiltonians}
The advantage of quantum decoders is particularly nicely illustrated in the context of \emph{commuting Hamiltonians}, where the terms $P_i$ all mutually commute.
The key insight is that the Gibbs state factorizes into a simple product form:
\begin{align}
    \exp\left(-\frac{\beta}{2}\cH\right)
    &= \prod_{i=1}^m
    \exp\left(-\frac{\beta}{2}(-1)^{y_i}v_i P_i\right) \\
    &= \mathrm{const.}
    \prod_{i=1}^m
    \left(
        \sqrt{1-p_i}\,\mathds{1}
        -
        (-1)^{y_i}\sqrt{p_i}\,P_i
    \right)\\
    & = \mathrm{const.}\sum_{e\in \mathbb{F}_2^m} f(e) Q_e,
\end{align}
where
\begin{equation}\label{eq:def_tilde_f_intro}
    f(e) \coloneqq \frac{1}{\sqrt{Z_\beta'}} (-1)^{(y+\vec 1)\cdot e} \prod_{i=1}^m \sqrt{(1-p_i)^{(1-e_i)}p_i^{e_i}}
    \, ,
\end{equation}
$p_i\coloneqq \frac{1}{2}\left(1 - \sech(\beta v_i)\right)$ and $Z_\beta'$ is some constant.
Therefore, running QD-HDQI with the noise function $f$ precisely prepares the Gibbs state
\begin{equation}
    \rho_{f} = \frac{1}{\tr[\exp(-\beta\cH)]}\exp\left(-\beta\cH\right)
    \, .
\end{equation}
Notice that the function $f$ defined in \cref{eq:def_tilde_f_intro} is not necessarily well-normalized, given that $Z_\beta'$ might not be equal to $1$.
In \cref{lem:average_normalization}, we show that if we average over the signs $y$ in the Hamiltonian, $Z_\beta'$ remains very close to 1, at least when the associated quantum decoding problem can be solved with small failure probability $\pfail$
\begin{equation}
    \mathbb{E}_{Y\sim\mathrm{Uniform}(\ff_2^m)}\left[ \abs{Z_\beta'-1} \right]
    \leq 2\sqrt{\pfail}
    \, .
\end{equation}
Therefore, we can treat $Z_{\beta}'$ to be essentially $1$.

One advantage of the product noise in \cref{eq:def_tilde_f_intro} is that the associated pilot state can be prepared efficiently: Indeed, we show that it is essentially just a product state, which can be prepared in constant depth.

Another important consequence of the noise in \cref{eq:def_tilde_f_intro} is that the quantum decoding problem involved in QD-HDQI can be understood as the decoding problem of transmitting the symplectic code over a sequence of independent binary-input pure-state classical-quantum channels $\psc[p_i]$ of the form
\begin{equation}
    \ff_2 \ni x \mapsto \sqrt{1-p_i}\ket{x} + \sqrt{p_i}\ket{1-x}
    \, .
\end{equation}
In particular, the noise strength $p_i$ of the pure-state channel is directly determined by the desired inverse temperature of the Gibbs state.
In the setting where the couplings $v_i$ are uniform, the probabilities are all equal $p_i\equiv p$.
This allows us to establish a particularly elegant connection between the achievable temperatures and the threshold of the decoder, i.e., the largest value $p^*$ such that $\forall p<p^*$ the decoding failure probability is asymptotically vanishing.
\begin{theorem}[Gibbs sampling for commuting Hamiltonians, informal version of \Cref{cor:hdqi_end_to_end}]
\label{thm:commthm}
Consider a family of commuting Pauli Hamiltonians
\begin{equation}
    \mathcal H_y
    =
    \sum_{i=1}^m (-1)^{y_i} P_i,
    \qquad
    y\in\mathbb F_2^m,
\end{equation}
with associated families of symplectic codes.
Suppose these codes admit polynomial-time classical or quantum codeword decoders with decoding threshold $p^\ast$ over, respectively, the binary symmetric channel or the binary pure-state channel.
Then, for every
\begin{equation}
    \beta
    <
    \beta^\ast
    \coloneqq
    \asech(1-2p^\ast),
\end{equation}
QD-HDQI efficiently prepares the Gibbs state of $\mathcal H_y$ with vanishing trace-distance error with probability $1-o(1)$ over a uniformly random choice $y\in\mathbb F_2^m$.
\end{theorem}

This correspondence makes the advantage of quantum decoding clear: the decoding threshold of any quantum decoder is fundamentally limited by the Holevo capacity of this channel, whereas classical decoders are limited by the Shannon capacity of the corresponding binary symmetric channel.
Consequently, any quantum decoder operating beyond the Shannon limit immediately translates into Gibbs-state preparation at temperatures unattainable with any classical decoder.

Let us briefly compare the result in \cref{thm:commthm} to the achievable temperatures of the HDQI algorithm.
Theorem~3 of \cite{schmidhuber2025hamiltoniandecodedquantuminterferometry}
shows that approximating the Gibbs state to trace distance $\delta$ can be achieved within the generic HDQI construction with a polynomial of degree
$\deg(\cP) \leq 1.12\beta\|\cH\|+O(\log(1/\delta))$.
Suppose the classical decoder used by HDQI is guaranteed to correct every
error of weight at most $\tau m$. 
The same decoder, when applied to a
$BSC[p]$, succeeds with probability $1-o(1)$ for every $p<\tau$, and
therefore has BSC threshold at least $\tau$. For extensive Hamiltonians
with $\|\cH\|=\Theta(m)$, the generic HDQI guarantee therefore gives an
achievable inverse temperature scaling linearly in $\tau$, whereas
QD-HDQI using the same classical decoder achieves inverse temperatures $\beta$ up to
$\asech(1-2\tau)=\Theta(\sqrt{\tau})$ for small $\tau$. This is a quadratic improvement in inverse temperature.
Furthermore, as $\tau$ approaches $1/2$, the inverse temperature $\beta$ achieved by QD-HDQI diverges, but not the one achieved by HDQI.
Quantum decoders can provide a further improvement by attaining a larger
threshold on the corresponding pure-state channel.

Let us now briefly comment on a peculiar property of the QD-HDQI algorithm which strongly differs from HDQI.
Notice that the noise function $f$ has support on \emph{every} error $e\in\ff_2^m$.
This is in stark contrast to the HDQI algorithm, where $\supp(f)$ only contained errors that are information theoretically perfectly decodable.
Therefore, perfect deterministic decoding will no longer be possible.
In the codeword picture, this translates to the statement that the overlap $\langle \Psi_c| \Psi_{c'}\rangle$ is generally nonzero, and therefore, the codeword states may not be perfectly distinguishable.
One of our contributions is thus to analyze QD-HDQI in the setting in which the quantum decoder is allowed to fail with some probability $p_{\mathrm{fail}}$.
We prove a robustness theorem \ref{thm:analysis_codeword} showing that the trace-distance error of the prepared state can be bounded directly in terms of the decoder's failure probability. 
In particular, we highlight that this robustness leads to a strikingly different behavior of QD-HDQI compared to the original HDQI: It can directly prepare Gibbs states of commuting Hamiltonians instead of approximating it with a polynomial filter.

\subsubsection{Extension to noncommuting Hamiltonians.}
For commuting Hamiltonians, the noise state $\sum_eg(e)\ket{e}$ in QD-HDQI, which plays a role analogous to the pilot state in HDQI, has a substantially simpler form than the latter. 
It is therefore natural to ask if writing the state in full generality in our algorithm has bought us anything. After all, 
the efficiency of QD-HDQI for commuting Hamiltonians seems to rely on substantial structure availed by commutativity: the Gibbs coefficient state factorizes into a simple product state, and the corresponding quantum decoding problem reduces to transmitting the symplectic code through independent binary-input pure-state channels. Neither property need hold for a general non-commuting Hamiltonian.

Perhaps surprisingly, we show that our algorithm nevertheless extends beyond the commuting setting. For a non-commuting Hamiltonian, let us group its Pauli terms so that only Paulis within the same group may anticommute:
\begin{equation}\label{eq:minimalnoncommuting}
    \cH
    = \sum_{\gamma=1}^C \cH_{\gamma},
    \qquad
    \cH_{\gamma}
    \coloneqq \sum_{j\in I_\gamma}(-1)^{y_j}v_jP_j,
\end{equation}
where $C\geq 1$ and $m_\gamma\coloneqq |I_\gamma|$. Furthermore, we assume that
\begin{enumerate}
    \item every Pauli in $I_\gamma$ commutes with every Pauli in
    $I_\delta$ whenever $\gamma\neq\delta$.
    \item $\max_\gamma m_\gamma$ is small, so typically $\mathcal O(\log m)$ or even $\mathcal O(1)$.
\end{enumerate}
We note that in the language of \cite{schmidhuber2025hamiltoniandecodedquantuminterferometry}, these conditions correspond to requiring that the connected components of the anti-commutation graph are of logarithmic size. This is precisely the regime in which the original HDQI algorithm was known to be efficient for Gibbs sampling. Ref.~\cite{schmidhuber2025hamiltoniandecodedquantuminterferometry} further exhibited $p$-semiclassical spin glasses as a concrete family satisfying this condition: starting from a classical $Z$-type spin glass, each Hamiltonian term is independently replaced by an $X$-type term with probability $p$, and for sufficiently small constant $p$ the resulting anti-commutation graph has logarithmic-size connected components with high probability.
We show that QD-HDQI remains efficient throughout this regime. 

The intuition for why our algorithm still works in this regime is that the clusters $H_\gamma$ of Pauli terms can in some sense behave similarly to a single Pauli term in the commuting case. Since the different clusters mutually commute, the exponential function factorizes as
\begin{align}
    \exp\left(-\frac{\beta}{2}\cH\right)
    = 
    \prod\limits_{\gamma=1}^C
    \exp\left(-\frac{\beta}{2}\cH_\gamma\right)
    \propto
    \prod\limits_{\gamma=1}^C
    \sum_{e_\gamma\in\ff_2^{m_\gamma}}(-1)^{y_\gamma\cdot e_\gamma}h_\gamma(e_\gamma)Q^{(\gamma)}_{e_\gamma} 
\end{align}
where we partitioned $e=(e_\gamma)_{\gamma}$, $y=(y_\gamma)_{\gamma}$ and $Q_e=\prod_\gamma Q_{e_\gamma}^{(\gamma)}$.
The small size of each block allows its function $h_\gamma$ to be efficiently computed on a classical computer through a Walsh-Hadamard transformation of the blocks $\exp(-\beta\cH_{\gamma}/2)$ (seen as a function of $y_\gamma$).

This block-product form is inherited by the noise state, making it efficiently preparable.
The corresponding decoding problem becomes transmission of the symplectic code through a product of block quantum channels, at least under an appropriately chosen Clifford change of basis (see \cref{prop:block_pure_state_channel} for details).
More concretely, there exist block pure-state channels
\begin{equation}\label{eq:block_noise}
    \mathsf W_\gamma:
    \ff_2^{m_\gamma}\ni x \mapsto \sum_{e_{\gamma}\in\ff_2^{m_\gamma}} h_\gamma'(e_\gamma) \ket{x + e_\gamma}
    \, .
\end{equation}
for some appropriately chosen $h_{\gamma}'$, such that the quantum decoding problem is precisely that of recovering
$c\in\cC$ after transmission through
\begin{equation}\label{eq:block_channel_intro}
    \mathsf W_1\otimes\mathsf W_2\otimes\cdots\otimes\mathsf W_C.
\end{equation}
The channel is independent across blocks, although it need not be identical
from one block to another and it may contain arbitrary correlations within
each block. 

This yields an efficient QD-HDQI Gibbs-state preparation algorithm for non-commuting Hamiltonians in this regime:
\begin{theorem}[QD-HDQI for non-commuting Hamiltonians, informal]\label{thm:analysis_clustered_inf}
    Consider a non-commuting Hamiltonian $\cH$ of the form \Cref{eq:minimalnoncommuting}, with $\max_\gamma m_\gamma = \mathcal O(\log m)$ and some target inverse temperature $\beta\geq 0$. Assume we have a polynomial-time classical or quantum decoding algorithm for the associated decoding problem which fails with sufficiently low probability $\pfail$. QD-HDQI in conjunction with this decoder prepares a state $\tilde{\rho}$ which is close to the Gibbs state $\rho_{\beta}(\calH_Y)$ in expectation over a uniformly random choice of signs in $\cH$:
    \begin{equation}
        \mathbb{E}_{Y\sim\mathrm{Uniform}(\ff_2^m)}
        \left[ \norm{\tilde{\rho} - \rho_{\beta}(\calH_Y)}_1 \right]
        \leq 6\sqrt{\pfail}
        \, .
    \end{equation}
\end{theorem}

\subsubsection{Dequantization of HDQI in conjunction with USD} 
To explore the full power of QD-HDQI, we now consider concretely instantiating the quantum decoders used for the decoding step.
First, we consider using QD-HDQI in conjunction with a quantum decoder based on unambiguous state discrimination (USD) followed by erasure decoding, which is due to Chailloux and Tillich\cite{chailloux_quantumdecoding_2023}.
The first step applies a coherent symbol-wise USD to the noisy codeword, yielding an all-or-nothing outcome for each bit: either the bit is recovered perfectly or all information about it is lost.
The problem of reconstructing the lost bits from the recovered bits can then in turn be viewed as an erasure decoding problem, which can be solved optimally in polynomial time using Gaussian elimination.
In particular, this decoder is efficient for any choice of the underlying symplectic code.
Due to this lack of required structure, QD-HDQI used in conjunction with this USD decoder, can be applied to much broader classes of Hamiltonians which lack an efficient quantum decoder for the underlying symplectic code.

\begin{corollary}[Temperatures achieved by HDQI+USD in terms of BEC threshold, informal]\label{corr:hdqiusdthreshold_intro}
    Consider a family of commuting Hamiltonians \footnote{While it is also possible to apply USD to the non-commuting Hamiltonians we consider, the exact analysis will depend on the precise form of the $h_\gamma$.} where the coefficients are uniformly randomly chosen from $\{+1,-1\}$ (i.e., $v_i=1$ for all $i$ and $y$ is uniformly random).
    Assume that the associated symplectic codes exhibit a threshold $\epsilon^*$ when transmitted over the erasure channel and using Gaussian elimination decoding. QD-HDQI in conjunction with USD can efficiently prepare Gibbs states of temperature $\beta$ with vanishing trace distance error for any $\beta\in [0,\beta^*)$ where
    \begin{equation}
        \beta^* = \atanh(\epsilon^*)
        \, .
    \end{equation}
\end{corollary}

However, we then realize that we can exactly match this performance with a classical simulation of a Gibbs sampler. The classical simulation is possible since we only consider commuting Hamiltonians, whose eigenbasis consists of stabilizer states, which have efficient classical descriptions. Our classical algorithm thus outputs such a description of a state sampled from the ensemble $\rho_{\beta}(\calH)$, for any $\beta$ for which QD-HDQI in conjunction with USD can prepare the Gibbs state.

A related phenomenon was observed by Chailloux and Tillich~\cite{chailloux_quantumdecoding_2023} in the context of the Short Codeword Problem: when their USD-based quantum decoder is used within Regev's reduction, the resulting quantum algorithm finds short codewords only down to the same weight achieved by the classical Prange algorithm.
To our knowledge, our result is the first manifestation of this dequantization phenomenon in the context of Gibbs-state preparation.

\begin{theorem}[Classical dequantization of HDQI with USD, informal version of \Cref{theorem:trace_distance_error} and \Cref{thm:expruntimerejectionsampler}] \label{thm:intro-classical-dequantization} Let \( \mathcal H_y = \sum_{i=1}^m(-1)^{y_i}v_iP_i \) be a commuting \(n\)-qubit Pauli Hamiltonian, and let \(t_i=\tanh(\beta v_i)\). There is a classical rejection sampler that outputs a classical description of a stabilizer state whose average density matrix \(\sigma_{\beta,y}\) satisfies \[ \mathbb E_y \left[ \left\| \sigma_{\beta,y} - \rho_{\beta}(\calH_y) \right\|_1\right] \leq 2p_{\mathrm{dep}}, \] 
where
\begin{equation}p_{\text {dep }}:=\underset{S \leftarrow \operatorname{Ber}(\mathbf{t})}{\operatorname{Pr}}\left[ \operatorname{rank}(H_S)<|S|\right]
\end{equation}
and $H_S$ is the matrix whose columns are the symplectic
representations of the Pauli terms \(\{P_i:i\in S\}\).
Its expected runtime per sample is $O\left(\frac{\text{poly}(m,n)}{1-p_{dep}}\right)$. In particular, the sampler runs in expected polynomial time
whenever $1-p_{\mathrm{dep}}$ is at least inverse-polynomial
in $m$ and $n$.
\end{theorem}
The classical sampler follows from a stabilizer decomposition of the Gibbs state. For a commuting Hamiltonian, we may expand the Gibbs state as a mixture indexed by subsets $S\subseteq[m]$ of Hamiltonian terms, where term $i$ is included independently with probability $t_i=\tanh(\beta v_i)$. Each subset $S$ specifies the simultaneous eigenspace of the corresponding commuting Pauli constraints, and hence a mixed stabilizer state. The only obstruction to sampling $S$ directly from this product Bernoulli distribution is that different subsets generally carry different normalization factors. However, if we condition on the event that the Pauli operators in $S$ are linearly independent (meaning their symplectic representations are linearly independent), this normalization factor is constant. Our algorithm therefore samples $S\sim\mathrm{Ber}(\mathbf t)$ and rejects whenever the corresponding Pauli constraints are linearly dependent. Conditioned on acceptance, it uses Gaussian elimination to sample uniformly from their common eigenspace and outputs a classical description of the resulting stabilizer state.

The connection to QD-HDQI with USD then becomes apparent: applying USD independently to the pure-state-channel outputs converts coordinate $i$ into an erasure with probability $t_i=\tanh(\beta v_i)$. Erasure decoding succeeds precisely when the erased coordinates correspond to linearly independent columns of the symplectic matrix -- the same condition on which our classical rejection sampler accepts. Thus, the rejection probability of the classical sampler is exactly the failure probability of HDQI+USD, explaining why the two algorithms achieve the same temperature threshold.

\subsubsection{Applications: dilute classical spin glass with and without defects} 
We next consider using a different quantum decoder in conjunction with QD-HDQI, the \emph{Belief Propagation with Quantum Messages} (BPQM) decoder.
Classical belief propagation is an iterative message-passing decoding algorithm that forms the backbone of modern coding theory \cite{richardson_modern_2008}.
When paired with an appropriately constructed low-density parity check (LDPC) code, it achieves exceptional performance with an effectively linear computational complexity.
BPQM is a more recently introduced quantum extension of the formalism that aims to translate this success to the quantum setting \cite{renes_belief_2017,rengaswamy_bpqm_2021,piveteau_bpqm_2022,brandsen_bpqm_2022,piveteau_2025_bpqmturbo,mandal_bpqmabelian_2026,mandal_bpqmblock_2026}.
Concretely, it solves the quantum decoding task associated to sending a linear code over independent pure-state classical-quantum channels, and has also been shown to approach the Holevo capacity on sparse graph codes\cite{brandsen_bpqm_2022,piveteau_2025_bpqmturbo}.

To demonstrate the capabilities of QD-HDQI paired with BPQM, we study dilute spin glasses as an exemplary class of Hamiltonians. \cite{schmidhuber2025hamiltoniandecodedquantuminterferometry} introduced this model as a $p$-semiclassical spin glass; its commuting parent model (i.e. when the probability of defects is zero) is the classical Ising $k$-spin glass (equivalently, Max-$k$-XORSAT) studied recently by \cite{zhou_2026_dqigibbs}.
These Hamiltonians are particularly well-suited for BPQM, because their sparse graph structure causes the induced symplectic code to be an LDPC code.
We apply our algorithm first to classical spin glasses, where all Pauli terms are diagonal and therefore commute, and then to spin glasses with random non-commuting defects, which were previously also studied in \cite{schmidhuber2025hamiltoniandecodedquantuminterferometry}.
In both cases, our numerics show that QD-HDQI in conjunction with the BPQM decoder can outperform any classical decoder. 
That is, even if we ignore computational restrictions and pick the information theoretically optimal classical decoder in QD-HDQI, it remains impossible to achieve temperatures as low as BPQM.

\paragraph{Dilute classical spin glass:} These are commuting Hamiltonians of the form
\begin{equation}
    \cH_{A,y}
    \coloneqq
    \sum_{i=1}^m(-1)^{y_i}Z^{a_i}
\end{equation}
where $y\in\ff_2^m$ is uniformly random and $a_i\in\ff_2^{n}$ is the $i$th column of some random sparse matrix $A\in\ff_2^{n\times m}$.
The symplectic code in this case is simply the LDPC code $\cC=\ker(A)$.

We analyze the asymptotic performance of BPQM using \emph{density evolution}, a standard technique in coding theory to analyze the asymptotic behavior of iterative message-passing algorithms by tracking the statistical evolution of messages transmitted over the tree neighborhood of a random codeword bit.
These recursive equations typically exhibit a threshold behavior, where below a certain noise rate, the distribution converges to a point mass distribution of ideal noise-free messages that allow for perfect information recovery.
In particular, we borrow from recent work that provides block error guarantees for BPQM below the density evolution threshold \cite{mandal_bpqmblock_2026}.

In \Cref{tab:classical_spin_glass_results}, we depict the achievable temperatures of QD-HDQI for dilute spin glasses associated to various $(\ell,r)$-regular LDPC code ensembles, based on BPQM thresholds that were numerically estimated by running the density evolution recursion.
We observe that for sparser LDPC ensembles, the estimated BPQM threshold exceeds the Shannon limit for classical decoding, allowing QD-HDQI to prepare Gibbs states at inverse temperatures that no classical decoder can reach.
In our examples, the improvement is as large as roughly \(28\%\) in inverse temperature.
In contrast, for denser LDPC ensembles QD-HDQI with BPQM still significantly outperforms QD-HDQI with classical BP, but the information theoretic advantage disappears.
This can be attributed to the fact that the performance of iterative message passing decoders deteriorates as the average degree of a Tanner graph increases.

\begin{table}[]
    \centering
    \begin{tabular}{c||c|c||c|c}
        $(\ell,r)$ & BP decoding & BPQM decoding & Shannon bound & Holevo bound \\
        \hline
        $(3,4)$ & $0.963$ & $1.327$ & $1.160$ & $1.567$ \\
        $(3,5)$ & $0.749$ & $1.068$ & $0.880$ & $1.225$ \\
        $(3,6)$ & $0.624$ & $0.920$ & $0.735$ & $1.045$ \\
        $(3,7)$ & $0.543$ & $0.818$ & $0.641$ & $0.928$ \\
        $(4,5)$ & $0.910$ & $1.274$ & $1.285$ & $1.719$ \\
        $(4,6)$ & $0.763$ & $1.088$ & $0.992$ & $1.363$ \\
        $(4,7)$ & $0.664$ & $0.965$ & $0.836$ & $1.171$ \\
        $(5,6)$ & $0.845$ & $1.196$ & $1.385$ & $1.840$ \\
        $(5,7)$ & $0.739$ & $1.058$ & $1.083$ & $1.474$ \\
        $(6,7)$ & $0.790$ & $1.125$ & $1.468$ & $1.940$
    \end{tabular}
    \caption{Achievable inverse temperature $\beta$ using our algorithm for $(\ell,r)$-regular spin glasses, with either classical or quantum belief propagation.
    The Shannon and Holevo bound denote the information theoretic limits for the lowest achievable temperatures with classical and quantum decoders respectively (for the given code rate).
    }
    \label{tab:classical_spin_glass_results}
\end{table}

\paragraph{Dilute classical spin glass with defects:} 
These are noncommuting Hamiltonians of the following form:
We start with a classical spin glass associated with a random sparse matrix $A$, as described above.
Then, we flip each Pauli term from $Z$-type to $X$-type with a probability $q$.
This recovers the semiclassical spin-glass family introduced in the original HDQI work, where efficient Gibbs-state preparation was established below an unspecified constant inverse-temperature threshold.
In our work, the achievable inverse temperature is instead explicitly given by the threshold of the quantum decoder used.

The quantum decoding task associated to QD-HDQI here is more intricate than in the classical case.
While the codeword bits belonging to a size-1 component of the anti-commutation graph are still transmitted over a pure state channel $\psc[p]$, the remaining bits are transmitted over block-correlated pure-state channels, as described by \cref{eq:block_channel_intro}.

We again utilize BPQM to solve this associated decoding problem.
It is somewhat tricky to adapt BPQM to the block-independent channel noise.
One strategy would be to first perform a USD on larger blocks, to degrade the effective channel model to block-erasure channels, which in turn can be handled naturally by BPQM.
Unfortunately analyzing the asymptotic performance of this approach analytically is very difficult.
To avoid this difficulty, we instead use a strictly worse decoding strategy, which is however easier to analyze.
We use BPQM, but we ignore every codeword bit that either corresponds to a flipped Pauli term, or which is neighboring to a flipped Pauli term in the Tanner graph.
Concretely, this means the codeword bits that are not ignored are all transmitted over $\psc[p]$, and the remaining codeword bits can be modelled as being transmitted over the full erasure channel $\psc[1/2]$.
For this situation, we can adapt the density evolution by keeping track not only of log-likelihood ratios sent over the tree neighborhood, but also tracking whether the environment of the root node contains flipped bits or not.
By numerically running the density evolution simulation, we obtain an estimate for the threshold $p_{\mathrm{DE}}(q)$ of the BPQM decoder as a function of $q$.
That is, the decoder has an asymptotically vanishing failure probability if $p<p_{\mathrm{DE}}(q)$.

Conversely, we also derive a bound on the maximal noise that can be tolerated by a classical decoder.
For this converse bound, we utilize the analog but opposite approach of the BPQM simplification above: we assume that flipped bits, and neighbors thereof, are perfectly transmitted, which can only help the decoder and reduce the failure probability.
Together with the channel coding strong converse bound, we can derive a quantity $p_{\mathrm{cl}}(q)$ s.t. classical decoders must fail with high probability whenever $p>p_{\mathrm{cl}}(q)$ (see \cref{lem:defected_classical_converse} for details).

We numerically evaluate the two quantities $p_{\mathrm{cl}}(q)$ and $p_{\mathrm{DE}}(q)$ for a $(3,6)$-regular LDPC code $A$ and depict them in \Cref{fig:defected_threshold_comparison_intro}.
We observe that when $q$ is small enough, there is a strict inequality $p_{\mathrm{cl}}(q)<p_{\mathrm{DE}}(q)$, which implies that in that regime, QD-HDQI can achieve temperatures with BPQM that are so low, that no other classical decoder can information theoretically achieve them.

We briefly note that our efficiency guarantee for QD-HDQI additionally requires $q<1/[\ell^2(r-1)^2]$ (so $q\lesssim 0.0045$ for $(\ell,r)=(3,6)$) to ensure that the anti-commutation graph has connected components of size $\mathcal O(\log m)$ with high probability\cite{schmidhuber2025hamiltoniandecodedquantuminterferometry}.
The advantage of quantum decoders is therefore only valid within this regime.

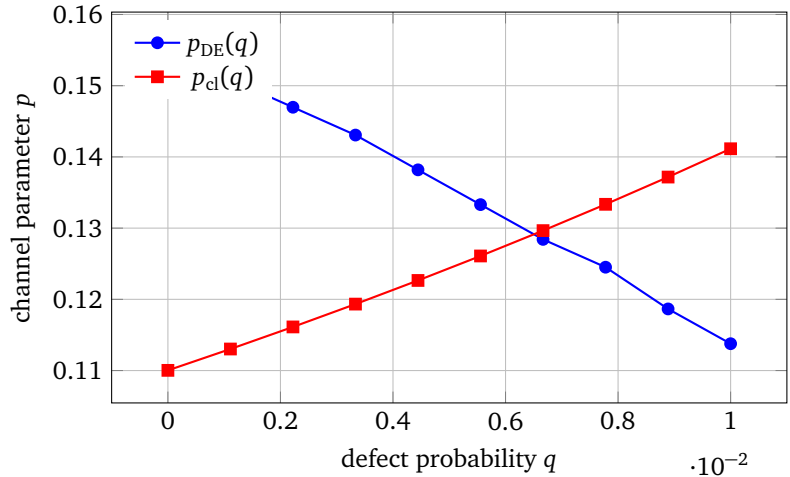
\begin{figure}[h]
\centering
\begin{tikzpicture}
\begin{axis}[
    width=0.70\textwidth,
    height=0.45\textwidth,
    xlabel={defect probability $q$},
    ylabel={channel parameter $p$},
    legend style={
        at={(0.02,0.98)},
        anchor=north west,
        draw=none
    },
    grid=both,
    mark size=2pt,
]
\addplot[
    blue,
    thick,
    mark=*,
]
coordinates {
    (0.0,0.15576171875)
    (0.0011111111111111111,0.15087890625)
    (0.0022222222222222222,0.14697265625)
    (0.0033333333333333335,0.14306640625)
    (0.0044444444444444444,0.13818359375)
    (0.005555555555555556,0.13330078125)
    (0.006666666666666667,0.12841796875)
    (0.0077777777777777776,0.12451171875)
    (0.008888888888888889,0.11865234375)
    (0.01,0.11376953125)
};
\addlegendentry{$p_{\mathrm{DE}}(q)$}

\addplot[
    red,
    thick,
    mark=square*,
]
coordinates {
    (0.0,0.11002786457538605)
    (0.0011111111111111111,0.1130250096321106)
    (0.0022222222222222222,0.11612477898597717)
    (0.0033333333333333335,0.1193319708108902)
    (0.0044444444444444444,0.12265169620513916)
    (0.005555555555555556,0.1260894313454628)
    (0.006666666666666667,0.1296510398387909)
    (0.0077777777777777776,0.13334280997514725)
    (0.008888888888888889,0.13717150688171387)
    (0.01,0.14114444702863693)
};
\addlegendentry{$p_{\mathrm{cl}}(q)$}
\end{axis}
\end{tikzpicture}
\caption{Numerical evaluation of $p_{\mathrm{DE}}(q)$ and $p_{\mathrm{cl}}(q)$ for the $(3,6)$-regular LDPC ensemble.
Clearly, for low enough $q$, $p_{\mathrm{DE}}(q)$ is strictly larger than $p_{\mathrm{cl}}(q)$.
This implies that QD-HDQI together with BPQM can achieve lower temperature for Gibbs sampling than QD-HDQI with any classical decoders for these values of $q$.}
\label{fig:defected_threshold_comparison_intro}
\end{figure}

\paragraph{Relation to prior work}
The original HDQI algorithm was restricted to the preparation of Gibbs states of Hamiltonians which had uniform couplings, i.e., $v_i=1$ in the terminology of \cref{eq:hamiltonian}.
Subsequent work by Bu, Gu and Li has generalized the HDQI algorithm to incorporate general non-uniform couplings \cite{bu_2026_generalizedhdqi}.
However, just like the original HDQI algorithm, their approach is strictly limited to noise models that have support only on information theoretically correctable errors, and as such, it is restricted to classical decoders and to the preparation of polynomial filters of the Hamiltonian instead of the actual Gibbs state.
Our QD-HDQI formalism naturally recovers the ability to handle Hamiltonians with non-uniform couplings in an elegant fashion by construction.
In the context of commuting Hamiltonians, the non-uniformity of the couplings simply translates to a non-uniformity of the channel parameters $p$ of the pure-state channels in the quantum decoding problem.

The robustness analysis of QD-HDQI that we present heavily borrows from techniques introduced in the analyses of Regev's reduction spearheaded in \cite{chailloux_improved_analysis,chailloux_opixsoftdec_2025}.
We note that robustness results on HDQI were also presented in \cite[Theorems 13 and 23]{schmidhuber2025hamiltoniandecodedquantuminterferometry} and \cite[Theorem 7 and Proposition 29]{bu_2026_generalizedhdqi}.
However, both of these results only cover a setting where the decoding problem remains information theoretically perfectly solvable, and as such, they do not suffice to cover the i.i.d. pure-state channel construction and the utilization of quantum decoders that we present in our manuscript.
For this reason, our analysis differs in one crucial qualitative aspect: The robustness only holds when averaged over $y\in\ff_2^m$.

The main application of our QD-HDQI algorithm is Gibbs sampling. In this regard, we would be remiss not to mention the complementary line of work which approaches this problem through \emph{thermalization}: one engineers dissipative quantum dynamics, typically a Lindbladian satisfying an appropriate detailed-balance condition, for which the target Gibbs state is a stationary state, and then evolves the system for sufficiently long that it mixes to thermal equilibrium \cite{kastoryano_quantum_2016, chen_efficient_2025, ding_efficient_2025, rouze_efficient_2026}. Recent works have established efficient implementations and rapid-mixing guarantees for such Gibbs samplers for increasingly broad classes of Hamiltonians. 

Our approach is conceptually rather different. Rather than preparing the Gibbs state as the stationary state of thermalizing dynamics, HDQI and QD-HDQI are to our knowledge the only Gibbs sampling algorithms that exploit the {\em algebraic structure} of the Hamiltonian through its symplectic code. In QD-HDQI in particular, the terms of the Hamiltonian define a symplectic code, while the Gibbs function determines an effective noise channel applied to codewords of that code. Consequently, whereas the complexity of thermalization-based approaches is controlled by properties such as the mixing time or spectral gap of the dissipative dynamics, the performance of QD-HDQI is governed by the strength of decoding algorithms for said noise channel. Ref. \cite{hwang2025gibbs} also reduces Gibbs-state preparation for commuting Hamiltonians to Gibbs sampling for an associated classical Hamiltonian. However, this reduction does not by itself yield an efficient Gibbs sampler, whereas our framework gives efficient algorithms whenever the corresponding decoding problem is efficiently solvable.

During the preparation of this manuscript, we became aware that some of our results on how to avoid the polynomial filter in HDQI were independently discovered and will be mentioned in an upcoming updated version of \cite{schmidhuber2025hamiltoniandecodedquantuminterferometry}.
We also became aware of simultaneous and independent efforts in \cite{zhou_2026_dqigibbs} that introduce quantum decoders to the framework of Decoded Quantum Interferometry and apply them to Gibbs sampling of classical Hamiltonians.

\paragraph{AI usage declaration}
Our main algorithm and its analysis in the commuting case (see \cref{sec:qdhdqi,sec:commuting}) were discovered and established without any AI assistance.
Back and forth discussions with LLMs were involved in the discovery of the classical Gibbs sampling algorithm in \cref{sec:dequantization}.
LLMs also helped prove certain technical steps in the non-commuting analysis, namely \cref{lem:independent_block_expansion,prop:block_pure_state_channel}.
Finally, the presentation of the manuscript was verified and improved using AI assistance.
The used models are GPT-5.6 Sol and GPT-6 Astra.
The authors take full responsibility for the correctness of the paper.

\section{Preliminaries}\label{sec:preliminaries}

\subsection{Pauli operators and the symplectic code}\label{subsec:pauli_symp}
Denote the $n$-qubit Pauli group by \begin{equation}
    \mathcal P_n \coloneqq \left\{ i^c (X^{a_1}Z^{b_1})\otimes\cdots\otimes(X^{a_n}Z^{b_n}) : c\in\{0,1,2,3\},\ a,b\in\ff_2^n \right\}
    \, .
\end{equation}
For \(P=i^cX^aZ^b\in\mathcal P_n\), define its symplectic representation and phase by \(\symp(P)\coloneqq(a,b)\in\ff_2^{2n}\) and \(\eta(P)\coloneqq i^c\).
The symplectic representation determines the commutation relation between Pauli operators.
In particular, if \(\symp(P)=(a,b)\) and \(\symp(Q)=(a',b')\), then 
\begin{equation}
    PQ=(-1)^{a\cdot b'+b\cdot a'}QP
    \, .
\end{equation}
Let \(\widetilde{\mathcal P}_n\coloneqq \{\mathds 1,X,Y,Z\}^{\otimes n}\subset\mathcal P_n\) denote the set of Pauli strings.
Although \(\widetilde{\mathcal P}_n\) is not closed under multiplication, it forms an orthonormal basis of the real vector space of Hermitian \(n\)-qubit operators.

\paragraph{Coherent symplectic representation}
Denote the maximally entangled state across two $n$-qubit quantum registers $A$ and $B$ by $\ket{\Phi^n}\coloneqq 2^{-n/2}\sum_{x\in\ff_2^n}\ket{x}_A\ket{x}_B$.

Define the coherent Bell-basis transformation \begin{equation}\label{eq:Usympl_definition} U_{\mathrm{sympl}} \coloneqq \bigl(\mathds 1_A\otimes(H_2^{\otimes n})_B\bigr) \mathrm{CNOT}_{B\to A}^{\otimes n}\end{equation} 
where $H_2$ denotes the Hadamard gate (the notation is chosen to distinguish it from the parity check matrix $H$ of the symplectic code).
For every \(Q\in\mathcal P_n\), it satisfies \begin{equation}\label{eq:Usympl_action} U_{\mathrm{sympl}} (Q\otimes\mathds 1)\ket{\Phi^n} = \eta(Q)\ket{\symp(Q)}. \end{equation} Thus, \(U_{\mathrm{sympl}}\) coherently maps the action of a Pauli operator on one half of \(\ket{\Phi^n}\) to its symplectic representation, while retaining its phase. For \(e\in\ff_2^m\), define \(Q_e\coloneqq P_1^{e_1}\cdots P_m^{e_m}\). \begin{definition}[Symplectic code of \(\cH\)] \label{def:sympH} For a Pauli Hamiltonian \(\cH=\sum_{i=1}^m(-1)^{y_i}v_iP_i\), let \begin{equation}\label{eq:symplectic_matrix} H\coloneqq \begin{pmatrix} \vert & \vert & & \vert \\ \symp(P_1) & \symp(P_2) & \cdots & \symp(P_m) \\ \vert & \vert & & \vert \end{pmatrix} \in\ff_2^{2n\times m}. \end{equation} The symplectic code of \(\cH\) is \(\cC_{\cH}\coloneqq\ker(H)\). \end{definition} By construction, \(He=\symp(Q_e)\). Hence, \Cref{eq:Usympl_action} gives \begin{equation}\label{eq:Usympl_Qe} U_{\mathrm{sympl}} (Q_e\otimes\mathds 1)\ket{\Phi^n} = \eta(Q_e)\ket{He}. \end{equation} Moreover, \(e+e'\in\cC_{\cH}\) if and only if \(Q_e\) and \(Q_{e'}\) differ only by a phase.

Another unitary which will be useful for us is $U_{\mathrm{phase}}$ which acts as
\begin{equation}
    \ket{e} \mapsto \eta(Q_e) \ket{e}
\end{equation}
for any $e\in\ff_2^m$.
It can be efficiently implemented, because $\eta(Q_e)$ can be efficiently computed from $e$.

\paragraph{Setup used throughout.} Let \(H\) be the symplectic matrix of the Pauli Hamiltonian \(\cH\), as defined in \Cref{def:sympH}. For simplicity, we assume throughout that \(H\in\ff_2^{2n\times m}\) has full row rank. If \(\operatorname{rank}(H)=r<2n\), we may instead work with a full-row-rank matrix \(\bar H\in\ff_2^{r\times m}\) whose rows form a basis of the row space of \(H\). This leaves the associated code unchanged, since \(\ker(\bar H)=\ker(H)\), and the syndrome \(\bar H e\) can be efficiently mapped back to the original symplectic representation \(He\) by a linear transformation. Throughout, we write 
\begin{itemize} 
\item The symplectic code of \(\cH\) is \(\cC_{\cH}\coloneqq\ker(H)\). 
\item Its dual code is \(\cC_{\cH}^{\perp}\coloneqq\operatorname{im}(H^\top)\), with full-column-rank generator matrix \(G\coloneqq H^\top\). 
\end{itemize} When the underlying Hamiltonian is clear from context, we write \(\cC\) and \(\cC^\perp\) in place of \(\cC_{\cH}\) and \(\cC_{\cH}^\perp\), respectively.

\subsection{Classical codeword and syndrome decoding}
We briefly recall the basics of the classical decoding problem for a given error distribution.
A binary linear $[n,k]$-code is a $k$-dimensional linear subspace $\cC\subset\ff_2^n$.
Linear codes are commonly described through a parity check matrix $H\in\ff_2^{m\times n}$ as $\cC=\ker(H)$ or through a generator matrix $G\in\ff_2^{m'\times n}$ where $\cC=\mathrm{rowspan}(G)$, for some $n-k\leq m$ and $k\leq m'\leq n$.
Note that neither the parity check matrix $H$, nor the generator matrix $G$ are unique representations.

For a given code $\cC$, its associated dual code is defined as
\begin{equation}
    \cC^\perp \coloneqq \{ x\in\ff_2^n | \forall c\in\cC : x\cdot c = 0\}
\end{equation}
where $x\cdot c$ denotes the standard inner product mod $2$.
Accordingly, the rows of a parity-check matrix for \(\cC\) generate \(\cC^\perp\).

The goal of decoding is to recover a codeword corrupted by noise. We assume throughout that the transmitted codeword is uniformly distributed over the code. We have the following two standard problems:
\begin{problem}[Classical codeword decoding problem]\label{prob:codeworddecode}
Let $\cC$ be an $[n,k]$ linear code and let $P_E$ be a distribution over $\ff_2^n$. \\
\noindent\textbf{Input:}
A noisy codeword $c+e$, where $c$ was uniformly randomly chosen from $\cC$ and $e$ was sampled from $P_E$. \\
\noindent\textbf{Output:}
$c$.
\end{problem}

\begin{problem}[Classical syndrome decoding problem]\label{prob:syndromedecode}
Let $\cC$ be an $[n,k]$ linear code with parity check matrix $H$ and let $P_E$ be a distribution over $\ff_2^n$. \\
\noindent\textbf{Input:}
$He$ where $e$ was sampled from $P_E$.\\
\noindent\textbf{Output:}
$e$.
\end{problem}

Since the transmitted codeword is uniformly distributed over $\cC$ and the error distribution does not depend on the transmitted codeword, codeword and syndrome decoding are equivalent under
randomized reductions preserving the average failure
probability.
Concretely, a syndrome decoder $\cD_{\mathrm{synd}}:\ff_2^m\rightarrow\ff_2^n$ induces a codeword decoder $\cD_{\mathrm{cw}}:\ff_2^n\rightarrow\ff_2^n$ through
$\cD_{\mathrm{cw}}(y)=y+\cD_{\mathrm{synd}}(Hy)$.
Conversely, given $s$, sample $y$ uniformly from
$\{y:Hy=s\}$ and output $y+\cD_{\mathrm{cw}}(y)$.
For a fixed error $e$ with $He=s$, the implicit
codeword $x$ s.t. $y=x+e$ is uniform in $\cC$, establishing, information theoretically:
\[
\text{Classical codeword decoding problem }\equiv \text{ Classical syndrome decoding problem}
\]

\subsection{Quantum decoding problem}\label{sec:quan_cod_synd_dec}
A central ingredient in Regev's reduction and in the quantum state-preparation algorithm studied here is the Quantum Decoding Problem \cite{chen_quantumalgo_2022,chailloux_quantumdecoding_2023}.
\begin{problem}[Quantum decoding problem]\label{prob:quantumdecode}
Let $\cC$ be an $[n,k]$ linear code and let $g:\ff_2^n\rightarrow\mathbb{C}$ s.t. $\norm{g}_2=1$. \\
\noindent\textbf{Input:}
An $n$-qubit quantum system in the state
\(\ket{\Psi_c} \coloneqq \sum_{e\in\ff_2^n} g(e)\ket{c+e}\)
\hfill\\
\noindent\textbf{Output:} $c$.
\end{problem}
In our algorithm, we encounter the coherent version of this problem, in which the decoder output is coherently stored in an ancilla register
\begin{equation}\label{eq:coherent_decoding}
    \sum_{c\in\cC}\ket{\Psi_{c}}\ket{0} \longmapsto \sum_{c\in\cC}\ket{\Psi_{c}}\ket{c}
    \, . 
\end{equation}
Running the above circuit \emph{in reverse} can in turn be seen as an uncomputation of the second register in the state $\sum_c\ket{\Psi_c}\ket{c}$.
This uncomputation is crucial for our QD-HDQI algorithm.
For an imperfect decoder, this uncomputation is only approximate. 
Note also that any POVM solving \cref{prob:quantumdecode} can be turned into a coherent decoder through an appropriate dilation.

Any classical algorithm that solves \cref{prob:codeworddecode} with an error distribution $P_E(e)=\abs{f(e)}^2$ can be turned into a quantum decoder solving \cref{prob:quantumdecode} with noise function $f$.
To do so, simply measure the $\ket{\Psi_c}$ in the computational basis, yielding $c+e$ with probability $\abs{f(e)}^2$, and then feed $c+e$ into the classical decoder.
To realize a coherent quantum decoder as in \cref{eq:coherent_decoding}, one would instead run the classical decoder coherently on the superposition $\ket{\Psi_c}$.
Consequently we have a reduction,
\[
\text{Quantum decoding problem } \leq \text{ Classical codeword problem }
\]

It is important to note however that this reduction does not imply that \emph{any} quantum decoding problem can be solved through classical decoding.
As we will see in \cref{sec:capacities}, there exist instances of \cref{prob:quantumdecode} which are solvable with vanishing error, however, the associated classical decoding problem with $P_E(e)=\abs{f(e)}^2$ is information theoretically unsolvable.
To solve the quantum decoding problem in such regimes, it is therefore necessary to use truly quantum decoding algorithms which utilize the coherent (as opposed to incoherent) superposition of the errors in $\ket{\Psi_c}$.
In the rest of the paper, when we refer to \emph{quantum decoders}, we implicitly refer to such decoding strategies that go beyond the utilization of a simple classical decoder.

\subsection{Information theoretic limits of decoding in the i.i.d. setting}\label{sec:capacities}
In this section we discuss particularly structured instances of the quantum decoding problem where the noise function $f$ factorizes into bitwise terms
\begin{equation}\label{eq:f_factorized}
    f(e) = f_1(e_1) f_2(e_2) \cdots f_n(e_n)
    \, , \qquad
    f_i(e_i) = \begin{cases}
        \sqrt{1-p_i} & \text{if } e_i=0 \\
        \sqrt{p_i} & \text{if } e_i=1 \\
    \end{cases}
\end{equation}
for some parameters $p_i\in [0,1/2]$.
This class of quantum decoding problems is particularly relevant for two reasons.
First, it will naturally show up in the context of QD-HDQI for commuting Hamiltonians, as we will see in \cref{sec:commuting}.
Second, it has a particularly simple channel coding interpretation, and as such, allows for a simple information theoretic demonstration that the quantum decoding problem can be solved strictly better than what is allowed by classical decoders.

First, let us introduce two important channels.
\begin{definition}[Pure state channel {{$\psc[p]$}}]
    For $p\in [0,1/2]$, we denote by $\psc[p]$ the binary-input symmetric pure-state classical-quantum channel defined as
    \begin{equation}
        \ff_2\ni x \mapsto \psc[p](x)\coloneqq \sqrt{1-p}\ket{x} + \sqrt{p}\ket{1-x}
        \, .
    \end{equation}
\end{definition}
\begin{definition}[Binary symmetric channel {{$\bsc[p]$}}]
    For $p\in [0,1/2]$, we denote by $\bsc[p]$ the binary symmetric channel, which maps $x\in\ff_2$ to $x$ with probability $(1-p)$ and to $1-x$ with probability $p$.
\end{definition}
Observe that for the factorized noise function in \cref{eq:f_factorized}, the associated quantum decoding problem is tasked to retrieve the codeword $c\in\cC$ from
\begin{equation}
    \ket{\Psi_c} = \bigotimes_{i=1}^n \psc[p_i](c_i)
    \, .
\end{equation}
In words, the task is to decode the code $\cC$ transmitted over a sequence of pure-state channels $\psc[p_i]$.
In particular, when all channel parameters are equal $p_i\equiv p$, the classical-quantum channel coding theorem states that reliable communication can be possible only if the rate $R$ of the code is below the Holevo capacity of the pure-state channel $C_H(p)=h_2\left(\frac12-\sqrt{p(1-p)}\right)$, where $h_2$ is the binary entropy function.
Rephrasing this as a bound on $p$, for $R\in(0,1)$, we obtain
\begin{equation}\label{eq:holevo_bound}
    p < \frac{1}{2} - \sqrt{h_2^{-1}(R) (1 - h_2^{-1}(R))}
    \, .
\end{equation}
If we instead solve this quantum decoding problem with a classical codeword decoder, then the effective noise seen by the classical decoders is Bernoulli $P_E(e)=\abs{f(e)}^2=(1-p)^{n-\abs{e}}p^{\abs{e}}$.
Therefore, the associated decoding problem corresponds to decoding the code transmitted over the binary symmetric channel $\bsc[p]$.
The reduction thus reads
\[
\text{Quantum decoding problem over } \psc[p] \leq \text{ Classical codeword decoding problem over } \bsc[p]
\]
The Shannon channel coding theorem states that reliable  communication over $\bsc[p]$ is possible only if the rate $R$ of the code is below the Shannon capacity $C_S(p)=1-h_2(p)$.
In particular, that means that the strategy of solving the quantum decoding problem with a classical decoder can only succeed if
\begin{equation}\label{eq:shannon_bound}
    p < h_2^{-1}(1-R)
    \, .
\end{equation}
Note that the bound in \cref{eq:shannon_bound} is strictly tighter than the Holevo bound in \cref{eq:holevo_bound}.
In summary, truly quantum decoders can potentially surpass the Shannon bound that fundamentally limits classical decoders.
And indeed, there are some known efficient quantum decoders which have been previously shown to surpass this Shannon barrier\cite{brandsen_bpqm_2022,piveteau_2025_bpqmturbo}.

\section{Quantum decoding Hamiltonian Decoded Quantum Interferometry (QD-HDQI)}\label{sec:qdhdqi}

In this section, we present and analyze QD-HDQI, in which the decoding step is formulated in the codeword picture so that quantum decoders can be used. Our construction is inspired by quantum reductions for the Short Codeword Problem (SCP), following Regev's reduction and its code-based formulation~\cite{regev_reduction_2009,thomas_reduction_2024}, as well as its extension to the Quantum Decoding Problem by Chailloux and Tillich~\cite{chailloux_quantumdecoding_2023}. These reductions show that the SCP can be approached through either syndrome decoding or codeword decoding, with both formulations producing equivalent states supported on the dual code. 
The transformation between the two formulations consists of a Fourier transform, a linear change of coordinates, and a second Fourier transform, and can be summarized schematically as follows:

\begin{equation}
    \frac{1}{\sqrt{\abs{\cC}}} \sum_{\substack{c\in\cC\\e\in\ff_2^m}} g(e)\ket{c+e} \xrightarrow{} \sum_{e\in\ff_2^m}g(e)\ket{He}
\end{equation}

We use this correspondence to reformulate the HDQI syndrome-decoding step in the codeword picture.

\subsection{Codeword-to-syndrome transformation}

Using \(\cC_{\cH}\) as a symplectic code, and \(G=H^\top\) as a generator matrix of \(\cC_{\cH}^{\perp}\), let \(U_{\mathrm{encode}}^{(G)}\) denote the corresponding encoding isometry:
\begin{equation}
\label{eq:dual_encoding_isometry}
U_{\mathrm{encode}}^{(G)} \colon
\ket{u}\ket{0^{m-2n}}
\longmapsto
\ket{Gu} = \ket{H^\top u},
\qquad u\in\ff_2^{2n}.
\end{equation}

The following lemma gives the resulting codeword-to-syndrome transformation, as used in~\cite{thomas_reduction_2024,chailloux_quantumdecoding_2023}.

\begin{lemma}[Codeword-to-syndrome transformation] \label{lem:codeword_to_syndrome} Let \(H\in\ff_2^{2n\times m}\) be a full row-rank parity check matrix for $\calC_{\calH}$, i.e. \(\cC_{\cH}\coloneqq\ker(H)\), and let \(g:\ff_2^m\rightarrow\mathbb C\). Applying \(H_2^{\otimes m}\), unencoding the dual code using \(\left(U_{\mathrm{encode}}^{(G)}\right)^\dagger\), and then applying \(H_2^{\otimes 2n}\) implements the transformation \begin{equation}\label{eq:codeword_to_syndrome} \frac{1}{\sqrt{\abs{\cC_{\cH}}}} \sum_{\substack{c\in\cC_{\cH}\\e\in\ff_2^m}} g(e)\ket{c+e} \longmapsto \sum_{e\in\ff_2^m}g(e)\ket{He}. \end{equation} \end{lemma} \begin{proof} Applying \(H_2^{\otimes m}\) to the state on the left-hand side of \Cref{eq:codeword_to_syndrome} gives \begin{align} &\frac{1}{\sqrt{\abs{\cC_{\cH}}2^m}} \sum_{\substack{c\in\cC_{\cH}\\e,y\in\ff_2^m}} g(e)(-1)^{y\cdot(c+e)}\ket y \\ &\qquad= \sqrt{\abs{\cC_{\cH}}} \sum_{y\in\cC_{\cH}^{\perp}} \widehat g(y)\ket y, \end{align} where we used \(\sum_{c\in\cC_{\cH}}(-1)^{y\cdot c} =\abs{\cC_{\cH}}\) for \(y\in\cC_{\cH}^{\perp}\), and zero otherwise. Since \(G=H^\top\) is a full-column-rank generator matrix of \(\cC_{\cH}^{\perp}\), applying \(\left(U_{\mathrm{encode}}^{(G)}\right)^\dagger\) yields \begin{equation} \sqrt{\abs{\cC_{\cH}}} \sum_{u\in\ff_2^{2n}} \widehat g(H^\top u)\ket u. \end{equation} A final application of \(H_2^{\otimes 2n}\) gives \(\sum_e g(e)\ket{He}\), as claimed. \end{proof}

\subsection{QD-HDQI algorithm description}

\label{subsec:codewordHDQI}

We now combine the coherent version of the quantum decoding problem described in \Cref{sec:preliminaries} with the codeword-to-syndrome transformation of 
\Cref{lem:codeword_to_syndrome}. We start by preparing
\begin{equation}
    \frac{1}{\sqrt{|\cC_{\cH}|}}
    \sum_{c\in\cC_{\cH}}\ket{\Psi_c}\ket c,
    \qquad
    \ket{\Psi_c}
    =\sum_{e\in\ff_2^m}f(e)\eta(Q_e)\ket{c+e}.
\end{equation}
Running the coherent decoder in reverse uncomputes the codeword register while preserving the noisy-codeword superposition.
We then apply the codeword-to-syndrome transformation of \Cref{lem:codeword_to_syndrome}, recovering the syndrome-picture state used in HDQI. Finally, applying \(U_{\mathrm{sympl}}^\dagger\) yields the target state.

\Cref{alg:algo1} summarizes this idealized construction. We then consider imperfect decoding and establish the corresponding error and runtime guarantees.

\begin{algorithm}[H] \caption{Quantum Decoding-Hamiltonian Decoded Quantum Interferometry (QD-HDQI)} \label{alg:algo1}
\begin{algorithmic}[1]
\Require A Pauli Hamiltonian \(\cH\), its symplectic code \(\cC_{\cH}=\ker(H)\), and a normalized noise function \(f:\ff_2^m\rightarrow\mathbb C\). 
\Ensure A quantum system in the state \(\widetilde{\rho}\) satisfying \( \norm{\widetilde{\rho}-\rho_f}_1 \leq \varepsilon, \) where
\begin{equation*}
    \rho_f \coloneqq \frac{1}{2^n} A_fA_f^\dagger
    \, ,\qquad\
    A_f \coloneqq \sum_ef(e)Q_e
    \, .
\end{equation*}
\State \textbf{Initial-state preparation:} Prepare $\sum_{e\in\ff_2^m}f(e)\ket e$ and $\frac{1}{\sqrt{\abs{\cC}}}\sum_{c\in\cC}\ket c $ in separate registers. The latter is always efficient due to injectivity of the encoding map; we devote the remaining sections to identifying Hamiltonians for which the former is efficient. 
Then:
\begin{align} 
\left(\sum_{e\in\ff_2^m}f(e)\ket e\right) \otimes \frac{1}{\sqrt{\abs{\cC}}}\sum_{c\in\cC}\ket c \xrightarrow{\,U_{\mathrm{phase}}\otimes\mathds I,\, \mathrm{CNOT}^{\otimes m}\,}& \frac{1}{\sqrt{\abs{\cC}}} \sum_{c\in\cC} \left(\sum_{e\in\ff_2^m} f(e)\eta(Q_e)\ket{c+e}\right)\ket c 
\\ =& \frac{1}{\sqrt{\abs{\cC}}} \sum_{c\in\cC}\ket{\Psi_c}\ket c 
\end{align}

\State \textbf{Uncompute using a solver for the quantum decoding problem}, outputting the state

\begin{equation}
     \frac{1}{\sqrt{\abs{\cC}}} \sum_{c\in\cC}\ket{\Psi_c}\ket 0
\end{equation}

\State \textbf{Transformation to syndrome picture:} By \Cref{lem:codeword_to_syndrome}, the above "codeword" state can be transformed into the ``syndrome" state
\begin{equation} \sum_{e\in\ff_2^m} g(e)\ket{ H e}. \end{equation}

\State Apply $U_{\mathrm{sympl}}^\dagger$ and trace out the second register.

\end{algorithmic}
\end{algorithm}

Step~2 describes the idealized action of the decoder.
If the decoder succeeds deterministically, it exactly uncomputes the codeword register and the algorithm prepares \(\rho_f\).
In general, the decoder may fail to identify $c$ with a certain nonzero probability $\pfail$.
As long as the $\pfail$ is small enough, we will show that the uncomputation step can still be performed approximatively.
We briefly note that whenever $\pfail>0$, the ideal target state $\rho_f$ might not be perfectly normalized, i.e., $\tr[\rho_f]\neq 1$.

We next specify the coherent implementation of the decoder and analyze how decoding errors affect the output.

\subsection{Robustness to imperfect decoding}

A general quantum decoder can be described by a POVM \(\{\Lambda_{c'}\}_{c'\in\cC}\), where \(c'\) is the decoder's estimate of the transmitted codeword \(c\).
Its average failure probability is
\begin{equation}
    \pfail\bigl(\{\Lambda_{c'}\}_{c'\in\cC}\bigr) \coloneqq 1-\frac{1}{\abs{\cC}} \sum_{c\in\cC} \braopket{\Psi_c}{\Lambda_c}{\Psi_c}
    \, .
\end{equation} 
To implement the uncomputation in Step~2 of \Cref{alg:algo1}, we require a coherent implementation of this POVM (i.e. a dilation), which is generally described by a unitary $U_{\mathrm{dec}}$ acting as
\begin{equation} \label{eq:general_quantum_decoder} U_{\mathrm{dec}} \bigl(\ket{\Psi_c}\ket{0_{\mathrm{anc}}}\bigr) = \sum_{c'\in\cC} \ket{f_{c,c'}}\ket{c'}, \qquad \norm{\ket{f_{c,c'}}}_2^2 = \braopket{\Psi_c}{\Lambda_{c'}}{\Psi_c}. \end{equation} 

Here, \(\ket{c'}\) stores the decoder's output, while \(\ket{f_{c,c'}}\) denotes the corresponding subnormalized workspace state. Since \(\{\Lambda_{c'}\}_{c'\in\cC}\) is a POVM, for every \(c\in\cC\),
\(
\sum_{c'\in\cC}\norm{\ket{f_{c,c'}}}_2^2=1.
\)
Hence, the average decoding failure probability is
\begin{equation}
    \pfail
    =
    1-\frac{1}{\abs{\cC}}
    \sum_{c\in\cC}
    \norm{\ket{f_{c,c}}}_2^2.
\end{equation}
Applying the decoder coherently to the uniform superposition of noisy codewords, with the corresponding codeword \(c\) stored in a reference register, gives
\begin{equation}
\label{eq:coherent-decoder-output}
\frac{1}{\sqrt{\abs{\cC}}}
\sum_{c\in\cC}
\ket{\Psi_c}\ket{c}\ket{0_{\mathrm{anc}}}
\xmapsto{U_{\mathrm{dec}}}
\frac{1}{\sqrt{\abs{\cC}}}
\sum_{c,c'\in\cC}
\ket{f_{c,c'}}\ket{c'}\ket{c}.
\end{equation}

We now use this coherent implementation to approximately realize the codeword-uncomputation step required in \Cref{alg:algo1}.

\begin{algorithm}[!htbp] \caption{Coherent codeword uncomputation} \label{alg:postselected-codeword-uncomputation} \textbf{Input:} The state \( \frac{1}{\sqrt{\abs{\cC}}} \sum_{c\in\cC} \ket{\Psi_c}_{A}\ket c_{R}, \) and a coherent decoder \(U_{\mathrm{dec}}\).

\medskip 

\textbf{Output:} A state approximating \( \frac{1}{\sqrt{\abs{\cC}}} \sum_{c\in\cC} \ket{\Psi_c}_{A}\ket{0^m}_{R}. \) \begin{algorithmic}[1] \State Append a decoder-output register \(D\) initialized to \(\ket{0_{\mathrm{anc}}}\). 
\State Apply \(U_{\mathrm{dec}}\) to registers \(A\) and \(D\):
\begin{equation}
    \frac{1}{\sqrt{\abs{\cC}}} \sum_{c\in\cC} \ket{\Psi_c}_{A}\ket c_{R}\ket{0_{\mathrm{anc}}}_{D} \longmapsto \frac{1}{\sqrt{\abs{\cC}}} \sum_{c,c'\in\cC} \ket{f_{c,c'}}_{A}\ket c_{R}\ket{c'}_{D}. 
\end{equation}
\State Add the decoder output into the reference register:
\begin{equation}
    \ket c_{R}\ket{c'}_{D} \longmapsto \ket{c+c'}_{R}\ket{c'}_{D}.
\end{equation}
\State Measure \(R\) and postselect on the outcome \(0^m\). Postselecting the reference register \(R\) on \(\ket{0^m}\) retains precisely the correctly decoded branches \(c'=c\). The postselection succeeds with probability \(  1-\pfail. \) Conditioned on success, the normalized postselected state is \begin{equation} \label{eq:decoder_postselected_state} \ket{\Psi} \coloneqq \frac{1}{\sqrt{\abs{\cC}(1-\pfail)}} \sum_{c\in\cC} \ket{f_{c,c}}_{A} \ket{0^m}_{R} \ket c_{D}. \end{equation}
\State Apply \(U_{\mathrm{dec}}^\dagger\) to registers \(A\) and \(D\). \State \Return the registers \(A\) and \(R\). \end{algorithmic} \end{algorithm}

For comparison, applying \(U_{\mathrm{dec}}\) to the ideal output of the uncomputation procedure gives the ideal output of that step \begin{equation} \label{eq:ideal_coherent_decoder_output} \ket{\Psi'} \coloneqq \frac{1}{\sqrt{\abs{\cC}}} \sum_{c,c'\in\cC} \ket{f_{c,c'}}_{A} \ket{0^m}_{R} \ket{c'}_{D} \end{equation}
(which might in general not be a properly normalized state).
Applying \(U_{\mathrm{dec}}^\dagger\) to this state recovers 
\begin{equation}
    \frac{1}{\sqrt{\abs{\cC}}} \sum_{c\in\cC} \ket{\Psi_c}_{A}\ket{0^m}_{R}
    \, ,
\end{equation}
which is the ideal output of the codeword-uncomputation step. The postselected state contains only the correctly decoded branches, whereas the ideal state also contains the branches \(c'\neq c\). Although \(\pfail\) controls the total weight of these incorrect branches, their coherent contributions can interfere and lead to a larger error. The following theorem bounds this contribution through a coherent interference term we denote by \(\alpha\).

\begin{theorem}[Error and runtime guarantees for \cref{alg:algo1}]\label{thm:analysis_codeword}
    Consider a function $f:\ff_2^m\rightarrow\mathbb{C}$ such that $\norm{f}_2=1$ and the state $\sum_ef(e)\ket{e}$ can be efficiently prepared.
    Consider a quantum decoding algorithm for the symplectic code $\mathcal{C}_{\mathcal H}$ which can be efficiently realized and which fails with probability $\pfail$.
    Furthermore, assume that $\frac{1}{1-\pfail} = \mathcal{O}(\mathrm{poly}(m))$.
Then, the above algorithm is efficient and prepares the state $\tilde\rho$ which is close to $\rho_f$
    \begin{equation}
        \norm{\tilde\rho - \rho_f}_1 \leq \pfail + 2\alpha + \alpha^2
        \, .
    \end{equation}
    where $\alpha$ is 
    \begin{equation}
        \alpha\coloneqq
        \sqrt{ \frac{1}{\abs{\cC}}\sum_{\bar c\in\cC} \norm{\sum_{c\neq \bar c} \ket{f_{c,\bar c}}}^2 }
        \, .
    \end{equation}
\end{theorem}
\begin{proof}
    The post-selection probability is $1-\pfail$, so if $\frac{1}{1-\pfail} = \mathcal{O}(\mathrm{poly}(m))$ then all steps of the algorithm, including the post-selection, are efficient.  
    
    The only error in the above derivation occurs in the decoding step. By contractivity of the trace norm, it therefore suffices to bound the trace distance between the post-measurement state \begin{equation} 
    \ket{\Psi}= \frac{1}{\sqrt{|\cC|(1-\pfail)}} \sum_{c\in\cC} \ket{f_{c,c}}\otimes\ket{c} 
    \end{equation} and the ideal state
    \begin{equation} 
    \ket{\Psi'}=\frac{1}{\sqrt{|\cC|}} \sum_{c,\bar c\in\cC} \ket{f_{c,\bar c}}\otimes\ket{\bar c}
    \end{equation}
    (which may in general not be properly normalized).
    We can write $\ket{\Psi'}=\sqrt{1-\pfail}\ket{\Psi} + \ket{B}$ where
    \begin{equation}
        \ket{B} \coloneqq \frac{1}{\sqrt{\abs{\cC}}}\sum_{c\neq \bar c}\ket{f_{c,\bar c}}\otimes\ket{\bar c}
        \, .
    \end{equation}
    Observe that
    \begin{align}
        \norm{ \proj{\Psi} - \proj{\Psi'} }_1 
        &= \norm{ \pfail\proj{\Psi} - \sqrt{1-\pfail}(\ketbra{\Psi}{B} + \ketbra{B}{\Psi}) - \proj{B} }_1 \\
        &\leq \pfail\norm{\Psi}^2 + 2\sqrt{1-\pfail}\norm{\Psi}\norm{B} + \norm{B}^2
    \end{align}
    The desired statement follows from $\norm{B}=\alpha$.
\end{proof}

In summary, the QD-HDQI algorithm can efficiently and successfully prepare the state $\rho_f$ if the following conditions are fulfilled:
\begin{itemize} 
\item The coefficient state \( \sum_{e\in\ff_2^m} f(e)\ket e \) can be prepared efficiently.
\item The coherent codeword decoder can be implemented efficiently and has small failure probability \(\pfail\)
\item The decoder's coherent interference term \[ \alpha \coloneqq \left( \frac{1}{\abs{\cC}} \sum_{\bar c\in\cC} \left\| \sum_{c\neq\bar c}\ket{f_{c,\bar c}} \right\|_2^2 \right)^{1/2} \] is small. \end{itemize}

In the next section, we instantiate this framework for Gibbs-state
preparation for commuting Pauli Hamiltonians.

\section{Application to commuting Hamiltonians}\label{sec:commuting}
As discussed at the end of \cref{sec:qdhdqi}, the successful preparation of the Gibbs state of a Hamiltonian $\cH$ via QD-HDQI requires three conditions to be satisfied.
In this section, we will show that these three conditions simplify significantly whenever $\cH$ is commuting, i.e., only contains pairwise commuting Pauli terms $P_i$.

Concretely, the noise state is a simple product state which can always be prepared efficiently, and the quantum decoding problem can be described as the decoding problem associated to transmitting the symplectic code over a series of independent binary-input pure-state classical-quantum channels.
Furthermore, we can show that $\alpha$ vanishes by averaging over the uniformly random values of $y$.

These simplifications are a consequence of the fact that, for commuting Hamiltonians, the Gibbs state factorizes in a convenient manner, yielding a particularly simple and elegant noise function $f$.
We start by describing this noise function in the following subsection.

\subsection{Noise function from the Gibbs state expansion}
Recall that for QD-HDQI, we require a function $f$ which fulfills $\rho_f = \exp(-\beta\cH) / Z_{\beta}$, which is ensured by choosing,
\begin{equation}
    \sum_{e\in\ff_2^m} f(e) Q_e = \sqrt{\frac{2^n}{Z_\beta}}\exp\left(-\frac{\beta}{2}\cH\right)
    \, .
\end{equation}
To find such a noise function, we use the crucial observation that the Gibbs state factorizes for commuting Hamiltonians
\begin{equation}
    \exp\left(-\frac{\beta}{2}\cH\right)
    = \prod_{i=1}^m
    \exp\left(-\frac{\beta}{2}(-1)^{y_i}v_i P_i\right) 
    = \prod_{i=1}^m
    \left(
        \cosh\left(-\frac{\beta}{2}(-1)^{y_i}v_i\right)\mathds{1}
        +
        \sinh\left(-\frac{\beta}{2}(-1)^{y_i}v_i\right)P_i
    \right) 
    \, .
\end{equation}
This directly translates to a simple expansion of $\exp(-\beta\cH/2)$ in terms of $Q_e$.
\begin{lemma}\label{lemma:coeff_function}
    Let $\cH=\sum_{i=1}^m(-1)^{y_i}v_iP_i$, where $v_i\geq 0$, be a Hamiltonian with mutually commuting Pauli terms.
    Define
    \begin{equation}
        p_i \coloneqq \frac{\sinh^2(\beta v_i/2)}{\cosh^2(\beta v_i/2)+\sinh^2(\beta v_i/2)} =\frac{1}{2}\left(1-\sech(\beta v_i)\right)
        \quad\text{and}\quad
        Z_\beta' \coloneqq \frac{Z_\beta}{2^n\prod_{i=1}^m\cosh(\beta v_i)}
        \, .
    \end{equation}
    Then the function $\tilde f:\ff_2^m\rightarrow\mathbb{C}$ defined by
    \begin{equation}
        \tilde f(e) \coloneqq \frac{1}{\sqrt{Z_\beta'}} (-1)^{(y+\vec 1)\cdot e} \prod_{i=1}^m \sqrt{(1-p_i)^{(1-e_i)}p_i^{e_i}}
    \end{equation}
    satisfies
    \begin{equation}
        \sum_{e\in\ff_2^m}\tilde f(e)Q_e = \sqrt{\frac{2^n}{Z_\beta}}\exp\left(-\frac{\beta}{2}\cH\right)
        \, .
    \end{equation}
\end{lemma}
Here, $\vec{1}$ denotes the all-1 vector.
\begin{proof}
Since the Pauli terms commute,
    \begin{align}
        \sqrt{\frac{2^n}{Z_\beta}}
        \exp\left(-\frac{\beta}{2}\cH\right)
        &=
        \sqrt{\frac{2^n}{Z_\beta}}
        \prod_{i=1}^m
        \left(
            \cosh\left(\frac{\beta v_i}{2}\right)\mathds{1}
            +(-1)^{y_i+1}
            \sinh\left(\frac{\beta v_i}{2}\right)P_i
        \right) \\
        &=
        \frac{1}{\sqrt{Z_\beta'}}
        \prod_{i=1}^m
        \left(
            \sqrt{1-p_i}\,\mathds{1}
            +(-1)^{y_i+1}\sqrt{p_i}\,P_i
        \right).
    \end{align}
    Expanding the product yields
    \begin{equation}
        \frac{1}{\sqrt{Z_\beta'}}
        \sum_{e\in\ff_2^m}
        (-1)^{(y+\vec 1)\cdot e}
        \prod_{i=1}^m
        (1-p_i)^{(1-e_i)/2}p_i^{e_i/2}Q_e
        =
        \sum_{e\in\ff_2^m}\tilde f(e)Q_e,
    \end{equation}
    which proves the desired identity.
\end{proof}
Notice that $\tilde f$ is not yet a valid noise function that we can utilize in QD-HDQI: It may generally not be properly normalized, since $\norm{\tilde f}_2=1/\sqrt{Z_\beta'}$ and $Z_\beta'$ is generally not equal to $1$.
Instead our strategy will be to pick the normalized version of $\tilde f$ as our noise function $f$, i.e.,
\begin{equation}\label{eq:f_commuting}
    f(e) \coloneqq \frac{\tilde f(e)}{\norm{\tilde f}_2} = (-1)^{(y+\vec 1)\cdot e} \prod_{i=1}^m \sqrt{(1-p_i)^{(1-e_i)}p_i^{e_i}}
    \, .
\end{equation}
The state prepared by QD-HDQI will therefore be $\rho_f$ instead of $\rho_{\tilde f}=\exp(-\beta\cH)/Z_\beta$.
We will show that the error induced by this approximation is vanishing, by showing that $Z_{\beta}'\approx 1$ for uniformly random choices of $y$. 

In the rest of this subsection, we verify that $f$ is a well-chosen noise function for QD-HDQI by verifying that the noise state is efficiently preparable, by characterizing the associated quantum decoding problem, and by showing that $Z_{\beta}'\approx 1$ on average.

\paragraph{Noise state preparation}
Notice that the noise state $\sum_e f(e)\ket{e}$ can be efficiently prepared by first preparing the product state
\begin{equation}
    \bigotimes_{i=1}^m\left( \sqrt{1-p_i}\ket{0} + \sqrt{p_i}\ket{1} \right)
\end{equation}
and then applying the operator $Z^{y\oplus 1}$.
This can be achieved with a constant-depth quantum circuit.

\paragraph{The corresponding quantum decoding problem is over a pure state channel}
Let us now have a look at the quantum decoding problem that needs to be solved within QD-HDQI with the noise function defined in \cref{eq:f_commuting}.
The task is to recover $c\in\cC$ from the state
\begin{equation}
    \ket{\Psi_c}=\sum_e f(e)\eta(Q_e)\ket{c+e}
    \, .
\end{equation}
The amplitudes of this state are a bit cumbersome, because the phase terms $(-1)^{(y+1)\cdot e}$ and $\eta(Q_e)$ generally cause correlations across the qubits.
However, a simple unitary basis change with $U_{\mathrm{phase}}^{\dagger}$ and $Z^{y\oplus 1}$ allows us to remove these phases
\begin{align}
    Z^{y \oplus 1}U_{\mathrm{phase}}^{\dagger}\ket{\Psi_c}
    &=\overline{\eta(Q_c)}(-1)^{(y\oplus 1)\cdot c}\sum_{e\in\ff_2^m} \prod\limits_{i=1}^m \left(\sqrt{(1-p_i)^{1-e_i}p_i^{e_i}} \right)\ket{c+e} \\
    &= \overline{\eta(Q_c)}(-1)^{(y\oplus 1)\cdot c} \left(\bigotimes_{i=1}^m \psc[p_i](c_i)\right)
\end{align}
where we used the fact that
\begin{equation}
    \eta(Q_{c+e}) = \eta(Q_cQ_e) = \eta(Q_c)\eta(Q_e)
\end{equation}
because the Hamiltonian is commutative and $Q_c$ is proportional to the identity.
Thus, the quantum decoding problem to be solved corresponds to transmitting the symplectic code over $\otimes_{i=1}^m \psc[p_i]$.
In particular, if we use a classical decoder to solve this quantum decoding problem, then the effective noise seen by said classical decoder corresponds to $\otimes_{i=1}^m \bsc[p_i]$.

In the special case where all $v_i$ are chosen to be $1$, then all channel probabilities are equal $p_i\equiv p$ and the information theoretic limitations of quantum and classical decoders are particularly easy to state.
In particular, the largest $p$ for which the quantum decoding problem can be used is limited by the Holevo bound.
And in particular, if a classical decoder is used, then the largest achievable $p$ is limited by the stricter Shannon bound.

\paragraph{$\tilde f$ is almost normalized}
We end this section by showing that indeed $Z_{\beta}'$ is typically very close to $1$, at least when we average over the signs $y$ of the Hamiltonian and whenever we are in the regime where the associated quantum decoding problem over $\otimes_i\psc[p_i]$ is information theoretically well solvable.
\begin{lemma}\label{lem:average_normalization}
    Let $\cH$ be a Pauli Hamiltonian of the form \cref{eq:hamiltonian} (not necessarily commuting).    
    Consider a family of functions $g_y:\ff_2^m\rightarrow\mathbb{C}$ of the form
    \begin{equation}
        g_y(e) = \frac{1}{\sqrt{K_y}}(-1)^{e\cdot y}h(e)
    \end{equation}
    for $y\in\ff_2^m$, $K_y>0$ and some $h:\ff_2^m\rightarrow\mathbb{C}$ independent of $y$ s.t. $\norm{h}_2=1$.
    Furthermore, assume that
    \begin{equation}
        \forall y\in\ff_2^m: \frac{1}{\sqrt{2^n}}\sum_e g_y(e) Q_e = \sqrt{\rho}_y
    \end{equation}
    for some normalized quantum states $\rho_y$ (i.e., $\rho_y\geq 0$ and $\tr[\rho_y]=1$).

    Furthermore, suppose there exists a quantum decoder that can solve the quantum decoding problem associated with the symplectic code $\cC$ and the noise function $h(e)\eta(Q_e)$ with a failure probability $\pfail$.
    Then,
    \begin{equation}
        \mathbb{E}_{Y\sim\mathrm{Uniform}(\ff_2^m)}\left[ \abs{K_y-1} \right]
        \leq 2\sqrt{\pfail}
        \, .
    \end{equation}
\end{lemma}
We briefly highlight that the existence of the quantum decoder in the above lemma should be understood as an information theoretic statement, i.e., we do not require that decoder to be efficiently implementable.
\begin{proof}
   Consider the expansion of the Hermitian operator $\sqrt{\rho_y}$ into the $n$-qubit Pauli basis $\widetilde{\mathcal{P}}_n=\{B_1,B_2,\dots,B_{4^n}\}$
    \begin{equation}
        \sqrt{\rho_y} = \sum_i c_i B_i
        \, .
    \end{equation}
    We have
    \begin{align}
        \sum_i c_i^2
        &= \frac{1}{4^n} \sum_{P\in\widetilde{\mathcal{P}}_n} \tr\left[ P \sqrt{\rho_y} \right]^2 \\
        &= \frac{1}{4^n} \sum_{P\in\widetilde{\mathcal{P}}_n} \tr\left[ (P\otimes P) (\sqrt{\rho_y}\otimes\sqrt{\rho_y}) \right] \\
        &= \frac{1}{2^n}\tr\left[ \mathrm{SWAP} (\sqrt{\rho_y}\otimes\sqrt{\rho_y}) \right] \\
        &= \frac{1}{2^n}\tr[\rho_y] \\
        &= \frac{1}{2^n}
    \end{align}
    where we used the SWAP trick.

    Next, observe that
    \begin{equation}
        c_i = \frac{\overline{\eta(B_i)}}{\sqrt{2^n}}\sum_{e:He=\symp(B_i)} g_y(e)\eta(Q_e)
        \, .
    \end{equation}
    Therefore, we have        
    \begin{equation}
        1
        = 2^n\sum_i \abs{c_i}^2
        = \sum_{s\in\ff_2^{2n}} \abs{ \sum_{e:He=s} g_y(e)\eta(Q_e) }^2
    \end{equation}
    or equivalently,
    \begin{equation}\label{eq:Ky_simplified}
        K_y
        = \sum_{s\in\ff_2^{2n}} \abs{ \sum_{e:He=s} h(e)(-1)^{e\cdot y}\eta(Q_e) }^2
    \end{equation}

    Next, let us consider the Gram matrix $G_{c,c'}\coloneqq \braket{\Psi_c}{\Psi_{c'}}$.
    Due to the symmetry $\ket{\Psi_c}=X^c\ket{\Psi_0}$, the matrix is circulant and the value on the off-diagonal indexed by $d\in\cC$ is given by
    \begin{equation}
        r(d) \coloneqq \braket{\Psi_d}{\Psi_0}
        = \sum_e \overline{h(e)\eta(Q_e)}h(e+d)\eta(Q_{e+d})
        \, .
    \end{equation}
    The eigenvalues of $G$ (indexed by $y\in\ff_2^m)$ are therefore given by
    \begin{equation}
        \lambda_y = \sum_{d\in\cC}(-1)^{y\cdot d}r(d)
        \, .
    \end{equation}
    On the other hand, expanding \Cref{eq:Ky_simplified}, we obtain
    \begin{align}
        K_y
        &= \sum_{e,e':e+e'\in\cC} \overline{h(e)\eta(Q_e)}h(e')\eta(Q_{e'})(-1)^{(e+e')\cdot y} \\
        &= \sum_{d\in\cC,e}\overline{h(e)\eta(Q_e)}h(e+d)\eta(Q_{e+d})(-1)^{d\cdot y} \\
        &= \lambda_y
        \, .
    \end{align}
    Thus, the random variable $K_y$, for uniformly random $Y\in\ff_2^m$, is distributed as a uniformly random eigenvalue of $G$ (sampled with
  multiplicity).
  In particular,
  \begin{equation}
      \mathbb E_Y[K_Y]
      =
      \frac{1}{\abs{\mathcal{C}}}\tr[G]
      =1
      \, .
  \end{equation}
  The ensemble $\{\ket{\Psi_c}\}_{c\in\cC}$ is geometrically uniform.
  For an equiprobable geometrically uniform pure-state ensemble, the pretty good measurement is optimal \cite{ban_optimum_1997,eldar_2001} and succeeds with probability
  \begin{equation}
      P_{\mathrm{succ}}^{\mathrm{opt}}
      =
      \frac{1}{\abs{\cC}^2}\left(\tr\sqrt G\right)^2
      =
      \left(\mathbb E_Y[\sqrt{K_Y}]\right)^2.
      \label{eq:GU_success_probability}
  \end{equation}
  Writing $\mu\coloneqq\mathbb E_Y[\sqrt{K_Y}]$, Cauchy-Schwarz gives us
  \begin{align}
      \mathbb E_Y[\abs{K_Y-1}]
      &=\mathbb E_Y[
          \abs{\sqrt{K_Y}-1}
          (\sqrt{K_Y}+1)
      ] \\
      &\leq
      \sqrt{
          \mathbb E_Y[(\sqrt{K_Y}-1)^2]
          \mathbb E_Y[(\sqrt{K_Y}+1)^2]
      } \\
      &= \sqrt{(2-2\mu)(2+2\mu)} \\
      &= 2\sqrt{1-P_{\mathrm{succ}}^{\mathrm{opt}}}
      \, .
  \end{align}
    
\end{proof}

\subsection{End-to-end analysis}
To prove that our improved HDQI procedure succeeds when equipped with a quantum decoder with low failure probability, it remains to prove that the term $\alpha$ in \cref{thm:analysis_codeword} is vanishing.
We can do this by relating $\alpha$ to the failure probability of the decoder when the signs $y$ in the Hamiltonian are chosen uniformly randomly.
\begin{lemma}\label{lem:average_alpha_codeword}
    Consider a family of noise functions $g_y:\ff_2^m\rightarrow\mathbb{C}$ of the form
    \begin{equation}
        g_y(e) = (-1)^{e\cdot y}h(e)
    \end{equation}
    for $y\in\ff_2^m$ and some function $h:\ff_2^m\rightarrow\mathbb{C}$ s.t. $\norm{h}_2=1$.
    The associated QDP problem is to retrieve the codeword $c\in\cC$ from the state $\ket{\Psi_c^y}\coloneqq \sum_e g_y(e)\eta(Q_e)\ket{c+e}$.
    Consider a quantum decoding algorithm that can decode the $y=0$ case (i.e. retrieve $c$ from $\ket{\Psi_c^0}$) with failure probability $\pfail$.
    This decoder induces a quantum decoder for arbitrary $y$ through the following procedure:
    \begin{itemize}
        \item Apply $Z^{y}$ on $\ket{\Psi_c^y}$ to obtain $(-1)^{y\cdot c}\ket{\Psi_c^0}$
        \item Apply the $y=0$ decoder to obtain $c$ from $\ket{\Psi_c^0}$.
    \end{itemize}
    Clearly, the failure probability $\pfail$ is then independent of $y$.
    Let $\alpha$ be defined as in \cref{thm:analysis_codeword}.
    Then,
    \begin{equation}
        \mathbb{E}_{Y\sim\mathrm{Uniform}(\ff_2^m)}\left[ \alpha^2 \right]
        = \pfail
        \, .
    \end{equation}
\end{lemma}
\begin{proof}
    We directly compute
    \begin{align}
        \mathbb{E}_{Y\sim\mathrm{Uniform}(\ff_2^m)}\left[ \alpha^2 \right]
        &= \frac{1}{\abs{\cC}}\sum_{\bar c\in\cC} \mathbb{E}_{Y\sim\mathrm{Uniform}(\ff_2^m)}\left[\norm{\sum_{c\neq \bar c} \ket{f_{c,\bar c}}}^2\right] \\
        &= \frac{1}{\abs{\cC}}\sum_{\bar c\in\cC} \mathbb{E}_{Y\sim\mathrm{Uniform}(\ff_2^m)}\left[ \sum_{c_1\neq \bar c,c_2\neq \bar c}\braopket{\Psi_{c_1}^y}{Z^{y}V^\dagger (\proj{\bar c}\otimes\mathds{1}) V Z^{y}}{\Psi_{c_2}^y} \right] \\
        &= \frac{1}{\abs{\cC}}\sum_{\bar c\in\cC} \mathbb{E}_{Y\sim\mathrm{Uniform}(\ff_2^m)}\left[ \sum_{c_1\neq \bar c,c_2\neq \bar c}(-1)^{y\cdot (c_1+c_2)}\braopket{\Psi_{c_1}^{0}}{V^\dagger (\proj{\bar c}\otimes\mathds{1}) V}{\Psi_{c_2}^{0}} \right] \\
        &= \frac{1}{\abs{\cC}}\sum_{\bar c\in\cC} \sum_{c\neq \bar c}\braopket{\Psi_c^{0}}{V^\dagger (\proj{\bar c}\otimes\mathds{1}) V}{\Psi_{c}^{0}} \\
        &= \frac{1}{\abs{\cC}} \sum_{\bar c\in\cC} \pfail \\
        &= \pfail
        \, .
    \end{align}
    where $V$ is the isometry describing the symmetrized decoder for $y=0$.
\end{proof}

\begin{theorem}\label{thm:analysis_commuting}
    Consider a commuting Hamiltonian $\cH$ and some target inverse temperature $\beta\geq 0$.
    Assume we have a polynomial-time classical or quantum decoding algorithm for the associated symplectic code, which fails with probability $\pfail$ when decoding codewords transmitted over $\otimes_{i=1}^m\bsc[p_i]$ or $\otimes_{i=1}^m\psc[p_i]$ respectively, where
    \begin{equation}
        p_i \coloneqq \frac{1}{2}\left(1 - \sech(\beta v_i)\right)
        \, .
    \end{equation}
    Furthermore, assume $1/(1-\pfail)=\mathcal{O}(\mathrm{poly}(m))$.
    Then, there exists an efficient quantum algorithm which prepares a state $\tilde{\rho}$ which is close to the Gibbs state in expectation over a uniformly random choice of signs
    \begin{equation}
        \mathbb{E}_{Y\sim\mathrm{Uniform}(\ff_2^m)}
        \left[ \norm{\tilde{\rho} - \frac{1}{Z_\beta}\exp\left(-\beta \cH\right)}_1 \right]
        \leq 6\sqrt{\pfail}
        \, .
    \end{equation}
\end{theorem}
\begin{proof}
    As seen previously, the function $\tilde f$ fulfills the desired property $\rho_{\tilde f} = (1/Z_\beta)\exp(-\beta\cH)$.
    We apply the QD-HDQI algorithm analyzed in \cref{thm:analysis_codeword} with the normalized noise function $f\coloneqq \tilde f / \norm{\tilde f}_2$.
    As previously discussed, the associated noise state $\sum_ef(e)\ket{e}$ can be efficiently prepared.
    We end up preparing a state $\tilde\rho$ satisfying
    \begin{equation}
        \norm{\tilde\rho - \rho_f}_1
        \leq \pfail + 2\alpha + \alpha^2
        \, .
    \end{equation}
    Furthermore, we have
    \begin{equation}
        \norm{\rho_f - \frac{1}{Z_\beta}\exp\left(-\beta\cH\right)}_1
        = \norm{\rho_f - \rho_{\tilde f}}_1
        = \abs{1 - Z_\beta '}
        \, .
    \end{equation}
    Combining this, we obtain
    \begin{align}
        \mathbb{E}_{Y\sim\mathrm{Uniform}(\ff_2^m)}
        \left[\norm{\tilde\rho - \frac{1}{Z_\beta}\exp\left(-\beta\cH\right)}_1\right]
        &\leq \pfail
        + 2\mathbb{E}_{Y\sim\mathrm{Uniform}(\ff_2^m)}[\alpha] \nonumber \\
        &\qquad
        + \mathbb{E}_{Y\sim\mathrm{Uniform}(\ff_2^m)}[\alpha^2]
        + \mathbb{E}_{Y\sim\mathrm{Uniform}(\ff_2^m)}[\abs{1 - Z_\beta '}] \\
        & \leq \pfail + 2\sqrt{\pfail} + \pfail + 2\sqrt{\pfail} \\
        &\leq 6\sqrt{\pfail}
    \end{align}
    where we used \cref{lem:average_alpha_codeword,lem:average_normalization} as well as Jensen's inequality $\mathbb{E}[\alpha]\leq \sqrt{\mathbb{E}[\alpha^2]}$.
\end{proof}

\begin{corollary}[Average-case achievable temperatures for commuting Hamiltonians]\label{cor:hdqi_end_to_end}
    Consider a family of commuting Hamiltonians where the coefficients are uniformly randomly chosen from $\{+1,-1\}$ (i.e., $v_i=1$ for all $i$ and $y$ is uniformly random).
    Consider a family of polynomial-time (classical or quantum) decoding algorithms for the associated symplectic codes, which exhibits a threshold $p^*$ when transmitted over i.i.d. binary symmetric channels or pure state channels\footnote{Precisely, this means that for any $p<p^*$, the failure probability of the decoder vanishes asymptotically when transmitting over i.i.d. copies of $\bsc[p]$ or $\psc[p]$}.

    In the asymptotic limit, for almost any choice of $y$, there exists an efficient quantum algorithm which can prepare Gibbs states of temperature $\beta$ with vanishing trace distance error for any $\beta\in [0,\beta^*)$ where
    \begin{equation}
        \beta^* = \asech\left(1 - 2p^*\right)
    \end{equation}
    where $\asech$ denotes the inverse hyperbolic secant.
\end{corollary}
The relation between $\beta^*$ and $p^*$ is depicted in \cref{fig:tanh}.
\begin{figure}[ht]
    \centering
    \begin{tikzpicture}
    \begin{axis}[ width=10cm, height=3cm, xlabel={$\beta^*$}, ylabel={$p^*$}, xmin=0, xmax=10, ymin=0, ymax=0.52, grid=major, thick, samples=200, domain=0:10, ytick={0,0.25,0.5} ]
        \addplot[blue] {sinh(x/2)^2 / cosh(x)};
    \end{axis}
    \end{tikzpicture}
    \caption{Relationship between achievable temperature $\beta^*$ and the decoding threshold $p^*$ from \cref{cor:hdqi_end_to_end}.}
    \label{fig:tanh}
\end{figure}
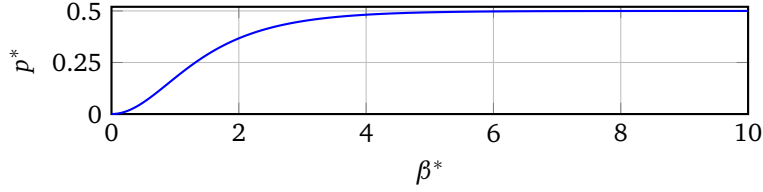

\subsection{Using unambiguous state discrimination as quantum decoder}\label{sec:USD}
Here, we recapitulate one specific strategy to solve the quantum decoding problem, which is due to Chailloux and Tillich \cite{chailloux_quantumdecoding_2023}.
Recall how we previously showed that the quantum decoding problem for a code $\cC$ transmitted over a sequence of pure-state channels $\otimes_i\psc[p_i]$ can be reduced to the decoding problem of the same code over $\otimes_i\bsc[p_i]$.
This is an instance of \emph{channel degradation}: By performing a well-chosen measurement on each individual channel output, the resulting decoding problem now sees some effective classical channel
\footnote{Note that for the purposes of HDQI, the degradation measurement has to be realized coherently through its Stinespring dilation.}.
To degrade from $\psc[p_i]$ to $\bsc[p_i]$, that measurement was simply in the computational basis.
A central insight in \cite{chailloux_quantumdecoding_2023} is that in certain instances, it can be advantageous to instead perform a different channel degradation.
Instead of degrading the pure-state channel to a binary symmetric channel, one can instead degrade it to an erasure channel.
\begin{definition}[Erasure channel]
    We denote by $\bec[\epsilon]$ the erasure channel, which maps $x\in\ff_2$ to $x$ with probability $1-\epsilon$ and to $\perp$ with probability $\epsilon$.
\end{definition}
The measurement to degrade from $\psc$ to $\bec$ is called an \emph{unambiguous state discrimination} (USD).
While the degradation from $\psc$ to $\bec$ can be shown to be strictly worse information theoretically than the degradation from $\psc$ to $\bsc$ (the associated capacity of the $\bec$ is lower than the capacity of the $\bsc$), it has one main advantage: The resulting erasure decoding problem can always be solved optimally in polynomial time.
In particular, the decoding can be done efficiently for any underlying code without assuming any structure.

\paragraph{USD-based degradation}
Consider a pure-state channel $\psc[p]$ for $p\in[0,1/2]$ and denote its two possible outputs by $\ket{\psi_x}\coloneqq \psc[p](x)$ for $x\in\ff_2$.
The USD operates in two steps.
First, an isometry $V$ is applied on the channel output, which is defined as
\begin{align}
    \ket{\psi_0} &\mapsto \sqrt{1-\epsilon}\ket{0} + \sqrt{\epsilon}\ket{\perp} \\
    \ket{\psi_1} &\mapsto \sqrt{1-\epsilon}\ket{1} + \sqrt{\epsilon}\ket{\perp}
\end{align}
for some appropriate $\epsilon\in[0,1]$.
The output space is three-dimensional with orthonormal basis $\{\ket 0, \ket 1, \ket \perp\}$.
For this to be a valid isometry, the angles must be preserved, and as such we obtain the condition
\begin{equation}
    \epsilon = \braket{\psi_0}{\psi_1} = 2\sqrt{(1-p)p}
    \, .
\end{equation}
In a second step, the output of the isometry is measured in its standard basis.
The resulting effective  channel is therefore $\bec[2\sqrt{p(1-p)}]$.

\paragraph{Optimal erasure decoding}
Here, we briefly recapitulate the well known technique for decoding over an erasure channel in polynomial time.

Consider a linear code $\cC\subset\ff_2^m$ with parity check matrix $H$, and let $x\in\cC$ be a codeword that is transmitted over a sequence of erasure channels.
Denote by $S\subset [m]$ the indices of the codeword bits which have been erased, and let $E\coloneqq\abs{S}$.
For notational convenience, let's re-arrange the codeword bits such that the first $m-E$ bits are not erased, and the last $E$ bits are erased.
Correspondingly, we separate $x$ into $(x_A, x_B)$ where $x_A\in\ff_2^{m-E}$ and $x_B\in\ff_2^E$.
Similarly, we separate the parity check matrix into two parts
\begin{equation}
    H = \begin{pmatrix}
        H_A & H_B
    \end{pmatrix}
    \, .
\end{equation}
Recall that $Hx=0$.
Therefore, the erased part of the codeword $x_B$ can be recovered if and only if the following set of equations has a unique solution
\begin{equation}
    H_B x_B = H_A x_A
    \, .
\end{equation}
In case the solution is unique, it can be recovered in polynomial time using Gaussian elimination.

\paragraph{Achievable temperature}
Consider a family of codes that exhibits a threshold $\epsilon^*$ when transmitted over independent and identical copies of the erasure channel and decoded with Gaussian elimination.
Then, USD-based decoding of the same code family over the pure state channel achieves a threshold of
\begin{equation}
    p^* = \frac{1}{2}\left( 1 - \sqrt{1-(\epsilon^*)^2} \right)
    \, .
\end{equation}
Combined together with \cref{cor:hdqi_end_to_end}, we arrive at following expression for the achievable temperature.
\begin{corollary}[Temperatures achieved by QD-HDQI+USD]\label{corr:hdqiusdthreshold}
    Consider a family of commuting Hamiltonians where the coefficients are uniformly randomly chosen from $\{+1,-1\}$ (i.e., $v_i=1$ for all $i$ and $y$ is uniformly random).
    Assume that the associated symplectic codes exhibit a threshold $\epsilon^*$ when transmitted over the erasure channel and using Gaussian elimination decoding.

    In the asymptotic limit, for almost any choice of $y$, there exists an efficient quantum algorithm which can prepare Gibbs states of temperature $\beta$ with vanishing trace distance error for any $\beta\in [0,\beta^*)$ where
    \begin{equation}
        \beta^* = \atanh(\epsilon^*)
        \, .
    \end{equation}
\end{corollary}

\subsection{Example: Dilute classical spin glasses}\label{sec:classical_spin_glass}
We now apply the previous construction in this section to an explicit family of commuting Pauli Hamiltonians.
Specifically, we will consider $n$-qubit Hamiltonians of the form 
\begin{equation}
    \cH_{A,y}
    \coloneqq
    \sum_{i=1}^m(-1)^{y_i}Z^{a_i}
\end{equation}
where $y\in\ff_2^m$ and $a_i\in\ff_2^{n}$ is the $i$th column of some sparse matrix $A\in\ff_2^{n\times m}$.
The symplectic code in this case is simply the LDPC code $\cC=\ker(A)$.
Our formalism can also easily capture the case where the couplings are non-uniform, but we do not pursue this generalization here for simplicity.

This family of Hamiltonians is a particularly well suited example to highlight the improvements that quantum decoders can provide in HDQI, because there has been significant progress in recent years to solve the quantum decoding problem over LDPC codes through the quantum algorithm called \emph{belief propagation with quantum messages} (BPQM) \cite{renes_belief_2017,rengaswamy_bpqm_2021,piveteau_bpqm_2022,brandsen_bpqm_2022,mandal_bpqmqary_2026,mandal_bpqmabelian_2026,mandal_bpqmblock_2026}.
In this section, we will show that our algorithm paired with the BPQM decoder can outperform any classical decoder.
That is, even if we ignore computational restrictions and pick the information theoretic best classical decoder, it remains impossible to achieve temperatures as low as QD-HDQI paired with BPQM.

Concretely, we will consider the matrix $A$ to be sampled from a standard regular LDPC ensemble, as for example defined in \cite{richardson_modern_2008}.
This allows us to analyze the  asymptotic performance of BP and BPQM using a technique called \emph{density evolution}.
We refer the interested reader to \cite{mandal_bpqmblock_2026} for a detailed introduction of density evolution for BPQM, and we only briefly note that density evolution is a set of iterative equations that describes the distribution of messages transmitted by BPQM over the tree neighborhood of a random codeword bit.
These recursive equations typically exhibit a threshold behavior, where below a certain noise rate, the distribution converges to a point mass distribution of ideal noise-free messages that allow for perfect information recovery.
For our purposes here, it suffices to mention that the BPQM threshold can be estimated numerically and that this threshold precisely characterizes the region where BPQM is successful.
\begin{theorem}[{{\cite[Theorem 30]{mandal_bpqmblock_2026}}} specialized to $q=2$]\label{thm:bpqm_block_analysis}
    Consider an $(\ell,r)$-regular LDPC ensemble that is transmitted over $\psc[p]$ for $\ell\geq 3$.
    If $p$ is smaller than the density evolution threshold $p_{\mathrm{DE}}$, then BPQM can solve this quantum decoding problem with vanishing failure probability.
\end{theorem}
In \cref{tab:classical_spin_glass_results}, we summarize the temperatures achieved by our algorithm paired with classical BP and BPQM for various LDPC ensembles.
The threshold values were obtained using a numerical implementation of density evolution based on population dynamics Monte Carlo, as described in \cite{brandsen_bpqm_2022}.
We observe that for sparser LDPC ensembles, the estimated BPQM threshold exceeds the Shannon limit for classical decoding, allowing QD-HDQI to prepare Gibbs states at inverse temperatures that no classical decoder can reach.
In contrast, for denser LDPC ensembles QD-HDQI with BPQM still significantly outperforms QD-HDQI with classical BP, but the information theoretic advantage disappears.
This can be attributed to the fact that the performance of iterative message passing decoder deteriorates as the average degree of a Tanner graph increases.

\section{Classical algorithms for Gibbs sampling from commuting Hamiltonians}\label{sec:dequantization}
When $\calH = \sum_{i=1}^m (-1)^{y_i}v_iP_i$, with $v_i>0$,
is commuting, the factorization
\begin{align}
    \rho_{\beta}(\calH)
    &= \frac{\prod_{i=1}^m\cosh(\beta v_i)}{Z_\beta}
    \prod_{i=1}^m \left[
        (1-\tanh(\beta v_i))\mathbb{I}
        + \tanh(\beta v_i)(\mathbb{I}-(-1)^{y_i}P_i)
    \right] \\
    &= \frac{\prod_{i=1}^m\cosh(\beta v_i)}{Z_\beta}
    \sum_{S\subseteq[m]}
    \left(\prod_{i\notin S}(1-\tanh(\beta v_i))\right)
    \left(\prod_{i\in S}\tanh(\beta v_i)\right)
    \prod_{i\in S}(\id-(-1)^{y_i}P_i)
    \label{eq:factorization}
\end{align}
in fact also leads to a classical simulation of a Gibbs sampler for this
Hamiltonian. Our algorithm's output takes the form of a classical description of a stabilizer state
sampled from the ensemble $\rho_{\beta}(\calH)$, whereas the quantum algorithms in the previous section prepare the Gibbs state coherently. The latter output is strictly more powerful, as it allows for efficient estimation of non-Clifford observables, while having Gibbs sampling access to the (stabilizer) eigenbasis does not generally allow this. 

We present our classical simulation algorithm in this section, noting that it achieves exactly the same temperature for Gibbs sampling as QD-HDQI using unambiguous state discrimination as a decoder (\Cref{sec:USD}).

\subsection{Reduction to correlated sampling of stabilizer states}

We present the core idea of the classical simulation. Inspecting \Cref{eq:factorization}, first observe that it is a sum over subsets $S$ of Paulis in the Hamiltonian, where the factor
\begin{equation}
    q_\beta(S):=\prod_{i\notin S}(1-t_i)\prod_{i\in S}t_i
\end{equation}
is the probability of Bernoulli sampling the subset via including each term $i$ in $S$ into the subset with independent probability
$t_i := \tanh(\beta v_i)$. 
The remaining part of the expression,
$\frac{\id - (-1)^{y_i}P_i}{2}$, takes a familiar form as the projector onto the $-1$ eigenspace of the Pauli $(-1)^{y_i}P_i$. Hence, for a sampled subset of terms $S$,
the projector onto the simultaneous $-1$ eigenspace of all Paulis $(-1)^{y_i}P_i$, $i\in S$,
which we will call $C_{S,y}$, is
\begin{equation}
    \Pi_{S,y} := \prod_{i\in S} \frac{\mathbb{I} -(-1)^{y_i}P_i}{2}
\end{equation}
(this still holds even if the eigenspace is null).

The expression
\begin{align}
    \rho_{\beta}(\calH)
    &= \frac{\prod_{i=1}^m\cosh(\beta v_i)}{Z_\beta}
    \sum_{S\subseteq[m]}q_\beta(S)2^{|S|}\Pi_{S,y}\\
    &= \frac{\prod_{i=1}^m\cosh(\beta v_i)}{Z_\beta}
    \sum_{\substack{S\subseteq[m]\\\mathsf{Tr}(\Pi_{S,y})>0}}
    q_\beta(S)2^{|S|}\mathsf{Tr}(\Pi_{S,y})\rho_{S,y}
    \label{eq:decomposition}
\end{align}
where we have defined the properly normalized density matrix
\begin{equation}\label{eq:rhoS}
    \rho_{S,y} := \frac{\Pi_{S,y}}{\mathsf{Tr}(\Pi_{S,y})},
    \qquad \mathsf{Tr}(\Pi_{S,y})>0,
\end{equation}
is a convex combination of density matrices, which can be sampled from
by sampling every density matrix with the appropriate probability, and
then sampling a basis state from each density matrix with probability
given by its spectrum.

This spectrum is particularly simple in our case -- $\rho_{S,y}$ is simply
the maximally-mixed state over $C_{S,y}$. A natural basis for this ensemble
is a simultaneous Pauli eigenbasis consisting of stabilizer states, each with
an efficient classical description. Were the $2^{|S|}\Tr(\Pi_{S,y})$ not
part of the expression \Cref{eq:decomposition}, we could sample from
$\rho_{\beta}$ in two steps:
\begin{enumerate}
    \item Bernoulli-sample $S$ by including every term $i$ in
    $\mathcal{H}$ into $S$ independently with probability $t_i.$
    \item Sample a stabilizer state uniformly from a stabilizer basis of $C_{S,y}$ and output
    its classical description.
\end{enumerate}

Now we must modify this strategy to account for the fact that
$2^{|S|}\Tr(\Pi_{S,y})$ actually depends on the Pauli subset $S$ (proven
shortly in \Cref{lem:dimgood}). So, in sampling $S$ (Step 1), we must
correct the Bernoulli weight assigned to $S$ with the multiplicative
factor $2^{|S|}\Tr(\Pi_{S,y})$. We now calculate exactly what that
multiplicative factor is.

\begin{lemma}[Dimension of the good eigenspace or $\Tr(\Pi_{S,y})$]
\label{lem:dimgood}
Call $S = \{i_1, \ldots, i_{|S|}\}\subseteq[m]$
{\em independent} iff
\[
\{\symp(P_{i_1}), \ldots, \symp(P_{i_{|S|}})\}
\]
is linearly independent over $\mathbb{F}_2$. Call $S$
{\em consistent} iff there exists a state $\ket{\psi}$ that
simultaneously satisfies
\[
(-1)^{y_i}P_i|\psi\rangle=-|\psi\rangle
\qquad
\text{for every }i\in S.
\]

\begin{itemize}
    \item If $S$ is an independent set then
    $\Tr(\Pi_{S,y}) = \dim(C_{S,y}) = 2^{n-|S|}$.

    \item If $S$ is dependent but inconsistent, $\Pi_{S,y}$ is the
    projector onto the null subspace, i.e. the zero matrix.

    \item If $S$ is dependent and consistent, $\Pi_{S,y}$ can be
    re-written in terms of an independent generating set of $r$
    Paulis, and so
    $\Tr(\Pi_{S,y}) = \dim(C_{S,y}) = 2^{n-\text{rank}(H_S)}$,
    where $H_S$ is defined analogously to \Cref{eq:symplectic_matrix} i.e.
    \begin{equation}
        H_S \coloneqq
        \begin{pmatrix}
            \vert & \vert & & \vert \\
            \symp(P_{i_1})
            & \symp(P_{i_2})
            & \cdots
            & \symp(P_{i_{|S|}}) \\
            \vert & \vert & & \vert
        \end{pmatrix}
    \end{equation}

    In the case that $S$ is an independent set,
    $\text{rank}(H_S) = |S|.$
\end{itemize}
\end{lemma}

\begin{proof} Recall that
\[
C_{S,y}
:=
\left\{
|\psi\rangle:
(-1)^{y_i}P_i|\psi\rangle=-|\psi\rangle
\text{ for every }i\in S
\right\}.
\]

We now consider the three cases separately.

\paragraph{Case 1: $S$ is independent.}
Suppose $S=\{i_1,\ldots,i_{|S|}\}$
is independent. Then no nonempty product
of the $(-1)^{y_{i_j}}P_{i_j}$ can be proportional to the identity:
indeed, for any nonzero $a\in\mathbb F_2^{|S|}$,
\[
\symp\left(
\prod_{j=1}^{|S|}
P_{i_j}^{a_j}
\right)
=
\sum_{j=1}^{|S|}
a_j\symp(P_{i_j})
\neq 0.
\]
Consequently, every such nonempty product is a nonidentity Pauli
operator and hence has trace zero.

Expanding the projector gives
\[
\begin{aligned}
\Pi_{S,y}
&=
2^{-|S|}
\prod_{j=1}^{|S|}
\left(I-(-1)^{y_{i_j}}P_{i_j}\right)=
2^{-|S|}
\sum_{a\in\mathbb F_2^{|S|}}
(-1)^{|a|}
\prod_{j=1}^{|S|}
\left((-1)^{y_{i_j}}P_{i_j}\right)^{a_j}.
\end{aligned}
\]
Upon taking the trace, all terms with $a\neq 0$ vanish, while the
$a=0$ term is the identity and has trace $2^n$. Therefore
\[
\Tr(\Pi_{S,y})
=
2^{-|S|}2^n
=
2^{n-|S|}\] and the conclusion follows.

\paragraph{Case 2: $S$ is dependent and inconsistent.}
By definition, inconsistency means that the common $-1$ eigenspace
$C_{S,y}$ is the zero subspace. Since $\Pi_{S,y}$ is precisely the orthogonal projector
onto $C_{S,y}$, it follows that $\Pi_{S,y}=0.$

\paragraph{Case 3: $S$ is dependent but consistent.}
Let $r:=\mathrm{rank}(H_S)$, and choose $r$ terms
\[
(-1)^{y_{j_1}}P_{j_1},\ldots,(-1)^{y_{j_r}}P_{j_r}
\]
with indices in $S$ whose symplectic vectors form a basis for the column space of
$H_S$. Denote this independent subset by
\begin{equation}\label{eq:B}
J:=\{j_1,\ldots,j_r\}.
\end{equation}

Every remaining term $(-1)^{y_i}P_i$, $i\in S\setminus J$, is dependent on this
basis. Hence there exist coefficients
$c_{i1},\ldots,c_{ir}\in\mathbb F_2$ and a sign
$\sigma_i\in\{\pm1\}$ such that
\[
(-1)^{y_i}P_i
=
\sigma_i
\prod_{\ell=1}^r
\left((-1)^{y_{j_\ell}}P_{j_\ell}\right)^{\,c_{i\ell}}.
\]
Now consider any state $|\psi\rangle$ satisfying the $r$ independent
constraints
\[
(-1)^{y_{j_\ell}}P_{j_\ell}|\psi\rangle=-|\psi\rangle,
\qquad \ell=1,\ldots,r.
\]
On such a state,
\[
(-1)^{y_i}P_i|\psi\rangle
=
\sigma_i
(-1)^{\sum_{\ell}c_{i\ell}}
|\psi\rangle.
\]
Because the full system $S$ is consistent, the dependent constraint
$(-1)^{y_i}P_i|\psi\rangle=-|\psi\rangle$ must already be implied by
the independent constraints. Therefore
\[
\sigma_i
(-1)^{\sum_{\ell}c_{i\ell}}
=
-1.
\]
Thus every state satisfying the constraints in $J$ automatically
satisfies every constraint in $S$. Conversely, since $J\subseteq S$,
every state satisfying all constraints in $S$ certainly satisfies
those in $J$. Hence the two common eigenspaces coincide:
\[
C_{S,y}=C_{J,y},
\]
and therefore
\[
\Pi_{S,y}
=
\prod_{\ell=1}^r
\frac{I-(-1)^{y_{j_\ell}}P_{j_\ell}}{2}.
\]

The operators indexed by $J$ are independent, so by the argument from the
first case,
\[
\Tr(\Pi_{S,y})
=
\Tr(\Pi_{J,y})
=
2^{n-r}.
\]
The claim follows by observing that $r=\text{rank}(H_S)$.
\end{proof}

Returning to the question of how to sample from $\rho_{\beta}$,
\Cref{lem:dimgood} gives a more explicit re-writing of
$\rho_{\beta}$ as a convex mixture of states:

\begin{equation}\label{eq:targetclassical}
\boxed{
\begin{aligned}
    \rho_{\beta}(\calH)
    &= \frac{\prod_{i=1}^m\cosh(\beta v_i)}{Z_\beta}
    \sum_{\substack{S\subseteq[m]\\\Tr(\Pi_{S,y})>0}}
    q_\beta(S)
    2^{|S|}
    2^{n-\operatorname{rank}(H_S)}
    \rho_{S,y}.
\end{aligned}
}
\end{equation}

Next, we will develop a classical algorithm to approximately sample
stabilizer states from this distribution.

\subsection{Naive rejection sampling dequantizes QD-HDQI+USD}
We will now introduce an approximate classical sampler that is efficient in the regime where QD-HDQI + USD is effective. This sampler is simple to describe conceptually: in \Cref{eq:targetclassical}, sampling $S$ in a Bernoulli fashion does not correctly implement the factor $2^{|S|}2^{n-\operatorname{rank}(H_S)}$. A simple way to avoid this problem is to only retain a subset of the $S$s in \Cref{eq:targetclassical} for which this factor is constant, for within this subset, we can sample in a Bernoulli fashion. The approximation error of this algorithm is then caused by discarding all other $S$s. 

The special subset of $S$'s we will retain is the set of $S$ consisting of independent stabilizers. Whenever this condition is satisfied, it is also particularly straightforward to implement Step 2, namely sampling a state from the stabilizer decomposition of $\rho_{S,y}$. 

We translate this sampling task into a linear algebra problem. In $\calH = \sum_{i=1}^m (-1)^{y_i}v_i P_i $, since the Pauli operators $\{P_i\}_{i=1}^m$ commute, we may simultaneously diagonalize them using an efficiently computable Clifford $C$. In particular, we can find, independently of $y$,

\begin{equation}
    r_i \in\mathbb F_2,
    \qquad
    g_i\in\mathbb F_2^n,
\end{equation}

such that
\begin{equation}\label{eq:garaya}
C(-(-1)^{y_i}P_i)C^\dagger
=
(-1)^{r_i\oplus y_i}Z^{g_i}.
\end{equation}
\begin{lemma}[Affine condition on $g_i, r_i, y_i$ for commuting Pauli Hamiltonians]
Let $g_i, r_i, y_i$ be as defined in \Cref{eq:garaya}. For $S\subseteq[m]$, the state $C^\dagger\ket{x}$ belongs to $C_{S,y}$ if and only if $x$ fulfils
    \begin{equation}
g_i\cdot x=r_i\oplus y_i
\qquad
\text{for every }i\in S.
\end{equation}
\end{lemma}
\begin{proof}
The stabilizer space associated with $S$ consists of the simultaneous $+1$ eigenstates of
\[
\{-(-1)^{y_i}P_i\}_{i\in S}.
\]
Hence any such eigenstate $\ket{\lambda}$ satisfies, for every $i\in S$,
\begin{equation}
-(-1)^{y_i}P_i\ket{\lambda}
=\ket{\lambda}
\nonumber
\Longrightarrow\qquad
Z^{g_i}C\ket{\lambda}
=
(-1)^{r_i\oplus y_i}C\ket{\lambda}.
\end{equation}
Since all of the operators on the left-hand side are $Z$-type Paulis, their simultaneous eigenspaces admit computational-basis eigenbases. Thus we may choose $\ket{\lambda}$ so that

$$
    C\ket{\lambda}=\ket{x},
    \qquad x\in\mathbb F_2^n,
$$

where $x$ must satisfy

$$
    (-1)^{g_i\cdot x}
    =
    (-1)^{r_i\oplus y_i},
    \qquad i\in S.
$$
\end{proof}

We may collect all these constraints into a matrix equation: 
define \begin{equation}\label{eq:defG} B := \begin{pmatrix} \text{---} & g_1^{\mathsf T} & \text{---} \\ \text{---} & g_2^{\mathsf T} & \text{---} \\ & \vdots & \\ \text{---} & g_m^{\mathsf T} & \text{---} \end{pmatrix} \in \mathbb{F}_2^{m\times n}, \qquad \mathbf r := \begin{pmatrix} r_1 \\ r_2 \\ \vdots \\ r_m \end{pmatrix} \in \mathbb{F}_2^m, \qquad y := \begin{pmatrix} y_1 \\ y_2 \\ \vdots \\ y_m \end{pmatrix} \in \mathbb{F}_2^m. \end{equation}

For a subset $S\subseteq[m]$, let $B_S$ denote the submatrix containing only the rows indexed by $S$, and similarly let $\mathbf r_S$ and $y_S$ denote the corresponding restrictions of $\mathbf r$ and $y$. Sampling uniformly from the stabilizer decomposition of $\rho_{S,y}$ therefore reduces to sampling uniformly from the solution set
\begin{equation}
B_Sx=\mathbf r_S\oplus y_S.
\end{equation}

The stabilizers in $S$ are independent if and only if $\operatorname{rank}(B_S)=\operatorname{rank}(H_S)=|S|$. In particular, the above affine system of linear equations is then guaranteed to be consistent and has exactly $2^{n-|S|}$ solutions. These solutions can be sampled uniformly using Gaussian elimination: we find an invertible $T_S$ such that $B_ST_S=(I_{|S|}\;0)$, fix the first $|S|$ coordinates of $z$ to $\mathbf r_S\oplus y_S$, sample the remaining $n-|S|$ coordinates uniformly from $\mathbb F_2$, and transform back via $x=T_Sz$. This samples a satisfying $x$. Finally, we output the classical description of $C^\dagger\ket{x}.$ 

Now let us return to the question of how to implement Step 1 of sampling $S$. We notice from Lemma \ref{lem:dimgood}, that whenever $S$ is independent, the reweighting factor in \Cref{eq:targetclassical} collapses to $2^{|S|}2^{n-|S|} = 2^n$. This does not depend on $S$ anymore. So, if we restrict the sum over $S$ in \Cref{eq:targetclassical} to only the independent $S$, the relative probability of $S$ is simply given by the Bernoulli weight assigned to $S$. In essence, restricting $S$ to the full-rank set 
\begin{equation}
    \mathcal{G}:= \{S \subseteq [m]: \text{rank}(B_S) = |S|\}
\end{equation}
has obviated the issue with Bernoulli sampling caused by differing normalizations of $S$. This suggests that we can approximate $\rho_{\beta}$ by simply sampling $S$ in the Bernoulli fashion described earlier, and rejecting whenever $S$ is not in $\mathcal{G}.$ Then, conditioned on sampling $S$, we implement Step 2.

The above algorithm is written out more formally as \Cref{alg:classical-rejection-sampler}. To aid the exposition, let us properly define the distribution from which we sample $S$ as the product Bernoulli measure $\operatorname{Ber}(\mathbf{t})$ where 
\begin{equation}\label{eq:prodBer}
    \underset{\operatorname{Ber}(\mathbf{t})}{\operatorname{Pr}}[S] :=\prod_{i \in S} t_i \prod_{i \notin S}\left(1-t_i\right)
\end{equation}
and 
\[
        t_i:=\tanh(\beta v_i).
    \]

\begin{algorithm}[t]
\caption{Classical rejection sampler}
\label{alg:classical-rejection-sampler}
\begin{algorithmic}[1]

\Require A commuting Pauli Hamiltonian
\[
    \calH=\sum_{i=1}^m (-1)^{y_i} v_i P_i,
    \qquad v_i>0,
\]
and inverse temperature $\beta>0$.

\State Simultaneously diagonalize the Pauli operators:
find a Clifford $C$ and vectors
$g_i\in\mathbb{F}_2^n$, $r_i\in\mathbb{F}_2$ such that
\[
    C(-P_i)C^\dagger = (-1)^{r_i} Z^{g_i},
    \qquad i\in[m].
\]
Let $B\in\mathbb{F}_2^{m\times n}, \mathbf r\in \mathbb{F}_2^m$ be as given in \Cref{eq:defG}.

\Repeat
    \State Sample $S\subseteq[m]$ from $S\leftarrow \operatorname{Ber}(\mathbf{t})$ given in \Cref{eq:prodBer}.
    \State Form the submatrix $B_S$ consisting of the rows
    indexed by $S$.
\Until{$\operatorname{rank}(B_S)=|S|$}

\State Define
\[
    b_S := \mathbf r_S\oplus y_S \in\mathbb{F}_2^{|S|}.
\]

\State Using Gaussian elimination, compute an invertible matrix
$T_S\in\operatorname{GL}(n,\mathbb{F}_2)$ such that
\[
    B_ST_S=
    \begin{pmatrix}
        I_{|S|} & 0
    \end{pmatrix}.
\]

\State Sample $u \leftarrow \mathbb{F}_2^{\,n-|S|}$
uniformly at random, and set
\[
    z :=
    \begin{pmatrix}
        b_S\\
        u
    \end{pmatrix}
    \in\mathbb{F}_2^n.
\]

\State Set
\[
    x:=T_S z.
\]

\State \Return a classical description of the stabilizer state $C^\dagger\ket{x}$.

\end{algorithmic}
\end{algorithm}

We analyze the correctness of \Cref{alg:classical-rejection-sampler}:
\begin{theorem}[Expected trace distance error of \Cref{alg:classical-rejection-sampler}'s output]\label{theorem:trace_distance_error}
    In $\calH=\sum_{i=1}^m (-1)^{y_i} v_i P_i$, and for a target temperature $\beta$, \Cref{alg:classical-rejection-sampler} outputs an ensemble with average state 
    \begin{equation}\label{eq:alg1output}
        \sigma_{\beta,y} := \frac{1}{\Pr_{S\leftarrow \operatorname{Ber}(\mathbf{t})}(S\in \mathcal{G})} \sum_{S\in \mathcal{G}} \prod_{i \in S} t_i \prod_{i \notin S}\left(1-t_i\right)\rho_{S,y}
    \end{equation}
    where $t_i =\tanh(\beta v_i)$ and
    \begin{equation}
        \rho_{S,y} := \prod_{i\in S} \frac{\mathbb{I} - (-1)^{y_i}P_i}{2}/\Tr\left(\prod_{i\in S} \frac{\mathbb{I} -(-1)^{y_i}P_i}{2}\right).
    \end{equation}
    Averaging over a uniformly random choice
    $y\in\ff_2^m$, this output state has a trace distance from the ideal output state $\rho_{\beta,y}:=\rho_\beta(\calH)$ of
    \begin{equation}
        \mathbb{E}[\lVert \rho_{\beta,y} -\sigma_{\beta,y}\rVert_1] \leq 2 p_{\mathrm{dep},\mathbf{t}},
    \end{equation}
    where $p_{\mathrm{dep},\mathbf{t}}$ depends on $\calH$ and is the probability that the Bernoulli-sampled subset of Pauli terms $S\leftarrow \operatorname{Ber}(\mathbf{t})$ is dependent,
    \begin{equation}
        p_{\mathrm{dep},\mathbf{t}} := \Pr_{S\leftarrow \operatorname{Ber}(\mathbf{t})}[\text{rank}(B_S) \neq |S|] .
    \end{equation}
\end{theorem}
\begin{proof}
Define $q_{\beta}(S) = \prod_{i \in S} t_i \prod_{i \notin S}\left(1-t_i\right) =: \Pr_{S\leftarrow \operatorname{Ber}(\mathbf{t})}(S).$ The ideal state (\Cref{eq:targetclassical}) is $\rho_{\beta,y}$ where
\begin{equation}\label{eq:target_q}
    \rho_{\beta,y} \propto \sum_{S\subseteq [m]}
    q_{\beta}(S)
    2^{|S|-n}
    \Pi_{S,y}.
\end{equation}
We will now divide the terms into the `Good' subset $\mathcal{G} := \{S \,\,s.t.\,\,\text{rank}(B_S) = |S|\}$ and the complementary `Rejected' subset. These correspond exactly to the $S$ that are accepted by \Cref{alg:classical-rejection-sampler} and those that are rejected. The density matrix proportional to the sum of `good' terms is
\begin{equation}
    \sigma_{G,y} := \frac{1}{a} \sum_{S\in \mathcal{G}} q_{\beta}(S)\rho_{S,y}.
\end{equation}
where $a$ is the acceptance probability, $a:=\Pr_{S\leftarrow \operatorname{Ber}(\mathbf{t})}(S\in \mathcal{G})=1-p_{\mathrm{dep},\mathbf{t}}>0.$
$\sigma_{G,y}$ is exactly the average state that is output by \Cref{alg:classical-rejection-sampler}. The normalization may be verified by observing that $q_{\beta}(S)$ forms a proper probability distribution. The part of \Cref{eq:target_q} corresponding to the `rejected' terms is 
\begin{equation}\label{def:Ry}
    R_y := \sum_{S \not\in \mathcal{G}} q_{\beta}(S) 2^{|S|-n} \Pi_{S,y}
\end{equation}
and the density matrix proportional to this operator
\begin{equation}
    \sigma_{R,y} :=
    \begin{cases}
        R_y/\Tr(R_y), & \Tr(R_y)>0,\\
        2^{-n}\id, & \Tr(R_y)=0.
    \end{cases}
\end{equation}
The normalization of \Cref{eq:target_q} is $\Tr(\sum_{S\subseteq [m]}
    q_{\beta}(S)
    2^{|S|-n}
    \Pi_{S,y}) = a + \Tr(R_y)$, so that
 the target state may be decomposed as
\begin{equation}
    \rho_{\beta, y} = \frac{a}{a+ \Tr(R_y)} \sigma_{G,y}  + \frac{\Tr(R_y)}{a+ \Tr(R_y)} \sigma_{R,y}.
\end{equation}
Then, the error of the above algorithm is 
\begin{align}
   \mathbb{E}_y \lVert \rho_{\beta, y} - \sigma_{\beta, y} \rVert_1 &= \mathbb{E}_y \left\rVert \left(\frac{a}{a+ \Tr(R_y)}-1\right) \sigma_{G,y}  +    \frac{\Tr(R_y)}{a+ \Tr(R_y)} \sigma_{R,y}\right \rVert_1 \\
    &\leq \mathbb{E}_y\left[\frac{2\Tr(R_y)}{a+ \Tr(R_y)} \right]\\
    &\leq \frac{2\mathbb{E}_y[\Tr(R_y)]}{a+\mathbb{E}_y[\Tr(R_y)]}.
\end{align}
The last inequality is Jensen's inequality for the concave function $u\mapsto u/(a+u)$ on $u\geq0$.
But writing out in full the definition of the projector in \Cref{def:Ry} we have, 
\begin{align}
2\mathbb{E}_y[\Tr(R_y)] &= 2\mathbb{E}_y[\sum_{S \not\in \mathcal{G}} q_{\beta}(S) 2^{|S|-\text{rank}(H_S)} \mathbb{I}(S \text{ consistent for } y)]\\
&=2\sum_{S \not\in \mathcal{G}} q_{\beta}(S) 2^{|S|-\text{rank}(H_S)} \Pr(S \text{ consistent for } y) \\
&=2\sum_{S \not\in \mathcal{G}} q_{\beta}(S) 2^{|S|-\text{rank}(H_S)} \frac{2^{\text{rank}(H_S)}}{2^{|S|}}\\
&= 2\sum_{S \not\in \mathcal{G}} q_{\beta}(S) \\
&= 2p_{\mathrm{dep},\mathbf{t}}.
\end{align} 

This proves the claim, since $a=1-p_{\mathrm{dep},\mathbf{t}}$.
\end{proof}

Now, we analyze the runtime of \Cref{alg:classical-rejection-sampler}.
\begin{theorem}[Expected runtime of \Cref{alg:classical-rejection-sampler}]\label{thm:expruntimerejectionsampler}
For a fixed Hamiltonian $\calH$, the pre-processing time required for \Cref{alg:classical-rejection-sampler} is $O(mn^2+n^3)$ and, after pre-processing, the expected runtime of \Cref{alg:classical-rejection-sampler} (over the sampling of $S$ in Line 3) is 
\begin{equation}
\mathbb{E}[T] = O\left(\frac{m+n^3}{1-p_{\mathrm{dep},\mathbf{t}}}\right).
\end{equation}
\end{theorem}
\begin{proof}
The $O(mn^2+n^3)$ pre-processing time arises as follows: first, we find a simultaneous basis for all the Paulis in $\calH$, which can be achieved by writing out $\bigl(\symp(P_1)\;\cdots\;\symp(P_m)\bigr) \in \mathbb{F}_2^{2n\times m}$ and performing Gaussian elimination on it ($O(mn^2)$ time). Then, we find a diagonalizing Clifford using stabilizer tableau methods in $\Theta(n^3)$ time. 

    Recall that 
    \begin{equation}
        p_{\mathrm{dep},\mathbf{t}} := \Pr_{S\leftarrow \operatorname{Ber}(\mathbf{t})}[\text{rank}(B_S) \neq |S|] =\sum_{S \not\in \mathcal{G}} q_{\beta}(S).
    \end{equation}
    So, the number of times we have to sample before we draw some $S$ such that $\text{rank}(B_S) = |S|$ and we accept is distributed as a geometric random variable whose expectation is the inverse of the success probability, $1/(1-p_{\mathrm{dep},\mathbf{t}})$. 

    The runtime for sampling any candidate $S$ is $O(m)$, for example by scanning over all the terms indexed by $1,\ldots, m$ and Bernoulli sampling with the correct bias for each. 
    
    After sampling, we further have to compute what the rank of $B_S$ is to decide whether to accept or reject. In fact, we may immediately reject if $|S|>n$ because $B_S \in \mathbb{F}_2^{|S|\times n}$ cannot possibly have rank greater than $n$. If $|S|\leq n$, computing its rank can be done by Gaussian elimination in time $O(n^3).$

    Thus, the expected total runtime to generate one accepted sample is $O\left(\frac{m+n^3}{1-p_{\mathrm{dep},\mathbf{t}}}\right)$.
    
    Finally, conditioned on sampling a good $S$, we have to perform Gaussian elimination on an $n\times n$ matrix and sample $u \in \mathbb{F}_2^{n-|S|}$ uniformly; these take time $O(n^3)$ and $O(n)$ respectively, so we absorb them into the existing term. 
\end{proof}
\subsubsection{Dequantization of QD-HDQI with USD}

We now observe that \Cref{alg:classical-rejection-sampler} dequantizes QD-HDQI utilizing USD as a quantum decoder. This is because $p_{\mathrm{dep},\mathbf{t}}$, which characterizes the trace-distance error of \Cref{alg:classical-rejection-sampler} in \Cref{theorem:trace_distance_error}, can be related exactly to the failure probability of QD-HDQI utilizing USD. Let us now explain why. 

We work in the setting of \Cref{corr:hdqiusdthreshold} which gives guarantees for our QD-HDQI+USD algorithm; i.e. we restrict to the case where the absolute values of the Hamiltonian coefficients satisfy $v_i = 1$, so that we can speak of a threshold for the binary erasure channel (BEC) induced by the USD decoder. With this restriction, let us consider a family of commuting Pauli Hamiltonians whose associated symplectic codes have BEC threshold $\epsilon^\ast$. The Hamiltonian signs do not affect the symplectic code, and hence do not affect this threshold. The threshold means that, for every fixed erasure probability $\epsilon<\epsilon^\ast$, the BEC decoding failure probability vanishes as the blocklength $m\to\infty$ along this code family. Since USD induces erasure probability $\epsilon=\tanh\beta$, \Cref{corr:hdqiusdthreshold} therefore gives the temperature threshold
\[
\beta^\ast=\tanh^{-1}(\epsilon^\ast).
\]

As explained before \Cref{corr:hdqiusdthreshold}, erasure decoding can be implemented optimally for unique recovery using Gaussian elimination. Let the erasure probability of the BEC be $t$. Given an erased set $S$, the erased coordinates are uniquely recoverable if and only if the corresponding columns of the parity-check matrix are linearly independent. Hence the failure probability of this decoder is
\begin{equation}
p_{\mathrm{fail}}^{\mathrm{BEC}(t)}
=
\Pr_{S\leftarrow \operatorname{Ber}(t\mathbf{1})}
\left[\operatorname{rank}(H_S)<|S|\right].
\end{equation}
This is exactly the quantity denoted by $p_{\mathrm{dep},t\mathbf{1}}$ for \Cref{alg:classical-rejection-sampler} in \Cref{theorem:trace_distance_error}.

Consequently, for every fixed $t<\epsilon^\ast$,
\begin{equation}
p_{\mathrm{fail}}^{\mathrm{BEC}(t)}
=
p_{\mathrm{dep},\mathbf{t}}
\longrightarrow 0
\qquad\text{as }m\to\infty,
\end{equation}
and by \Cref{theorem:trace_distance_error}, when the Bernoulli parameter used in \Cref{alg:classical-rejection-sampler} is $t$ for all coordinates, its error also vanishes:
\begin{equation}
\mathbb{E}_{y}\!\left[
\left\|\rho_{\beta,y}-\sigma_{\beta,y}\right\|_1
\right]
\leq
2p_{\mathrm{dep},\mathbf{t}}
\longrightarrow 0.
\end{equation}
Moreover, the acceptance probability of \Cref{alg:classical-rejection-sampler} is $1-p_{\mathrm{dep},\mathbf{t}}$, so it also runs efficiently in this regime: its expected number of repetitions is $1/(1-p_{\mathrm{dep},\mathbf{t}})=1+o(1)$. 
Since $t=\tanh(\beta)$ for \Cref{alg:classical-rejection-sampler} to sample the Gibbs state at inverse temperature $\beta$, every
\begin{equation}
    \beta<\beta^\ast:=\tanh^{-1}(\epsilon^\ast)
\end{equation}
is achievable, which matches \Cref{corr:hdqiusdthreshold} exactly. 

Let us give an insight into why: Gibbs sampling from commuting Hamiltonians reduces to decoding a PSC. As explained in \Cref{sec:commuting}, applying USD independently to the PSC outputs produces an erased set
$S\leftarrow \operatorname{Ber}(\mathbf t)$ with $t_i=\tanh(\beta v_i)$. Decoding of the resulting BEC fails precisely when the symplectic matrix restricted to $S$ is rank deficient. This is exactly the event on which \Cref{alg:classical-rejection-sampler} rejects, which is why it determines the error made by the algorithm.

\section{Application to non-commuting Hamiltonians}\label{sec:noncommuting}
For commuting Hamiltonians, we showed how QD-HDQI can use a substantially simplified noise state compared with the original HDQI algorithm.
It is therefore natural to ask whether this simplification extends beyond this commuting setting.
Indeed, the elegant structure of QD-HDQI for commuting Hamiltonians relies on substantial structure: the Gibbs noise state factorizes into a simple product state, and the corresponding quantum decoding problem reduces to transmitting the symplectic code through independent binary-input pure-state channels. Neither property need hold for a general non-commuting Hamiltonian.

Perhaps surprisingly, we show that our approach nevertheless extends beyond the commuting setting.
Specifically, we consider Hamiltonians whose anti-commutation graph decomposes into small (logarithmic-sized) connected components.
This is precisely the regime of non-commutativity in which the original HDQI algorithm admitted efficient pilot-state preparation.
We show that a natural choice for the noise function in this non-commuting setting leads to a noise state that factorizes across the connected components, bypassing the elaborate classical preprocessing required by the original HDQI algorithm to construct an MPS representation of the pilot state.
Furthermore, we show that the decoding problem corresponds to transmitting the symplectic code over a sequence of block-independent symmetric pure-state channels.

Finally, we show that for certain Hamiltonians in this regime, QD-HDQI can reach lower temperatures with quantum decoders than achievable with any classical decoder.
We demonstrate this explicitly for dilute semiclassical spin glasses.

\subsection{Intuition for QD-HDQI's noncommuting extension}
Let us group the Pauli terms in the Hamiltonian such that only Paulis within the same group may anticommute. 
To describe this mathematically, consider a partition $\{I_\gamma\}_{\gamma=1,\dots,C}$ of $[m]$ and consider a Pauli Hamiltonian of the form
\begin{equation}
    \cH
    = \sum_{\gamma=1}^C \cH_{\gamma},
    \qquad
    \cH_{\gamma}
    \coloneqq \sum_{j\in I_\gamma}(-1)^{y_j}v_jP_j,
\end{equation}
where $C\geq 1$ and $m_\gamma\coloneqq |I_\gamma|$. Furthermore, we assume that
\begin{enumerate}
    \item every Pauli in $I_\gamma$ commutes with every Pauli in
    $I_\delta$ whenever $\gamma\neq\delta$.
    \item $\max_\gamma m_\gamma$ is small, so typically $\mathcal O(\log m)$ or even $\mathcal O(1)$.
\end{enumerate}
We note that in the language of \cite{schmidhuber2025hamiltoniandecodedquantuminterferometry}, these conditions correspond to requiring that the connected components of the anti-commutation graph are of small size.
The anti-commutation graph of a Pauli Hamiltonian is defined as follows: It contains one vertex $v_i$ for every Pauli term $P_i$, and for $i\neq j$, $v_i$ and $v_j$ are connected by an edge if and only if $P_i$ and $P_j$ anti-commute.
The connected components of this graph correspond to the partition $\{I_\gamma\}_{\gamma=1,\dots,C}$ considered above.
The condition $\max_{\gamma}m_{\gamma} = \mathcal O(\log n)$ ensures efficient pilot-state preparation in the original HDQI algorithm for non-commuting Hamiltonians.

The intuition for why our algorithm still works in this regime is that the clusters $H_\gamma$ of Pauli terms can in some sense behave similarly to a single Pauli term in the commuting case.
That is, we will see that the noise state has a block product form across the partition, and similarly, the noise model corresponds to a block-independent pure-state channel.

We write $e=(e_1,\ldots,e_C)$ and $y=(y_1,\ldots,y_C)$ according to this
partition, and define
\begin{equation}
    Q_{e_\gamma}^{(\gamma)}
    \coloneqq \prod_{j\in I_\gamma}P_j^{e_j},
    \qquad
    Q_e\coloneqq \prod_{j=1}^m P_j^{e_j},
\end{equation}
where each product is taken in the fixed order of the Pauli terms.
By definition it holds that
\begin{equation}
    Q_e=\prod_{\gamma=1}^C Q_{e_\gamma}^{(\gamma)}.
\end{equation}

\subsection{Characterization of the noise function and the involved decoding problem}
We first obtain a normalized Pauli expansion for each block with following helper lemma.
\begin{lemma}
\label{lem:independent_block_expansion}
    Consider a Hamiltonian of the form
    \begin{equation}
        \mathcal K_y
        \coloneqq
        \sum_{j=1}^{r}(-1)^{y_j}v_jP_j,
    \end{equation}
    where $P_1,\ldots,P_r\in\widetilde{\mathcal P}_n$ are arbitrary $n$-qubit Pauli strings, $y\in\ff_2^r$ and
    $v_1,\ldots,v_r\in [0,1]$.
    For $e\in\ff_2^r$, let us define
    \begin{equation}
        Q_e\coloneqq P_1^{e_1}\cdots P_r^{e_r}
        \, .
    \end{equation}
    Then, for every $\beta\geq 0$, there exists a unique function
    $h:\ff_2^r\rightarrow\mathbb C$, independent of $y$, such that $\norm{h}_2=1$ and
    \begin{equation}\label{eq:sign_covariant_block_expansion}
        \forall y\in\ff_2^r: \sqrt{\frac{2^n}{Z_{\beta,y}}}
        \exp\!\left(-\frac{\beta}{2}\mathcal K_y\right)
        =
        \sum_{e\in\ff_2^r}
        h(e)(-1)^{y\cdot e}Q_e
    \end{equation}
    where
    \begin{equation}
        Z_{\beta,y}\coloneqq\tr\!\left[\exp(-\beta\mathcal K_y)\right]
        \, .
    \end{equation}
    In the special case that the symplectic vectors $\symp(P_j)$ are linearly independent, 
    $h(e)$ can be understood as the expansion coefficients of $\sqrt{\frac{2^n}{Z_{\beta,0}}}\exp\!\left(-\frac{\beta}{2}\mathcal K_0\right)$ into the basis $\{Q_e | e\in\ff_2^r\}$ and $Z_{\beta,y}=Z_{\beta,0}$ for all $y\in\ff_2^r$.
\end{lemma}
\begin{proof}
    Write $p_j\coloneqq\symp(P_j)$ and let
    \begin{equation}
        A_y
        \coloneqq
        \sqrt{\frac{2^n}{Z_{\beta,y}}}
        \exp\!\left(-\frac{\beta}{2}\mathcal K_y\right).
    \end{equation}
    Since the algebra generated by $P_1,\ldots,P_r$ is spanned by the
    ordered products $\{Q_e:e\in\ff_2^r\}$, every $A_y$ belongs to
    \begin{equation}
        \mathcal A
        \coloneqq
        \operatorname{span}\{Q_e:e\in\ff_2^r\}.
    \end{equation}
    Denote by $H$ the $2n\times r$ binary matrix whose $j$-th column is $p_j$.
    Therefore, $\cC\coloneqq\ker(H)$ is the symplectic code associated to $\mathcal K_y$.

    Consider the linear map $\Phi:\ff_2^{2n}\rightarrow\ff_2^r$ defined as
    \begin{equation}
        \Phi(x) \coloneqq \begin{pmatrix}
            \langle x,p_1 \rangle_s \\
            \langle x,p_2 \rangle_s \\
            \vdots \\
            \langle x,p_r \rangle_s
        \end{pmatrix}
    \end{equation}
    where
    \begin{equation}
        \langle a,b \rangle_s \coloneqq a\cdot Jb \, ,
        \quad 
        J \coloneqq \begin{pmatrix}
            0 & \mathds{1} \\ \mathds{1} & 0
        \end{pmatrix}
    \end{equation}
    is the standard symplectic inner product of $a,b\in\ff_2^{2n}$.
    In words, if $q$ is the symplectic representation of $Q\in\widetilde{\mathcal P}_n$, then the $i$-th entry of the vector $\Phi(q)$ is zero if $Q$ and $P_i$ commute, and it is $1$ if $Q$ and $P_i$ anti-commute.

    Observe that the image of $\Phi$ is
    \begin{equation}
        \Ima(\Phi) = \Ima(H^TJ) = \Ima(H^T) = \cC^\perp
    \end{equation}
    because $J$ is invertible.
    In particular, this means that for all $t\in\cC^\perp$ there exists a Pauli $R_t\in \widetilde{\mathcal{P}}_n$ such that
    \begin{equation}
        \forall j=1,\dots,r : R_tP_jR_t = (-1)^{t_j}P_j
        \, .
    \end{equation}
    In particular, this implies
    \begin{equation}\label{eq:Qe_conjugation}
        \forall t\in \cC^\perp : R_t Q_e R_t = (-1)^{t\cdot e}Q_e
        \, ,
    \end{equation}
    \begin{equation}
        \forall t\in \cC^\perp : R_t \mathcal K_y R_t = \mathcal K_{y+t}
        \, ,
    \end{equation}
    \begin{align}
        \forall t\in \cC^\perp :
 Z_{\beta,y+t}
  &=\tr\left[e^{-\beta\mathcal K_{y+t}}\right] \\
  &=\tr\left[R_t e^{-\beta\mathcal K_y}R_t\right]\\
  &=\tr\left[e^{-\beta\mathcal K_y}R_t^2\right]
  =Z_{\beta,y}
        \, ,
    \end{align}
    and therefore also
    \begin{equation}\label{eq:code_invariance}
        \forall t\in\cC^\perp : A_{y+t} = R_t A_y R_t
        \, .
    \end{equation}
    
    Let's now consider the Fourier transform  of $A_y$, i.e., we view $A_y$ as a matrix where each entry is a function of $y$, and we take the element-wise Fourier transform of each such function.
    We get
    \begin{equation}\label{eq:operator_ft}
        \hat{A}(e) = \frac{1}{\sqrt{2^r}} \sum_{y\in\ff_2^r}(-1)^{y\cdot e}A_y
        \quad\text{and}\quad
        A_y = \frac{1}{\sqrt{2^r}}\sum_{e\in\ff_2^r} (-1)^{y\cdot e}\hat{A}(e)
        \, .
    \end{equation}

    We shall now show that $\hat{A}(e)$ is proportional to $Q_e$.
    Consider the expansion of $\hat{A}(e)$ into the basis $\{ Q_x \mid x\in\mathcal{X}\}$ of $\mathcal A$, where $\mathcal{X}$ is a set of representatives for each coset in $\ff_2^r/\cC$
    \begin{equation}
        \hat{A}(e) = \sum_{x\in\mathcal X} \lambda_x(e) Q_x
        \, .
    \end{equation}
    Hence, we can see from \cref{eq:Qe_conjugation} that
    \begin{equation}
        \forall t\in\cC^\perp : R_t \hat{A}(e) R_t = \sum_{x\in\mathcal X} \lambda_x(e) (-1)^{t\cdot x} Q_x
    \end{equation}
    However, at the same time \cref{eq:code_invariance} also implies that
    \begin{equation}
        \forall t\in\cC^\perp : R_t\hat{A}(e)R_t = (-1)^{t\cdot e}\hat{A}(e)
        \, .
    \end{equation}
    Note that two characters $(-1)^{t\cdot x}$ and $(-1)^{t\cdot e}$ can only be equal for all $t\in\cC^\perp$ if $x+e\in \cC$.
    By linear independence of the basis operators $Q_x$, we must have $\lambda_x(e)=0$ unless $x+e\in\cC$.
    There is exactly one such representative $x\in\mathcal X$, and $Hx=He$ implies that $Q_x$ differs from $Q_e$ only by a phase.   
    Hence, there must exist a scalar $h(e)\in\mathbb{C}$ such that
    \begin{equation}
        \hat{A}(e) = \sqrt{2^r}h(e) Q_e
    \end{equation}
    And we conclude from \cref{eq:operator_ft} that
    \begin{equation}
        A_y = \sum_{e\in\ff_2^r}(-1)^{y\cdot e}h(e) Q_e
        \, .
    \end{equation}
    It remains to argue that $h$ is a properly normalized function.
    Recall that $\frac{1}{2^n}A_y^2$ is properly normalized, i.e.,
    \begin{equation}
        \forall y\in\ff_2^r : \tr[A_y^\dagger A_y] = 2^n
        \, .
    \end{equation}
    We conclude that
    \begin{align}
        1
        &= \frac{1}{2^n}\frac{1}{2^r}\sum_{y\in\ff_2^r} \tr[A_y^\dagger A_y] \\
        &= \frac{1}{2^n}\frac{1}{2^r}\sum_{y\in\ff_2^r}\sum_{e,e'\in\ff_2^r} \overline{h(e)}h(e') (-1)^{y\cdot (e+e')} \tr[Q_e^{\dagger}Q_{e'}] \\
        &= \frac{1}{2^n}\sum_{e\in\ff_2^r}\abs{h(e)}^2\tr[Q_e^{\dagger}Q_{e}] \\
        &= \sum_{e\in\ff_2^r} \abs{h(e)}^2
        \, .
    \end{align}
\end{proof}

We can apply \cref{lem:independent_block_expansion} to every block $\cH_\gamma$ and denote the
resulting coefficient function by $h_\gamma$.
Using the fact that the blocks $\cH_{\gamma}$ mutually commute, we can now write
\begin{align}
    \sqrt{\frac{2^n}{Z_{\beta,y}}}
    \exp\left(-\frac{\beta}{2}\cH\right)
    &= \sqrt{\frac{2^n}{Z_{\beta,y}}}
    \prod\limits_{\gamma=1}^C
    \exp\left(-\frac{\beta}{2}\cH_\gamma\right) \\
    &= \frac{1}{\sqrt{Z_{\beta,y}'}}
    \prod\limits_{\gamma=1}^C
    \sum_{e_\gamma\in\ff_2^{m_\gamma}}(-1)^{y_\gamma\cdot e_\gamma}h_\gamma(e_\gamma)Q^{(\gamma)}_{e_\gamma} \\
    &= \sum_{e\in\ff_2^m} \tilde f(e)Q_e,
\end{align}
where 
\begin{equation}
    Z_{\beta,y}' \coloneqq
    \frac{(2^n)^{C-1}Z_{\beta,y}} {\displaystyle\prod_{\gamma=1}^C Z_{\beta,y_\gamma}^{(\gamma)}}
\end{equation}
and
\begin{equation}
    \tilde f(e)\coloneqq \frac{1}{\sqrt{Z_{\beta,y}'}}(-1)^{y\cdot e}\prod\limits_{\gamma=1}^C h_\gamma(e_\gamma)
    \, .
\end{equation}
Unfortunately, $\tilde f$ is not necessarily a valid noise function for HDQI, as it may not properly be normalized, i.e., in general one could have $Z_\beta '\neq 1$.
As in the commuting case, our strategy will instead be to use the function 
\begin{equation}
    f(e) \coloneqq (-1)^{y\cdot e}\prod\limits_{\gamma=1}^C h_\gamma(e_\gamma)
\end{equation}
as a noise function in HDQI.
If HDQI succeeds, then the obtained state $\rho_f$ is a good approximation of the desired Gibbs state $\rho_{\tilde f}$ as long as $Z_{\beta}'\approx 1$.
As in the commuting case, we can show that this occurs with probability approaching $1$ (as long as the decoding problem is solvable).
\begin{theorem}[QD-HDQI for non-commuting Hamiltonians]\label{thm:analysis_clustered}
    For $y\in \mathbb{F}_2^m$, consider a non-commuting Hamiltonian $\mathcal H_y=\sum_{i=1}^m(-1)^{y_i}P_i$ with $\max_\gamma m_\gamma = \mathcal O(\log m)$ and some target inverse temperature $\beta\geq 0$.
    Assume we have a polynomial-time classical or quantum decoding algorithm for the associated decoding problem which fails with probability $\pfail$.
    Furthermore, assume $1/(1-\pfail)=\mathcal{O}(\mathrm{poly}(m))$.
    Then, there exists an efficient quantum algorithm which prepares a state $\tilde{\rho}$ which is close to the Gibbs state in expectation over a uniformly random choice of signs
    \begin{equation}
        \mathbb{E}_{Y\sim\mathrm{Uniform}(\ff_2^m)}
        \left[ \norm{\tilde{\rho} - \frac{1}{Z_\beta}\exp\left(-\beta \cH_y\right)}_1 \right]
        \leq 6\sqrt{\pfail}
        \, .
    \end{equation}
\end{theorem}
\begin{proof}
    We need to argue that the prerequisites for \cref{thm:analysis_codeword} are fulfilled, and that the involved quantity $\alpha$ is small.
    
    First, we argue that the noise state $\sum_ef(e)\ket{e}$ can always be prepared efficiently.
    First, we need to (classically) determine the functions $h_\gamma$ for each cluster.
    Then, the circuit complexity of preparing the state
    \begin{equation}
        \bigotimes_{\gamma=1}^C \left( \sum_{e_{\gamma}\in\ff_2^{m_\gamma}} h_\gamma(e_\gamma)\ket{e_\gamma}\right)
    \end{equation}
    is at most $\mathcal{O}(2^{\max_\gamma m_\gamma}C)$.
    The desired noise state is then obtained by simply applying $Z^y$.
    It remains to understand the complexity of classically calculating the functions $h_\gamma$.

    Following the proof of \cref{lem:independent_block_expansion}, recall that $h_\gamma(e)$ is the proportionality factor between $Q_{e_\gamma}^{(\gamma)}$ and $\hat{A}(e_\gamma)$ (up to a $\sqrt{2^{m_\gamma}}$ term).
    That is, it can be obtained by numerically evaluating both of these matrices.
    These operators should be represented symbolically as vectors of Pauli coefficients of size $\leq 2^r$ for $r=m_\gamma$, not as $2^n\times 2^n$ matrices.

    More precisely, the products $Q_{e_\gamma}^{(\gamma)}$ span an algebra of dimension $\leq 2^r$.
    Enumerating these products and removing duplicates up to phase gives a Pauli basis in time $2^{\mathcal O(r)}\operatorname{poly}
  (n)$.

    For each sign string $y_\gamma$, let $M_{y_\gamma}$ be the matrix of left multiplication by $\mathcal K_{y_\gamma}$ in this Pauli basis.
    Computing $e^{-\beta M_{y_\gamma}/2}\mathbf e_I$, where $\mathbf e_I$ represents the identity, gives the coefficient vector of $e^{-\beta\mathcal K_{y_\gamma}/2}$.
    Normalizing this vector gives the vector representation of $A_{y_\gamma}$.
    Next, we compute the Fourier transformation in this basis representation in order to obtain $\hat{A}(e_\gamma)$.
    Since $r=\mathcal O(\log m)$, computing $h_\gamma$ has polynomial complexity.

    Next, observe that due to \Cref{lem:average_normalization} the trace error due to the normalization is given by
    \begin{equation}
        \mathbb{E}_Y\norm{ \rho_{f} - \rho_{\tilde f} }_1 \leq 2\sqrt{\pfail}
    \end{equation}
    where $\pfail$ is the failure probability of the used quantum decoder.
    Finally, \cref{lem:average_alpha_codeword} applies, and, we have $\mathbb{E}_Y\left[ \alpha^2 \right]=\pfail$.

    The trace distance error calculation from the proof of \cref{thm:analysis_commuting} can be entirely adapted.
\end{proof}

\subsection{Involved classical and quantum decoding problems}
We have seen in \cref{thm:analysis_clustered} that HDQI can be utilized to prepare Gibbs states of non-commuting Hamiltonians, as long as the anti-commutation graph has small connected components and the associated quantum decoding problem can be solved with small failure probability.
In this section, we investigate the structure of this quantum decoding problem and show that it inherits a block-independence structure after an appropriate Clifford basis change.
Recall that the quantum decoding problem in this case consists of estimating the codeword  of the symplectic code $c\in\cC$ given the state
\begin{equation}\label{eq:psic_noncommuting}
    \ket{\Psi_c} = \sum_{e\in\ff_2^m} (-1)^{y\cdot e} \eta(Q_e) \prod\limits_{\gamma=1}^C h_\gamma(e_\gamma) \ket{c+e}
    \, .
\end{equation}

\begin{proposition}[Reduction to a block-product quantum decoding problem]\label{prop:block_pure_state_channel}
    There exists an efficiently computable Clifford unitary $U$ such that
    \begin{equation}
        \forall c\in \cC: U\ket{\Psi_c} =
        \tau(c)\bigotimes\limits_{\gamma=1}^C \ket{\phi_{\gamma,c_\gamma}}
    \end{equation}
    where $\ket{\Psi_c}$ is given by \cref{eq:psic_noncommuting}, $\tau(c)\in\{+1,-1,i,-i\}$ is some phase and
    \begin{equation}
        \ket{\phi_{\gamma,c_\gamma}} \coloneqq 
        \sum_{e_\gamma\in\ff_2^{m_\gamma}} h_\gamma'(e_\gamma)\ket{c_\gamma + e_\gamma}
    \end{equation}
    for some functions $h_\gamma':\ff_2^{m_\gamma}\rightarrow\mathbb{C}$ s.t. $\norm{h_\gamma'}=1$.
\end{proposition}
\begin{proof}
    We will decompose our Clifford circuit into three steps $U=U_3U_2U_1$.
    In our first step, we simply apply $U_1=Z^y$ to obtain the state
    \begin{equation}
        \ket{\Psi_c'} = (-1)^{y\cdot c} \sum_{e\in\ff_2^m} \eta(Q_e) \prod\limits_{\gamma=1}^C h_\gamma(e_\gamma) \ket{c+e}
    \, .
    \end{equation}

    For the second step, let us write $P_i=i^{\kappa_i}X^{a_i}Z^{b_i}$.
    Direct multiplication gives
    \begin{equation}
        \eta(Q_x)=i^{q(x)},
        \qquad
        q(x)=\sum_i\kappa_ix_i
        +2\sum_{i<j}(b_i\cdot a_j)x_ix_j\pmod 4
    \end{equation}
    for any $x\in\ff_2^m$.
    We define our second unitary $U_2$ to act as
    \begin{equation}
        U_2\ket{x}\coloneqq\overline{\eta(Q_x)}\ket{x}
    \end{equation}
    and it is a diagonal Clifford, by the standard
    quadratic-form description of Clifford operators.
    To understand the action of $U_2$, let us first define some preliminaries.
    Denote by $H_\gamma$ the binary matrix with columns $\symp(P_i)$ for $i\in I_\gamma$.
    Then, $A_\gamma\coloneqq H_\gamma^TJH_\gamma$ is the anticommutation matrix corresponding to the block $\gamma$.
    Denote by $L_\gamma$ the strictly lower-triangular part of $A_\gamma$, i.e., $A_\gamma=L_\gamma + L_\gamma^T$.
    By definition, it follows that we can write
    \begin{equation}
        Q_u^{(\gamma)}Q_v^{(\gamma)} = (-1)^{u^TL_\gamma v}Q_{u+v}^{(\gamma)}
    \end{equation}
    for any $u,v\in\ff_2^{m_\gamma}$.
    In particular, we obtain
    \begin{equation}
        Q_{c+e}
        = \prod\limits_{\gamma} Q_{c_\gamma+e_\gamma}^{(\gamma)}
        = \prod\limits_{\gamma} Q_{c_\gamma}^{(\gamma)}Q_{e_\gamma}^{(\gamma)} (-1)^{c_\gamma^T L_\gamma e_\gamma}
        = Q_cQ_e \prod\limits_{\gamma} (-1)^{c_\gamma^T L_\gamma e_\gamma}
    \end{equation}
    where in the last step we used the commutativity between blocks.
    Furthermore, we have
    \begin{equation}
        \eta(Q_{c+e})=\eta(Q_c)\eta(Q_e)\prod_\gamma(-1)^{c_\gamma^T L_\gamma e_\gamma}
    \end{equation}
    because $Q_c$ is proportional to the identity.
    We hence get
    \begin{align}
        \ket{\Psi_c''}
        \coloneqq U_2\ket{\Psi_c'}
        &= (-1)^{y\cdot c} \sum_{e\in\ff_2^m} \eta(Q_e) \prod\limits_{\gamma=1}^C h_\gamma(e_\gamma) \overline{\eta(Q_{c+e})}\ket{c+e} \\
        &= (-1)^{y\cdot c} \overline{\eta(Q_c)}\sum_{e\in\ff_2^m} \prod\limits_{\gamma=1}^C h_\gamma(e_\gamma) (-1)^{c_\gamma^T L_\gamma e_\gamma}\ket{c+e} \\
        &= (-1)^{y\cdot c}\overline{\eta(Q_c)} \bigotimes_{\gamma=1}^C X^{c_\gamma}Z^{L_\gamma^T c_\gamma} \ket{h_\gamma}
    \end{align}
    where $\ket{h_\gamma}\coloneqq\sum_{e_\gamma}h_\gamma(e_\gamma)\ket{e_\gamma}$.

    We have now achieved a state in block-product form, however, there are still unwanted $c_\gamma$-dependent phases within each block caused by the $Z^{L_\gamma^T c_\gamma}$ operator.
    In our third step, we will get rid of this phase by finding a Clifford operator $V_\gamma$ for each block which fulfills
    \begin{equation}\label{eq:Vgamma_action}
        V_\gamma X^{c_\gamma}Z^{L_\gamma^T c_\gamma} V_\gamma^\dagger = \alpha_\gamma(c_\gamma)X^{c_\gamma}
    \end{equation}
    for some phase $\alpha_\gamma(c_\gamma)\in\{1,-1,i,-i\}$.
    Picking $U_3=\otimes_\gamma V_\gamma$ leads us to the desired final state
    \begin{align}
        \ket{\Psi_c'''}
        &\coloneqq U_3\ket{\Psi_c''} \\
        &= (-1)^{y\cdot c}\overline{\eta(Q_c)}\prod_\gamma \alpha_\gamma(c_\gamma) \bigotimes_{\gamma=1}^C X^{c_\gamma} V_\gamma \ket{h_\gamma} \\
        &= (-1)^{y\cdot c}\overline{\eta(Q_c)}\prod_\gamma \alpha_\gamma(c_\gamma) \bigotimes_{\gamma=1}^C X^{c_\gamma} \ket{h_\gamma'}
    \end{align}
    where $\ket{h_\gamma'}\coloneqq\sum_{e_\gamma}h_\gamma'(e_\gamma)\ket{e_\gamma}$ and $h_\gamma'$ is defined through the equation $V_\gamma\ket{h_\gamma} = \ket{h_\gamma'}$.

    It remains to construct Clifford unitaries that act as \cref{eq:Vgamma_action}.
    For convenience, let us define the Paulis $\mathcal P_\gamma(u)\coloneqq X^uZ^{L_\gamma^Tu}$ and $\mathcal Q_\gamma(u)\coloneqq X^{u}$ for $u\in\ff_2^{m_\gamma}$.
    Our goal is to find $V_\gamma$ such that
    \begin{equation}
        V_\gamma \mathcal P_\gamma(u) V_\gamma^\dagger = \alpha_\gamma(u) \mathcal Q_\gamma(u)
    \end{equation}
    for all $x\in D_\gamma \coloneqq \{c_\gamma | c\in\cC \}$.
    $D_\gamma$ is clearly a subspace of $\ff_2^{m_\gamma}$ and it fulfills $D_\gamma\subset \ker(A_\gamma)$ because for any $c\in\cC$
    \begin{equation}
        A_\gamma c_\gamma = H_\gamma^TJH_\gamma c_\gamma = -\sum_{\delta\neq\gamma} H_\gamma^TJH_\delta c_\delta = 0
    \end{equation}
    where we used $\sum_\gamma H_\gamma c_\gamma=Hc=0$ and $H_\gamma^TJH_\delta=0$ for $\gamma\neq\delta$ due to the commutativity between blocks.

    Clearly, $\mathcal Q_\gamma(x)$ and $\mathcal Q_\gamma(x')$ commute for all $x,x'\in D_\gamma$.
    We now show that $\mathcal P_\gamma(x)$ and $\mathcal P_\gamma(x')$ commute as well:
    \begin{equation}
        \langle \mathcal P_\gamma(x), \mathcal P_\gamma(x') \rangle_{\mathrm s}
        = x \cdot L_\gamma^T x' + x' \cdot L_\gamma^T x = x^TA_\gamma x' = 0
    \end{equation}
    because $x\in\ker(A_\gamma)$.

    The map $(x,L_\gamma^T x)\mapsto(x,0)$, for $x\in D_\gamma$, is therefore an isomorphism between isotropic subspaces and therefore extends to a symplectic transformation.
    Symplectic Gaussian elimination and Clifford synthesis yield an efficiently computable Clifford $V_\gamma$ implementing this transformation, so \cref{eq:Vgamma_action} holds with phases $\alpha_\gamma(x)\in\{1,-1,i,-i\}$.
\end{proof}

In summary, under the Clifford basis transformation, the quantum decoding problem takes a block-product form.
More concretely, define the block pure-state channel
\begin{equation}
    \mathsf W_\gamma:
    \ff_2^{m_\gamma}\ni x \mapsto \sum_{e_{\gamma}\in\ff_2^{m_\gamma}} h_\gamma'(e_\gamma) \ket{x + e_\gamma}
    \, .
\end{equation}
Then the quantum decoding problem is precisely that of recovering
$c\in\cC$ after transmission through
\begin{equation}
    \mathsf W_1\otimes\mathsf W_2\otimes\cdots\otimes\mathsf W_C.
\end{equation}
The channel is independent across blocks, although it need not be identical
from one block to another and it may contain arbitrary correlations within
each block. 
Since $U$ is a known Clifford circuit independent of $c$, a decoder for this product block channel can be composed with this basis change to obtain a decoder for the original quantum decoding problem that occurs in our algorithm, without changing its success probability.

\subsection{Example: Dilute classical spin glasses with non-commuting defects}
\label{sec:defected_xorsat}
We now apply our construction for non-commuting Hamiltonians to random dilute spin glasses with non-commuting defects.
Let
\begin{equation}
    A\in\ff_2^{(1-R)m\times m},
    \qquad
    R=1-\frac{\ell}{r},
\end{equation}
be drawn from the standard binary $(\ell,r)$-regular LDPC ensemble (as described in \cite{richardson_modern_2008}), where $3\leq \ell<r$ are fixed.
Fix some constant probability $q\in [0,1/2]$ and independently mark each $i\in[m]$ with probability $q$.
Let $D\subseteq[m]$ denote the set of marked indices.
Finally, choose $y\in\ff_2^m$ uniformly.
We consider the Hamiltonian of the form
\begin{equation}
    \cH_{A,D,y} \coloneqq \sum_{i=1}^m(-1)^{y_i}P_i
    \, .
    \qquad
    P_i
    \coloneqq
    \begin{cases}
        X^{a_i} & \text{if }i\in D,\\
        Z^{a_i} & \text{if }i\notin D
    \end{cases}
    \label{eq:defected_xorsat_hamiltonian}
\end{equation}
where $a_i$ denotes the $i$th column of $A$.
We will refer to the coordinates in $D$ as the \emph{flipped} bits.
Notice that we take coupling magnitudes to be one (i.e. $v_i=1$).
Our formalism can also easily capture the case where the couplings are non-uniform, but we do not pursue this generalization here for simplicity.

Let $N\coloneqq[m]\setminus D$.
For notational simplicity, we can re-order the Pauli terms such that the coordinates in $D$ occur before those in $N$.
Under reordering, the symplectic matrix of the Hamiltonian is
\begin{equation}
    H = \begin{pmatrix}
        A_D&0\\
        0&A_N
    \end{pmatrix}
\end{equation}
where $A_S$ denotes the submatrix of $A$ formed by the columns in $S$, and correspondingly, $A_N$ is the submatrix formed by the columns in $N$.
We conclude that the symplectic code $\cC$ fulfills
\begin{equation}
    \cC \coloneqq \ker(H) = \ker(A_D)\oplus\ker(A_N)
    \, .
\end{equation}

In this section, we will apply our algorithm to prepare a Gibbs state of $\cH_{A,D,y}$ at an inverse temperature $\beta$ that is strictly larger than what is achievable by classical decoders, at least when $q$ is small enough (but still asymptotically constant).
It was shown in \cite[Theorem 5]{schmidhuber2025hamiltoniandecodedquantuminterferometry} that the connected components of the anti-commutation graph have size $\mathcal{O}(\log m)$ with high probability when $q < \tfrac{1}{\ell^2(r-1)^2}$, and therefore, the noise state can be prepared efficiently.
To apply \cref{thm:analysis_clustered}, it remains to construct an efficient quantum decoder that solves the associated decoding task.

\subsubsection{Simplifying the decoding problem}
Let us start by looking at the quantum decoding task associated to the spin glass with non-commuting defects.
Let us group the codeword bits $i\in[m]$ into two groups: \emph{cluster bits} and \emph{non-cluster bits}.
The set of cluster bits is defined as 
\begin{equation}
    K \coloneqq D\cup\{i\in[m]:i\text{ is adjacent to some }j\in D\}
    \, ,
\end{equation}
where two variables are said to be adjacent if they share a common check node.
Correspondingly, the set of non-cluster bits is $U\coloneqq[m]\setminus K$.
Clearly, $K$ contains all bits that are part of a connected component of the anti-commutation graph of size larger than 1.
Put differently: no anti-commutation edge is incident on $U$.
Therefore, every coordinate in $U$ is a singleton component of the anti-commutation graph.
For a singleton component, the block channel in \Cref{prop:block_pure_state_channel} is unitarily equivalent to $\psc[p]$ for $p = \tfrac{1}{2}(1-\sech\beta)$.

In comparison, the coordinates in $K$ may be part of a more complicated block channel, which can be determined using \cref{prop:block_pure_state_channel}.
Determining these exact channels analytically is quite tedious given that they depend on the exact size and topology of the involved connected components.
In order to keep the analysis as simple as possible, we will avoid determining these channels exactly.
In fact, for the purpose of showing our desired information-theoretic advantage, it suffices to follow a very simple decoding strategy: The decoder simply ignores the information stored in the clustered coordinates, and instead, treats them as erasures.
Using this simplification, the effective noise seen by the decoder is
\begin{equation}\label{eq:effective_noise}
    \bigotimes_{i\in U}\psc[p]
    \otimes
    \bigotimes_{i\in K}\psc[1/2].
\end{equation}
Let us highlight again that this decoding strategy is strictly sub-optimal, as the information stored in the clustered bits is completely lost.
A more refined strategy would be to apply a block-USD on each connected component of the anti-commutation graph, resulting in block-correlated erasure channels\footnote{In fact, the information from such block-erasure channels would be trivial to integrate with the BPQM decoder we will utilize further below.}.
We leave the analysis of more refined schemes like this for future work.
For us, the more basic decoding strategy will suffice, because as long as $q$ is small enough, the fraction of bits in $K$ remains small as well, and therefore the amount of information lost remains manageable.
The relation between the size of $U$ and $q$ is captured by following lemma.

Let us define
\begin{equation}
    d_0\coloneqq 1+\ell(r-1)
    \, ,
    \qquad
    F_q\coloneqq(1-q)^{d_0}
    \, .
\end{equation}
\begin{lemma}[Concentration of the non-cluster fraction]\label{lem:defected_noncluster_concentration}
    Consider a random $(\ell,r)$-regular LDPC code and let $F_m\coloneqq |U|/m$.
    For every fixed $\epsilon>0$,
    \begin{equation}
        \Pr\!\left[|F_m-F_q|\geq\epsilon\right]
        \leq
        2\exp\!\left(-\frac{m\epsilon^2}{2d_0^2}\right)
        +\mathcal O\!\left(\frac{1}{m\epsilon}\right),
        \label{eq:defected_noncluster_concentration}
    \end{equation}
    where the implicit constant depends only on $\ell$ and $r$.
    In particular, $F_m\to F_q$ in probability and in expectation.
\end{lemma}

\begin{proof}
    Condition on the Tanner graph $A$.
    For a variable node $v$, let
    \begin{equation}
        S_v \coloneqq \{v\}\cup\{u:u\text{ shares a check node with }v\}
    \end{equation}
    be its closed variable-node neighborhood.
    The node $v$ belongs to $U$ if and only if no node in $S_v$ is marked, and hence
    \begin{equation}
        \Pr[v\in U\mid A]=(1-q)^{|S_v|}.
    \end{equation}

    Let's group the variable nodes into two groups: ``good'' nodes and ``bad'' nodes.
    A node $v$ is defined to be bad if the $\ell(r-1)$ positions reached from it are not distinct, otherwise it's called good.
    For a good node, we have $\mathbb{E}[\mathds{1}_{v\in U}]=F_q$.
    Hence, we conclude that
    \begin{equation}
        \left|\mathbb E[F_m\mid A]-F_q\right|
        \leq \frac{B_m}{m}.
    \end{equation}
    where $B_m$ denotes the number of bad nodes.
    In the standard LDPC code ensemble, a fixed variable node is locally bad with probability $\mathcal O(m^{-1})$ \cite{richardson_modern_2008}.

    Since the degrees are fixed, $\mathbb E[B_m]=\mathcal O(1)$, and Markov's inequality gives
    \begin{equation}
        \Pr\!\left[
           \left|\mathbb E[F_m\mid A]-F_q\right|\geq\frac\epsilon2
        \right]
        =\mathcal O\!\left(\frac{1}{m\epsilon}\right).
        \label{eq:conditional_fraction_markov}
    \end{equation}

    For a fixed Tanner graph, changing one defect indicator changes the cluster
    status of at most $d_0$ variable nodes and hence changes $F_m$ by at most
    $d_0/m$.
    McDiarmid's inequality consequently yields
    \begin{equation}
        \Pr\!\left[
           |F_m-\mathbb E[F_m\mid A]|\geq\frac\epsilon2
           \,\middle|\,A
        \right]
        \leq
        2\exp\!\left(-\frac{m\epsilon^2}{2d_0^2}\right).
    \end{equation}
    This bound is uniform in $A$.
    Combining it with \cref{eq:conditional_fraction_markov} proves \cref{eq:defected_noncluster_concentration}.
    Convergence in expectation follows from convergence in probability because $0\leq F_m\leq1$.
\end{proof}

Finally, we make a second simplification to our decoding problem.
Since
\begin{equation}
    \cC=\ker(A_D)\oplus\ker(A_N)\subset\ker(A)=:\cC',
\end{equation}
we decode using the regular Tanner graph of the original LDPC code $\cC'$, ignoring the additional parity constraints defining $\cC$.
This makes the analysis more tractable, but may sacrifice decoding performance.
For our purposes, this will suffice.

\subsubsection{BPQM density evolution on the simplified noise model}
To decode the resulting simplified quantum decoding problem, we again utilize the BPQM decoder.
The analysis of the performance of BPQM is unfortunately more involved than in the classical spin glass case, because the noise model is not i.i.d. anymore, and is instead described by \cref{eq:effective_noise}.
In this section, we instead show how standard density evolution can be adapted in this correlated noise setting.
We will assume a certain familiarity with BPQM density evolution for the rest of the section.
A full introduction to density evolution goes beyond the scope of this paper, and we refer the interested reader to \cite{mandal_bpqmblock_2026}.

We denote the check-node and variable-node convolutions of two binary-input classical-quantum channels $W_1,W_2$ by 
\begin{equation}
    \left[ W_1 \vnop W_2  \right](x) \coloneqq W_1(x) \otimes W_2(x)
    \, ,
\end{equation}
\begin{equation}
    \left[ W_1 \cnop W_2  \right](x) \coloneqq \frac{1}{2}\left( W_1(x) \otimes W_2(0) + W_1(1-x) \otimes W_2(1) \right)
    \, .
\end{equation}
We adopt following convention for the density evolution: We track the combined messages of all $(\ell-1)$ incoming check-to-variable messages arriving at a root variable node, but without the information of the physical channel at that root variable.
For an ordinary i.i.d. setting with physical channel $W$, the density evolution recursion in this convention reads
\begin{equation}
    W_{t+1} = \left(
      (W\vnop W_t)^{\cnop(r-1)}
    \right)^{\vnop(\ell-1)}
    \, , \qquad
    W_0=\psc[1/2]
    \, .
    \label{eq:ordinary_extrinsic_de}
\end{equation}

The erasure pattern in our problem is not i.i.d.
Nevertheless, it has a
finite-state description on a computation tree. 
Instead of only keeping track of the intermediate heralded pure-state channel after $t$ steps, we now have to keep track of an additional classical label that can take one of the values $A$, $B$ or $C$.
The label has the following meaning:
\begin{description}
    \item[Case A:] the root variable and all $(\ell-1)(r-1)$ immediately
    descending variables are unflipped;
    \item[Case B:] the root is unflipped, but at least one of these immediate
    descendants is flipped;
    \item[Case C:] the root is flipped.
\end{description}
Importantly, cases A and B concern only this one-shell neighborhood, not all descendants of the computation tree.

For notational convenience, let us define
\begin{equation}
    \mathsf P\coloneqq\psc[p]
    \, , \qquad
    \mathsf E\coloneqq\psc[1/2]
    \, ,\qquad
    \Delta\coloneqq(\ell-1)(r-1)
    \, .
\end{equation}

We model the additional classical information by adding the label as a special classical heralding system on an output register $L$ of the channel $W_t$, i.e.
\begin{equation}
    W_t = \pi_{t,A} W_{t,A}\otimes\proj{A}_L + \pi_{t,B} W_{t,B}\otimes\proj{B}_L + \pi_{t,C} W_{t,C}\otimes\proj{C}_L
\end{equation}
where $\pi_t\coloneqq (\pi_{t,A},\pi_{t,B},\pi_{t,C})$ is a distribution and $W_{t,\lambda}$ is the heralded pure-state channel conditioned on the label $\lambda$.
Initially, we choose 
\begin{equation}
    W_0 \coloneqq \mathsf E\otimes \left( \pi_{0,A}\proj{A}_L + \pi_{0,B}\proj{B}_L + \pi_{0,C}\proj{C}_L \right)
\end{equation}
where
\begin{equation}
\begin{split}
    \pi_{0,A}&=(1-q)^{1+\Delta} \, ,\\
    \pi_{0,B}&=(1-q)\left[1-(1-q)^\Delta\right] \, ,\\
    \pi_{0,C}&=q \, .
\end{split}
\end{equation}
We now describe the density evolution recursion that defines $W_{t+1}$ from $W_t$.
First, we take $(\ell-1)(r-1)$ copies of $W_t$ for $j\in[\ell-1]$ and $k\in[r-1]$ to describe the message from the previous step.
We write the product channel as
\begin{equation}
    W_t^{\otimes (\ell-1)(r-1)} = 
    \sum_{\vec{\lambda}\in\{A,B,C\}^{(\ell-1)(r-1)} } p_{\vec{\lambda}} \bigotimes_{k=1}^{r-1} \left( \bigotimes_{j=1}^{\ell-1} W_{t,\lambda_{jk}} \otimes \proj{\lambda_{jk}}  \right) 
\end{equation}
where we grouped the labels into the string $\vec{\lambda}=(\lambda_{jk})_{j,k}$ and $p_{\vec{\lambda}} = \prod_{j,k}\pi_{t,\lambda_{jk}}$
It will also be convenient to define the substring $\vec{\lambda}_k\coloneqq (\lambda_{jk})_{j=1,\dots,\ell-1}$.
Finally, we denote by the binary random variable $Z$ the indicator whether the root node variable is flipped or not.

We define the variable-to-check message on branch $(j,k)$, conditioned on the label $\vec{\lambda}$ and the indicator $Z$, as
\begin{equation}
    M_{t,jk}^{\vec{\lambda},Z} \coloneqq \begin{cases}
        W_{t,\lambda_{jk}}\vnop \mathsf P & \text{if } Z=0 \text{ and } \lambda_{jk}=A \text{ and $\lambda_{j,k'}\neq C$ \text{  for each } k' } \\
        W_{t,\lambda_{jk}} & \text{otherwise}
    \end{cases}
    \, .
\end{equation}
This definition reflects that the physical channel at this child variable is non-erased precisely when the child is not flipped in its other one-shell neighborhood, none of its siblings is flipped, and its parent is not flipped.
Next, the check-to-variable messages are
\begin{equation}
    C_{t,j}^{\vec{\lambda},Z} \coloneqq
    M_{t,j1}^{\vec{\lambda},Z}\cnop\cdots\cnop
    M_{t,j(r-1)}^{\vec{\lambda},Z}
\end{equation}
and hence, the final message is
\begin{align}
    W_{t+1} =
    \sum_{\vec{\lambda}\in\{A,B,C\}^{(\ell-1)(r-1)} } p_{\vec{\lambda}}
    \sum_{Z=0,1}(1-q)^{1-Z}q^{Z}
    \left( C_{t,1}^{\vec{\lambda},Z}\vnop\cdots\vnop C_{t,\ell-1}^{\vec{\lambda},Z} \right) \otimes \proj{Z,\vec{\lambda}} \otimes \proj{\lambda'(\vec{\lambda},Z)}_L
\end{align}
where the new label $\lambda'(\vec{\lambda},Z)$ is defined as
\begin{equation}
    \lambda'(\vec{\lambda},Z) \coloneqq \begin{cases}
        A & \text{if } Z=0 \text{ and } C\notin \vec{\lambda} \\
        B & \text{if } Z=0 \text{ and } C\in \vec{\lambda} \\
        C & \text{if } Z=1
    \end{cases}
\end{equation}
It is simple to verify that the label distribution $\pi_{t,\lambda}$ remains unchanged during the density evolution recursion, i.e., it is independent of $t$.
To reflect that, we write $\pi_{t,\lambda}\equiv \pi_{\lambda}$.

For a binary-input CQ channel $W$, we define $F(W)$ to be the fidelity between the two quantum states $W(0)$ and $W(1)$.
We define $F_t(p,q)$ as $F(W_t)$ where we explicitly highlight the dependence on $p$ and $q$ in the definition.
The Fuchs-van de Graaf inequality implies that the probability of erroneously guessing the input of $W_t$ from its output is upper bounded by $F_t(p,q)/2$.
We therefore define the marked density-evolution threshold by
\begin{equation}
    p_{\mathrm{DE}}(q)
    \coloneqq
    \sup\left\{
       p\in[0,1/2]:\lim_{t\to\infty}F_t(p,q)=0
    \right\}.
    \label{eq:defected_de_threshold}
\end{equation}

This threshold concerns the extrinsic decoder just defined.
A standard root-bit decoder can additionally use the physical channel at the root and can therefore only perform better.

The next lemma upgrades convergence in \cref{eq:defected_de_threshold} to a double-exponential decay, which will prove necessary for the block-error argument later.
\begin{lemma}[Double-exponential decay]
\label{lem:defected_double_exponential}
If $\ell\geq3$ and $\lim\limits_{t\rightarrow\infty}F_t(p,q) = 0$, then there exist constants
$A_0,b_0>0$ and $t_0\geq0$ such that
\begin{equation}
    \forall s\geq 0: F_{t_0+s}(p,q)
    \leq
    A_0\exp\!\left[-b_0(\ell-1)^s\right]
    \label{eq:defected_double_exponential}
\end{equation}
\end{lemma}

\begin{proof}
    Replace every occurrence of the physical channel $\mathsf P$ by the erasure channel $\mathsf E$.
    This is a degradation and hence cannot decrease fidelity.
    In summary, we obtain
    \begin{equation}
        F(W_{t+1})
        \leq F\left(\left(\mathsf E\vnop W_t^{\cnop(r-1)}\right)^{\vnop(\ell-1)}\right)
        = F\left(\left(W_t^{\cnop(r-1)}\right)^{\vnop(\ell-1)}\right)
        \, .
    \end{equation}
    Applying the check-node and variable-node fidelity bounds of
    \cite[Lemma~25]{mandal_bpqmblock_2026} gives
    \begin{equation}
        F_{t+1}(p,q)
        \leq
        \left((1+F_t(p,q))^{r-1}-1\right)^{\ell-1}
    \end{equation}
    for every $p,q$ and every $t\geq0$.

As $x\to0$,
\begin{equation}
    \left((1+x)^{r-1}-1\right)^{\ell-1}
    =
    (r-1)^{\ell-1}x^{\ell-1}+\mathcal O(x^\ell).
\end{equation}
Thus there are $C<\infty$ and $\epsilon>0$ such that
\begin{equation}
    F_{t+1}\leq C F_t^{\ell-1}
    \qquad\text{whenever }F_t\leq\epsilon.
\end{equation}
If $F_t\to0$, choose $t_0$ so that
$C^{1/(\ell-2)}F_{t_0}<1$ and $F_{t_0}\leq\epsilon$.  Iteration yields
\begin{equation}
    F_{t_0+s}
    \leq
    C^{-1/(\ell-2)}
    \left(C^{1/(\ell-2)}F_{t_0}\right)^{(\ell-1)^s}
    \, .
\end{equation}
\end{proof}

\subsubsection{From tree error to block error}
The marked density evolution above is a tree calculation that only characterizes the bit error rate for codeword bits with a tree environment.
Here, we elevate that result to a guarantee on the block error rate by adopting the two-stage construction of \cite{mandal_bpqmblock_2026} 

For a fixed parity check matrix $A$, let's denote the set of ``bad bits'' $B_t$ to be the codeword bits for which the depth-$(t+1)$ neighborhood in the Tanner graph contains a cycle.
Conversely, the remaining bits $G_t\coloneqq [m]\setminus B_t$ are called ``good bits''.

In a first step, we sequentially run depth-$t$ BPQM to decode every bit in $G_t$.
After averaging over the independent defect marks, each good bit has a bit-error probability of at most $F_t(p,q) / 2$.
In a second step, we recover the remaining bad bits, based on the assumption that the good bits were correctly decoded.
This is essentially an erasure decoding problem.

By choosing $t=\lfloor \tau \log \log m\rfloor $, for a large enough $\tau$, the double exponential decay of $F_t(p,q)$ remains strong enough to apply Gao's quantum union bound \cite{gao_unionbound_2015} and guarantee that all good bits are correctly decoded with probability $1-o(1)$.
But at the same time, we can guarantee that the fraction of bad bits is asymptotically vanishing (see \cite[Lemma 29]{mandal_bpqmblock_2026}).
It is known that random regular LDPC codes have minimum distance linear in the blocklength with probability $1-o(1)$ (see e.g. \cite{yang_weightdistrib_2011}), and therefore, the second stage of the decoder is guaranteed to successfully recover the bad bits with probability $1-o(1)$.

In summary, the following result holds
\begin{theorem}[Vanishing block error rates, adapted from {{\cite[Theorem 30]{mandal_bpqmblock_2026}}}]
\label{thm:defected_bpqm_block_error}
    Consider a random $(\ell,r)$-regular LDPC code transmitted over the channel in \cref{eq:effective_noise} for a random choice of flipped bits $D$.
    Suppose $\ell\geq3$ and $p<p_{\mathrm{DE}}(q)$.    
    Then the two-stage BPQM decoder achieves a vanishing ensemble-average block-error probability.
\end{theorem}
Once the double exponential decay of $F_t$ is established, the proof is entirely identical to \cite[Theorem 30]{mandal_bpqmblock_2026}.
For this reason, we do not reproduce the proof here.

We briefly note that the error probability of our two-stage decoder is independent of the transmitted codeword (because BPQM error probability on a tree is independent of the codeword).
Consequently, our decoding error guarantee for $\cC'=\ker(A)$ also applies to codewords in $\cC\subset \cC'$.

\subsubsection{An information-theoretic converse for classical decoding}
We now derive a converse upper bound on the maximal noise rate at which any classical decoder can solve the associated quantum decoding problem in QD-HDQI.
In the proof of this converse, we make the crude assumption that all cluster bits are transmitted perfectly.
This can clearly only help the classical decoder, and as such, provides a valid upper bound for the maximal achievable $p$.
A tighter analysis would be technically more challenging, and we leave that to future work.
Our current analysis suffices for the purpose of proving an information theoretic separation between classical and quantum decoders.
\begin{lemma}[Classical converse]\label{lem:defected_classical_converse}
    Suppose that a family of classical decoders for $\cC_{\mathrm{symp}}$ has an ensemble-average block-error probability that is asymptotically vanishing.
    Then,
    \begin{equation}\label{eq:classical_converse}
        R\leq1-F_qh_2(p).
    \end{equation}
    When $F_q>1-R$, reliable classical decoding therefore requires
    \begin{equation}
        p\leq p_{\mathrm{cl}}(q)
        \coloneqq
        h_2^{-1}\!\left(\frac{1-R}{F_q}\right),
        \label{eq:defected_classical_threshold_bound}
    \end{equation}
    where $h_2^{-1}$ denotes the inverse on $[0,1/2]$.
\end{lemma}

\begin{proof}
    Let $\Theta=(A,D)$ denote the random Tanner graph and choice of flipped bits.
    Condition on $\Theta$, and let $C=(C_K,C_U)$ be uniformly distributed over $\cC$.
    Assume $C_K$ is transmitted perfectly without noise to the decoder, whereas the remaining bits $C_U$ are transmitted over $\bsc[p]$ yielding $Y_U\coloneqq\bsc[p]^{\otimes \abs{U}}(C_U)$.
    This assumption of perfect transmission of $C_K$ can only decrease the error probability.

    Recall that
    \begin{equation}
        0 = A_NC_N = A_UC_U + A_{K\cap N}C_{K\cap N}
        \, .
    \end{equation}
    Therefore, from the decoder's perspective, the task is to decode a uniformly random element from a coset of $\ker(A_U)$ transmitted through $\bsc[p]$.
    This scheme has an encoding rate of 
    \begin{equation}
        R_U
        = \frac{\mathrm{dim}(\ker(A_U))}{|U|}
        \geq \frac{|U|-(1-R)m}{|U|}
        = 1-\frac{(1-R)m}{|U|}
      \, .
    \end{equation}
    
    By \cref{lem:defected_noncluster_concentration}, $|U|/m\to F_q$ in probability.
    If $F_qh_2(p)>1-R$, then with probability $1-o(1)$, $R_U$ exceeds the BSC capacity $1-h_2(p)$ by a fixed positive amount.
    The strong converse of the channel coding theorem then implies that the ensemble-average block-error probability must tend
  to one, contradicting the hypothesis.
    Hence $F_qh_2(p)\leq1-R$.
  \end{proof}

\subsubsection{Information theoretic advantage of BPQM over classical decoder}
We return to the quantum decoding task involved in preparing Gibbs states of spin-glasses with non-commuting defects using QD-HDQI.
We have seen in \cref{thm:defected_bpqm_block_error} that BPQM can solve this decoding task whenever $p<p_{\mathrm{DE}}(q)$.
Conversely, we have shown in \cref{lem:defected_classical_converse} that any classical decoder cannot solve the task if $p>p_{\mathrm{cl}}(q)$.

Combining these two results, we can now establish an information theoretic gap for small enough $q$ by numerically showing that $p_{\mathrm{DE}}(q) > p_{\mathrm{cl}}(q)$.
\Cref{fig:defected_threshold_comparison_intro} depicts these two quantities for the $(3,6)$ ensemble, clearly showing the desired information theoretic gap.
The quantity $p_{\mathrm{DE}}(q)$ is estimated numerically using a modification of the population Monte Carlo method of \cite{brandsen_bpqm_2022}.
To incorporate the required labeling, each sample in the population is represented by a pair of a log-likelihood ratio together with a label in $\{A,B,C\}$.

We briefly recall that for QD-HDQI to be efficient in the first place, it was necessary for $q<\tfrac{1}{\ell^2(r-1)^2}$ to guarantee that the clusters are of at most logarithmic size.

\newpage
\appendix

\bibliographystyle{plain}   
\bibliography{refs}   

\end{document}